\documentclass[reprint, superscriptaddress,groupedaddress, amsmath, amssymb, aps, floatfix]{revtex4-2}

\usepackage{graphicx}
\usepackage{dcolumn}
\usepackage{bm}

\usepackage{xcolor}
\usepackage{amssymb}
\usepackage{amsmath}
\usepackage{amsthm}
\usepackage{mathtools}
\usepackage{changepage}
\usepackage{color}
\usepackage{upgreek}
\usepackage{dsfont}
\usepackage{hyperref}
\usepackage{qcircuit}

\usepackage{float}
\usepackage{mathrsfs}
\usepackage{array}
\usepackage{footnote}

\usepackage{tikz}

\usetikzlibrary{matrix}
\usetikzlibrary{shapes,arrows}
\usepackage{enumitem}

\usepackage{comment}

\usepackage[algoruled]{algorithm2e}

\newcommand{\lrb}[1]{\left( #1 \right)}
\newcommand{\lrbb}[1]{\left[ #1 \right]}

\newcommand{\sinp}[1]{\sin\left(#1\right)}
\newcommand{\cosp}[1]{\cos\left(#1\right)}
\newcommand{\sinpt}[1]{\sin^2\left(#1\right)}

\newcommand{\arccosp}[1]{\arccos\left(#1\right)}
\newcommand{\arcsinp}[1]{\arcsin\left(#1\right)}
\renewcommand{\arg}[1]{\mathrm{Arg}\left(#1\right)}
\newcommand{\sgn}[1]{\mathrm{sgn}\left(#1\right)}

\newcommand{\rep}[1]{\mathrm{Re}\left( #1 \right)}
\newcommand{\imp}[1]{\mathrm{Im}\left( #1 \right)}

\newcommand{\identity}{I}
\newcommand{\ket}[1]{\left| #1 \right\rangle}
\newcommand{\bra}[1]{\left \langle #1 \right |}
\newcommand{\ketbra}[1]{\left| #1 \right\rangle\left \langle #1 \right |}

\def\e{\ensuremath{\mathrm{e}}}
\def\i{\ensuremath{\mathrm{i}}}
\def\d{\ensuremath{\mathrm{d}}}

\DeclareMathAlphabet{\mathsfit}{T1}{\sfdefault}{\mddefault}{\sldefault}
\SetMathAlphabet{\mathsfit}{bold}{T1}{\sfdefault}{\bfdefault}{\sldefault}

\def\root#1#2{\mathrm{root}_{#1}(#2)}

\theoremstyle{plain}
\newtheorem{thm}{Theorem}
\newtheorem{lemma}[thm]{Lemma}
\newtheorem{proposition}[thm]{Proposition}

\DeclareMathOperator*{\argmax}{arg\,max}
\DeclareMathOperator*{\argmin}{arg\,min}

\begin{document}

\title{Minimizing estimation runtime on noisy quantum
computers}
\author{Guoming Wang}
\email{guoming.wang@zapatacomputing.com}
\affiliation{Zapata Computing Inc., 100 Federal Street, Boston, MA 02110, USA}
\author{Dax Enshan Koh}
\email[Current address: Institute of High Performance Computing, Agency for Science, Technology and Research (A*STAR), 1 Fusionopolis Way, \#16-16 Connexis, Singapore 138632, Singapore. ]{dax\_koh@ihpc.a-star.edu.sg}
\affiliation{Zapata Computing Inc., 100 Federal Street, Boston, MA 02110, USA}
\author{Peter D.\ Johnson}
\email{peter@zapatacomputing.com}
\affiliation{Zapata Computing Inc., 100 Federal Street, Boston, MA 02110, USA}
\author{Yudong Cao}
\email{yudong@zapatacomputing.com}
\affiliation{Zapata Computing Inc., 100 Federal Street, Boston, MA 02110, USA}

\begin{abstract}
The number of measurements demanded by hybrid quantum-classical algorithms such as the variational quantum eigensolver (VQE) is prohibitively high for many problems of practical value.
For such problems, realizing quantum advantage will require methods which dramatically reduce this cost.
Previous quantum algorithms that reduce the measurement cost (e.g. quantum amplitude and phase estimation) require error rates that are too low for near-term implementation. 
Here we propose methods that take advantage of the available quantum coherence to maximally enhance the power of sampling on noisy quantum devices, reducing measurement number and runtime compared to the standard sampling method of the variational quantum eigensolver (VQE). 
Our scheme derives inspiration from quantum metrology, phase estimation, and the more recent ``alpha-VQE" proposal, arriving at a general formulation that is robust to error and does not require ancilla qubits.
The central object of this method is what we call the  ``engineered likelihood function'' (ELF), used for carrying out Bayesian inference. We show how the ELF formalism enhances the rate of information gain
in sampling as the physical hardware transitions from the regime of noisy intermediate-scale quantum computers into that of quantum error corrected ones.
This technique speeds up a central component of many quantum algorithms, with applications including chemistry, materials, finance, and beyond.
Similar to VQE, we expect small-scale implementations to be realizable on today's quantum devices.
\end{abstract}

\maketitle

\section{Introduction}
\label{sec:intro}

Which quantum algorithms will deliver practical value first? A recent flurry of methods that cater to the limitations of near-term quantum devices have drawn significant attention. These methods include the variational quantum eigensolver (VQE) \cite{Peruzzo2014, Wecker2015, mcclean2016theory, Romero_2018}, quantum approximate optimization algorithm (QAOA) \cite{Farhi2014} and variants \cite{Hadfield2019}, variational quantum linear systems solver \cite{1909.05820v1, 1909.05500v2, 1909.03898v1}, other quantum algorithms leveraging the variational principles  \cite{Li2017}, and quantum machine learning algorithms \cite{Romero_2017,PhysRevLett.122.040504,Zhueaaw9918}.
In spite of such algorithmic innovations, many of these approaches have appeared to be impractical for commercially-relevant problems owing to their high cost in terms of number of measurements \cite{Wecker2015,1907.13117v3} and hence runtime, of the expectation value estimation subroutine.
Through an extensive benchmarking study of the VQE algorithm \cite{gonthier2020identifying}, it was shown that for a set of molecules having industrial relevance, VQE is very unlikely to yield an advantage over state-of-the-art quantum chemistry methods.
Unfortunately, methods offering a quadratic speedup over VQE in the runtime of the expectation value estimation subroutine, such as phase estimation, demand quantum resources that are far beyond the reach of near-term devices for moderately large problem instances \cite{Babbush2018}.

Expectation value estimation is also a building block for many non-variational quantum algorithms that have high-impact applications.
Unfortunately, the standard versions of these algorithms lie out of reach for near-term quantum computers, in part due to the coherence requirements needed to implement estimation subroutines. Such techniques include quantum algorithms for Monte Carlo estimation \cite{montanaro2015quantum} and quantum algorithms for solving linear systems of equations \cite{harrow2009quantum, CJS2013}.
These algorithms find application in finance, engineering, and machine learning.
Thus, there is strong motivation for developing estimation methods that make these techniques more amenable to near-term implementation.

Recently, the method of ``$\alpha$-VQE" \cite{wang2019accelerated} was proposed for interpolating between VQE and phase estimation in terms of the asymptotic tradeoff between sample count and quantum coherence.
The basic idea is to start from the general framework of VQE, namely the iterative optimization of the energy expectation which is a sum of individual operator expectations, and proceed to estimate each individual operator with a Bayesian variant of the overlap estimation algorithm \cite{knill2007optimal} that shares the same statistical inference backbone with known Bayesian parameter estimation schemes \cite{Sergeevich2011,Ferrie2012,svore2013faster,wiebe2016efficient}.
While phase estimation is commonly regarded as a quantum algorithm intended for fault-tolerant quantum computers, previous works \cite{wiebe2016efficient,wang2019accelerated} have demonstrated that in a noisy setting, {Bayesian} phase estimation can still yield quantum advantages in sampling. For instance, in \cite[Sec.\ IIIA]{wiebe2016efficient} it is shown that depolarizing noise reduces but does not eliminate the ability of the likelihood function to distinguish between different possible values of the parameter to be estimated.

This motivates the central question of our work: with realistic, \emph{noisy} quantum computers, how do we maximize information gain from the coherence available to speed up  expectation value estimation, and in doing so, speed up algorithms such as VQE that rely on sampling? We note that this question is not only urgently relevant in the current era of noisy quantum devices without quantum error correction, but remains relevant for error-corrected quantum computation.

In this work, we investigate the impact of gate and measurement-readout error on the performance of quantum estimation tasks such as amplitude estimation and expectation value estimation. Note that the standard formulation of amplitude estimation \cite{brassard2002quantum} is equivalent to estimating the expectation value of an observable with eigenvalues $\pm 1$ with respect to a quantum state generated by a given circuit. We introduce a simple noise model and show that the typical sample-generation scheme is hindered by a phenomenon we refer to as ``dead spots'' in the likelihood function. Motivated by these findings, we develop the framework of \emph{engineered likelihood functions (ELFs)}, in which signal processing techniques are used to boost the information gain per sample during the inference process of estimation. We develop several algorithms for the framework of engineered likelihood functions and investigate their performance with simulations. Finally, we develop a model for the runtime performance of these estimation algorithms and discuss the implications for near-term and far-term quantum computing.

The remainder of the paper is organized as follows. The remaining subsections of the introduction review relevant prior work on quantum estimation and describe our main results in more detail. In Section \ref{subsec:example}, we present a concrete example of our scheme for readers who wish to glean only a first impression of the key ideas. Subsequent sections then expand on the general formulation of our scheme. Section \ref{sec:elf} describes in detail the general quantum circuit construction for realizing ELFs, and analyzes the structure of ELF in both noisy and noiseless settings. In addition to the quantum circuit scheme, our scheme also involves 1) tuning the circuit parameters to maximize information gain, and 2) Bayesian inference for updating the current belief about the distribution of the true expectation value. Section \ref{sec:optalg} presents heuristic algorithms for both. We then show numerical results in Section \ref{sec:simulatin_results} comparing our approach with existing methods based on Chebyshev likelihood functions (CLFs). In Section \ref{sec:runtimemodel}, we derive a mathematical model for the runtimes of our algorithms on noisy devices. We conclude in Section \ref{sec:outlook} with implications of our results from a broad perspective of quantum computing.

\subsection{Prior work}

Evaluating the expectation value of an operator $O$ with respect to a quantum state $|A\rangle$ is a fundamental element of quantum information processing. 
In the simplest setting where samples are drawn repeatedly by measuring the same operator $O$ on the same quantum state $\ket{A}$, the measurement process is equivalent to a sequence of independent Bernoulli experiments. Yielding an estimate of the expectation value within error $\varepsilon$
 (with high probability) requires the number of samples to scale as 
$O(1/\varepsilon^2)$. We highlight the following key points regarding quantum estimation that are relevant to the context of this work:
\begin{enumerate}
    \item {\bf Cost scaling improvement using phase estimation}. Quantum effects introduce the opportunity to asymptotically accelerate the measurement process. In particular, there is a set of schemes \cite{knill2007optimal} based on quantum phase estimation that is able to reduce the sample complexity to $O(\log\frac{1}{\varepsilon})$. This saturates the information-theoretical lower bound for the number of samples since in order to determine a bounded quantity to resolution $\varepsilon$ one must be able to distinguish  $O(1/\varepsilon)$ different values, requiring at least  $\Omega(\log\frac{1}{\varepsilon})$ bits \cite{svore2013faster}. However, such optimal sample complexity comes at the cost of $O(1/\varepsilon)$ many coherent quantum operations. 
    This tradeoff between sample complexity and quantum coherence is also well understood in quantum metrology \cite{Giovannetti2004,Giovannetti2011}. 
    \label{item:pea}

    \item {\bf Amplitude estimation and generalized reflections}.
    Phase estimation is closely related to the task of amplitude estimation \cite{brassard2002quantum} with many of the performance guarantees of the former applying to the latter. In its original definition, amplitude estimation is the problem of evaluating the quantity $\bra{A} P_+ \ket{A}$, where $P_+$ is a projection operator and $\ket{A}=A\ket{0^n}$ is an ansatz state. This problem is essentially equivalent to estimating the expectation value $\bra{A} P \ket{A}$ where $P=2P_+-I$ is a reflection operator. \cite{brassard2002quantum} showed that amplitude estimation can be solved by running phase estimation on the Grover iterate $U=(2\ketbra{A}-I)P$. Namely, the desired amplitude information is encoded in the eigenvalues of the unitary operator $U$. Subsequent works demonstrated \cite{yoder2014fixed} that using generalized reflection operators (i.e. those with a phase parameter $\varphi$ such that $R_\varphi=(1-e^{i\varphi})|A\rangle\langle A|-I$ or $R_\varphi=(1-e^{i\varphi})P_+-I$), one can realize a much larger set of SU(2) rotations in the subspace $\mathrm{span}\{\ket{A}, P\ket{A}\}$ than with only common reflection operators. The set of SU(2) rotations realizable with such generalized construction has also been rigorously characterized \cite{low2016methodology} and later used for some of the most advanced Hamiltonian simulation algorithms such as qubitization \cite{Low2019} and signal processing \cite{Low2017}.
    \label{item:ae}
        \item {\bf Bayesian inference perspective.} 
        The problem of expectation value estimation can be framed as a parameter estimation problem, common in statistical inference. In fact, previous work (for example \cite[Sec.\ IVA]{mcclean2016theory}) has already pointed out a Bayesian perspective for considering the standard sampling process for VQE algorithms. The general setting is first to treat the operator expectation $\Pi=\langle A|O|A\rangle$ as the parameter to be estimated. Then, a parametrized quantum circuit $V(\vec{\theta})$ that may be related to $\ket{A}$ and $O$ is constructed. The ability to execute the circuit and collect measurement outcome $d$ translates to the ability to sample from a likelihood function $p(d|\vec{\theta},\Pi)$. For a given prior $p(\Pi)$ representing the current belief of the true value of $\Pi$, Bayesian inference uses a measurement outcome $d$ to produce (or update the prior to) a posterior distribution $p(\Pi|d)=\frac{p(d|\Pi,\vec{\theta})p(\Pi)}{\int p(d|\Pi,\vec{\theta})p(\Pi)d\Pi}$. For the settings considered in this paper, as well as in previous works \cite{Sergeevich2011,Ferrie2012,svore2013faster,wiebe2016efficient}, the prior and posterior distributions are maintained on the classical computer, while sampling from the likelihood function involves using a quantum device.
    \label{item:bayesian}
\end{enumerate}

The combination of phase estimation and the Bayesian perspective gives rise to Bayesian phase estimation techniques \cite{svore2013faster,wiebe2016efficient,O_Brien_2019} that are more suitable for noisy quantum devices capable of realizing limited-depth quantum circuits than earlier proposals \cite{quant-ph/9511026}. 
The goal of Bayesian phase estimation is to estimate the phase $\theta = \arccos(\Pi)$ in an eigenvalue $e^{i\theta}$ of the unitary $U$.
The quantum circuits used in this algorithm yield measurement outcomes with likelihoods given by
\begin{align}
\label{eq:pdmbP}
p(d|m,\Pi)=\frac{1+(-1)^d\mathcal{T}_m(\Pi)}{2},
\end{align}
where $d\in\{0,1\}$ and $\mathcal{T}_m(\Pi)=\cosp{m\arccos(\Pi)}$ is the $m$th-degree Chebyshev polynomial, found in many settings beyond Bayesian phase estimation (c.f.\ \cite[Eq.\ 2]{Sergeevich2011}, \cite[Eq.\ 1]{Ferrie2012}, \cite[Eq.\ 2]{wiebe2016efficient}, and \cite[Eq.\ 4]{wang2019accelerated}). In \cite{Ferrie2012} the exponential advantage of Bayesian inference with a Gaussian prior over other non-adaptive sampling methods is established by showing that the expected posterior variance $\sigma$ decays \emph{exponentially} in the number of inference steps. Such exponential convergence is at a cost of $O(1/\sigma)$ amount of quantum coherence required at each inference step \cite{Ferrie2012}. Such scaling is also confirmed in \cite{wiebe2016efficient} in the context of Bayesian phase estimation.

Combining the above observations one may devise a Bayesian inference method for expectation value estimation that smoothly interpolates between the standard sampling regime and phase estimation regime. This is proposed in \cite{wang2018quantum} as ``$\alpha$-VQE", where the asymptotic scaling is $O(1/\varepsilon^\alpha)$ with the extremal values of $\alpha=2$ corresponding to the standard sampling regime (typically realized in VQE) and $\alpha=1$ corresponding to the quantum-enhanced regime where the scaling reaches the Heisenberg limit (typically realized with phase estimation). By varying the parameters for the Bayesian inference one can also achieve $\alpha$ values between $1$ and $2$. The lower the $\alpha$ value is, the deeper the quantum circuit is needed for Bayesian phase estimation. This accomplishes the tradeoff between quantum coherence and asymptotic speedup for the measurement process (point \ref{item:pea} above).

It is also worth noting that phase estimation is not the only paradigm that can reach the Heisenberg limit for amplitude estimation \cite{4655455,Zintchenko2016,Suzuki2020}. In \cite{4655455} the authors consider the task of estimating the parameter $\theta$ of a quantum state $\rho_\theta$. A parallel strategy is proposed where $m$ copies of the parametrized circuit for generating $\rho_\theta$, together with an entangled initial state and measurements in an entangled basis, are used to create states with the parameter $\theta$ amplified to $m\theta$. Such amplification can also give rise to likelihood functions that are similar to that in Eq. \eqref{eq:pdmbP}. In \cite{Zintchenko2016} it is shown that with randomized quantum operations and Bayesian inference one can extract information in fewer iterations than classical sampling even in the presence of noise. In \cite{Suzuki2020} quantum amplitude estimation circuits with varying numbers $m$ of iterations and numbers $N$ of measurements are considered. A particularly chosen set of pairs $(m,N)$ gives rise to a likelihood function that can be used for inferring the amplitude to be estimated. The Heisenberg limit is demonstrated for one particular likelihood function construction given by the authors. Both works \cite{Zintchenko2016,Suzuki2020} highlight the power of parametrized likelihood functions, making it tempting to investigate their performance under imperfect hardware conditions.
As will become clear, although the methods we propose can achieve Heisenberg-limit scaling, they do not take the perspective of many previous works that consider interfering many copies of the same physical probe.

\subsection{Main results}

This work focuses on estimating the expectation value $\Pi=\langle A|O|A\rangle$ where the state $|A\rangle$ can be prepared by a circuit $A$ such that $|A\rangle=A|0^n\rangle$ for some integer $n \ge 1$. We consider a family of quantum circuits such that as the circuit deepens with more repetitions of $A$ it allows for likelihood functions that are polynomial in $\Pi$ of ever higher degree. As we will demonstrate in the next section with a concrete example, a direct consequence of this increase in polynomial degree is an increase in the power of inference, which can be quantified by Fisher information gain at each inference step. After establishing this ``enhanced sampling" technique, we further introduce parameters into the quantum circuit and render the resulting likelihood function tunable. We then optimize the parameters for maximal information gain during each step of inference. The following lines of insight emerge from our efforts:

\begin{enumerate}
    \item {\bf The role of noise and error in amplitude estimation:} Previous works \cite{wang2019accelerated,wiebe2016efficient,O_Brien_2019,Zintchenko2016} have revealed the impact of noise on the likelihood function and the estimation of the Hamiltonian spectrum. Here we investigate the same for our scheme of amplitude estimation. Our findings show that while noise and error do increase the runtime needed for producing an output that is within a specific statistical error tolerance, they do not necessarily introduce systematic bias in the output of the estimation algorithm. 
    Systematic bias in the estimate can be suppressed by using active noise-tailoring techniques \cite{wallman2016noise} and calibrating the effect of noise.
    
    We have performed simulation using realistic error parameters for near-term devices and discovered that the enhanced sampling scheme can outperform VQE in terms of sampling efficiency. Our results have also revealed a perspective on tolerating error in quantum algorithm implementation where higher fidelity does not necessarily lead to better algorithmic performance. For fixed gate fidelity, there appears to be an optimal circuit fidelity around the range of $0.5-0.7$ at which the enhanced scheme yields the maximum amount of quantum speedup.
    
    \item {\bf The role of likelihood function tunability:} Parametrized likelihood functions are centerpieces of phase estimation or amplitude estimation routines. To our knowledge, all of the current methods focus on likelihood functions of the Chebyshev form (Eq. \eqref{eq:pdmbP}). For these Chebyshev likelihood functions (CLF) we observe that in the presence of noise there are specific values of the parameter $\Pi$ (the ``dead spots") for which the CLFs are significantly less effective for inference than other values of $\Pi$. We remove such dead spots by engineering the form of the likelihood function with generalized reflection operators (point \ref{item:ae} in Section \ref{sec:intro}) whose angle parameters are made tunable.
    
    \item {\bf Runtime model for estimation as error rates decrease:} 
    Previous works \cite{wang2019accelerated,Suzuki2020} have demonstrated smooth transitions in the asymptotic cost scaling from the $O(1/\varepsilon^2)$ of VQE to $O(1/\varepsilon)$ of phase estimation. We advance this line of thinking by developing a model for estimating the runtime $t_{\varepsilon}$ to target accuracy $\varepsilon$ using devices with degree of noise $\lambda\in[0,\infty)$ (c.f. Section \ref{sec:runtimemodel}):
    \begin{align}
        \label{eq:runtime_model}
        t_{\varepsilon} \sim O\left(\frac{ \lambda }{\varepsilon^2} + \frac{1}{\sqrt{2}\varepsilon} +\sqrt{\left(\frac{\lambda}{\varepsilon^2}\right)^2+\left(\frac{2\sqrt{2}}{\varepsilon}\right)^2}\right).
    \end{align}
    The model interpolates between the $O(1/\varepsilon)$ scaling and $O(1/\varepsilon^2)$ scaling as a function of $\lambda$. Such bounds also allow us to make concrete statements about the extent of quantum speedup as a function of hardware specifications such as the number of qubits and two-qubit gate fidelity, and therefore estimate runtimes using realistic parameters for current and future hardware.
\end{enumerate}

\begin{table*}[!ht]
    \caption{
    Comparison of our scheme with relevant proposals that appear in the literature. Here the list of features include whether the quantum circuit used in the scheme requires ancilla qubits in addition to qubits holding the state for amplitude estimation or phase estimation, whether the scheme uses Bayesian inference, whether any noise resilience is considered, whether the initial state is required to be an eigenstate, and whether the likelihood function (LF) is fully tunable like engineered likelihood functions (ELFs) proposed here or restricted to Chebyshev likelihood functions (CLFs).}
    \begin{ruledtabular}
    \begin{tabular}{cccccc}
         Scheme & Bayesian inference & Noise consideration & Fully tunable LFs & Requires ancilla & Requires eigenstate \\
         \hline
         Knill et al.\ \cite{knill2007optimal} & No & No & No & Yes & No \\
         Svore et al.\ \cite{svore2013faster} & No & No & No & Yes & Yes \\
         Wiebe and Grenade \cite{wiebe2016efficient} & Yes & Yes & No & Yes & Yes \\
         Wang et al.\ \cite{wang2019accelerated} & Yes & Yes & No & Yes & Yes \\
         O'Brien et al.\ \cite{O_Brien_2019} & Yes &  Yes & No & Yes & No \\
         Zintchenko and Wiebe \cite{Zintchenko2016} & No & Yes & No & No & No \\
         Suzuki et al.\ \cite{Suzuki2020} & No & No & No & No & No \\
         This work (Section \ref{sec:elf}) & Yes & Yes & Yes & No & No \\
         This work (Appendix \ref{sec:ancilla_based}) & Yes & Yes & Yes & Yes & No \\
    \end{tabular}
    \end{ruledtabular}
    \label{tab:comparison}    
\end{table*}

\section{A first example}
\label{subsec:example}

There are two main strategies for estimating the expectation value $\bra{A} P \ket{A}$ of some operator $P$ with respect to a quantum state $\ket{A}$.
The method of \emph{quantum amplitude estimation} \cite{brassard1998quantum} provides a provable quantum speedup with respect to certain computational models.
However, to achieve precision $\varepsilon$ in the estimate, the circuit depth needed in this method scales as $O(1/\varepsilon)$, making it impractical for near-term quantum computers.
The variational quantum eigensolver uses \emph{standard sampling} to carry out amplitude estimation. 
Standard sampling allows for low-depth quantum circuits, making it more amenable to implementation on near-term quantum computers.
However, in practice, the inefficiency of this method makes VQE impractical for many problems of interest \cite{Wecker2015}.
In this section we introduce the method of \emph{enhanced sampling} for amplitude estimation.
This technique draws inspiration from quantum-enhanced metrology \cite{giovannetti2006quantum} and seeks to maximize the statistical power of noisy quantum devices.
We motivate this method by starting from a simple analysis of standard sampling as used in VQE.
We note that, although the subroutine of estimation is a critical bottleneck, other aspects of the VQE algorithm also must be improved, including the optimization of the parameters in parameterized quantum circuits \cite{kubler2020adaptive, sweke2019stochastic, arrasmith2020operator,  sung2020exploration}.

The energy estimation subroutine of VQE estimates amplitudes with respect to Pauli strings.
For a Hamiltonian decomposed into a linear combination of Pauli strings $H=\sum_j \mu_j P_j$ and ``ansatz state'' $\ket{A}$, the energy expectation value is estimated as a linear combination of Pauli expecation value estimates
\begin{align}
    \hat{E}=\sum_j \mu_j \hat{\Pi}_j,
\end{align}
where $\hat{\Pi}_j$ is the (amplitude) estimate of $\langle A|P_j|A\rangle$.
VQE uses the standard sampling method to build up Pauli expectation value estimates with respect to the ansatz state, which can be summarized as follows. Prepare $\ket{A}$ and measure operator $P$ receiving outcome $d \in \{0,1\}$. Repeat this $M$ times, receiving $k$ outcomes labeled $0$ and $M-k$ outcomes labeled $1$. Estimate $\Pi=\langle A|P|A\rangle$ as $\hat{\Pi}=\frac{k-(M-k)}{M}$.

We can quantify the performance of this estimation strategy
using the mean squared error of the estimator as a function of time $t=TM$, where $T$ is the time cost of each measurement.
Because the estimator is unbiased, the mean squared error is simply the variance in the estimator,
\begin{align}
    \textup{MSE}(\hat{\Pi})=\frac{1-\Pi^2}{M}.
\end{align}
For a specific mean squared error $\textup{MSE}(\hat{\Pi})=\varepsilon^2$, the runtime of the algorithm needed to ensure mean squared error $\varepsilon^2$ is
\begin{align}
    t_{\varepsilon}=T\frac{1-\Pi^2}{\varepsilon^2}.
\end{align}
The total runtime of energy estimation in VQE is the sum of the runtimes of the individual Pauli expectation value estimation runtimes.
For problems of interest, this runtime can be far too costly, even when certain parallelization techniques are used \cite{kandala2017hardware}.
The source of this cost is the insensitivity of the standard sampling estimation process to small deviations in $\Pi$: the expected information gain about $\Pi$ contained in the standard-sampling measurement outcome data is low.

Generally, we can measure the information gain of an estimation process of $M$ repetitions of standard sampling with the Fisher information
\begin{align}
\label{eq:fisherinfo_def}
    I_M(\Pi)&=\mathbb{E}_{D}\left[\left(\frac{\partial}{\partial \Pi} \log \mathbb{P}(D|\Pi)\right)^2\right]
    \nonumber\\
    &=-\mathbb{E}_{D}\left[\frac{\partial^2}{\partial \Pi^2} \log \mathbb{P}(D|\Pi)\right]\nonumber\\
    &=\sum_D \frac{1}{\mathbb{P}(D|\Pi)}\left(\frac{\partial}{\partial \Pi}\mathbb{P}(D|\Pi)\right)^2,
\end{align}
where $D=\{d_1,d_2,\cdots,d_M\}$ is the set of outcomes from $M$ repetitions of the standard sampling.
The Fisher information identifies the likelihood function $\mathbb{P}(D|\Pi)$ as being responsible for information gain.
We can lower bound the mean squared error of an (unbiased) estimator with the Cramer-Rao bound
\begin{align}
    \text{MSE}(\hat{\Pi})\geq \frac{1}{I_M(\Pi)}.
\end{align}
Using the fact that the Fisher information is additive in the number of samples, we have $I_M(\Pi)=MI_1(\Pi)$ where $I_1(\Pi)=1/(1-\Pi^2)$ is the Fisher information of a single sample drawn from likelihood function $\mathbb{P}(d|\Pi)=(1+(-1)^d\Pi)/2$. Using the Cramer-Rao bound, we can find a lower bound for the runtime of the estimation process as
\begin{align}
    t_{\varepsilon}\geq \frac{T}{I_1(\Pi)\varepsilon^2},
\end{align}
which shows that in order to reduce the runtime of an estimation algorithm we should aim to increase the Fisher information.

\begin{figure*}[!ht]
\includegraphics[trim=40 50 50 50, clip=true, width=0.75\linewidth]{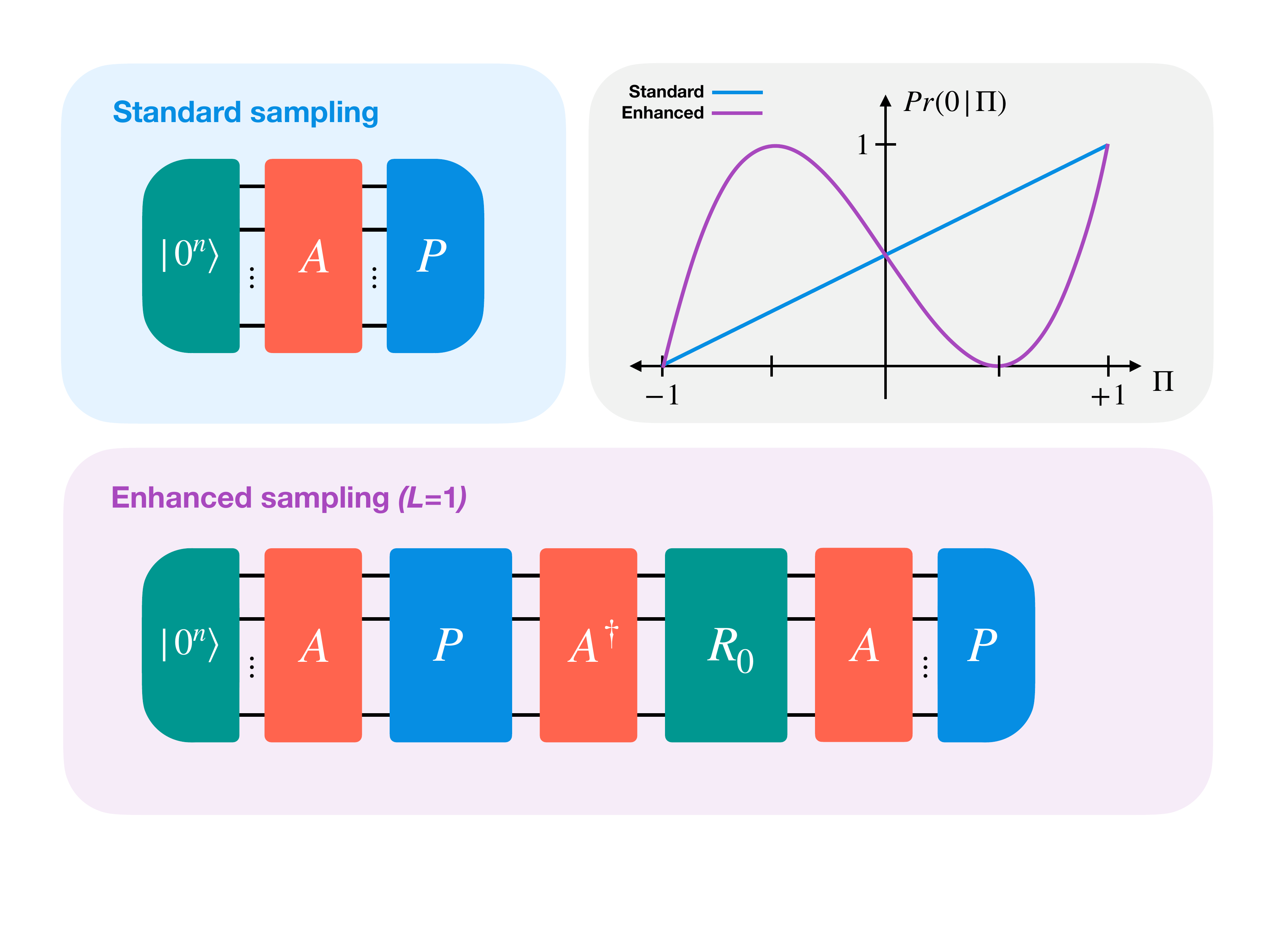}
\caption{This figure exhibits the quantum circuits for standard sampling (used in VQE) and for the simplest non-trivial version of enhanced sampling, along with their corresponding likelihood functions. The blocks represent unitary transformations, while the caps at the left and right indicate state preparation and measurement, respectively. Enhanced sampling can yield a larger statistical power in this estimation of $\Pi=\bra{A}P\ket{A}$. The likelihoods of the outcome data in enhanced sampling can depend more sensitively on the unknown value $\Pi$ than they do in standard sampling. This increased sensitivity can reduce the runtime needed to achieve a target average error in the estimate of $\Pi$.}
\label{fig:L1example}
\end{figure*}

The purpose of enhanced sampling is to reduce the runtime of amplitude estimation by engineering likelihood functions that increase the rate of information gain.
We consider the simplest case of enhanced sampling, which is illustrated in Figure \ref{fig:L1example}.
To generate data we prepare the ansatz state $\ket{A}$, apply the operation $P$, apply a phase flip about the ansatz state, and then measure $P$.
The phase flip about the ansatz state can be achieved by applying the inverse of the ansatz circuit $A^{-1}$, applying a phase flip about the initial state $R_0=2\ketbra{0^n}-\identity$, and then re-applying the ansatz circuit $A$.
In this case, the likelihood function becomes
\begin{align}
\label{eq:standardlf}
\mathbb{P}(d|\Pi)&=\frac{1+(-1)^d\cos(3\arccos(\Pi))}{2}\nonumber\\
&=\frac{1+(-1)^d(4\Pi^3-3\Pi)}{2}.
\end{align}
The bias is a degree-$3$ Chebyshev polynomial in $\Pi$.
We will refer to such likelihood functions as \emph{Chebyshev likelihood functions} (CLFs).

In order to compare the Chebyshev likelihood function of enhanced sampling to that of standard sampling, we consider the case of $\Pi=0$.
Here, $\mathbb{P}(0|\Pi=0)=\mathbb{P}(1|\Pi=0)$ and so the Fisher information is proportional to the square of the slope of the likelihood function
\begin{align}
    I_1(\Pi=0)=4\left(\frac{\partial \mathbb{P}(d=0|\Pi)}{\partial\Pi}\right)^2.
\end{align}
As seen in Figure \ref{fig:L1example}, the slope of the Chebyshev likelihood function at $\Pi=0$ is steeper than that of the standard sampling likelihood function.
The single-sample Fisher information in each case evaluates to
\begin{align}
    \textup{Standard: }I_1(\Pi=0)=1\nonumber\\ 
    \textup{Enhanced: }I_1(\Pi=0)=9,
\end{align}
demonstrating how a simple variant of the quantum circuit can enhance information gain.
In this example, using the simplest case of enhanced sampling can reduce the number of measurements needed to achieve a target error by at least a factor of nine.
As we will discuss later, we can further increase the Fisher information by applying $L$ layers of $P \circ A^{\dagger} \circ R_0 \circ A$ before measuring $P$. In fact, the Fisher information $I_1(\Pi)=\frac{(2L+1)^2}{1-\Pi^2}=O(L^2)$ grows quadratically in $L$.

We have yet to propose an estimation scheme that converts enhanced sampling measurement data into an estimation.
One intricacy that enhanced sampling introduces is the option to vary $L$ as we are collecting measurement data.
In this case, given a set of measurement outcomes from circuits with varying $L$, the sample mean of the $0$ and $1$ counts loses its meaning.
Instead of using the sample mean, we use \emph{Bayesian inference} to process the measurement outcomes into information about $\Pi$.
Section \ref{subsec:bi_elf} describes the use of Bayesian inference for estimation.

At this point, one may be tempted to point out that the comparison between standard sampling and enhanced sampling is unfair because only one query to $A$ is used in the standard sampling case while the enhanced sampling scheme uses three queries of $A$. It seems that if one considers a likelihood function that arises from \emph{three} standard sampling steps, one could also yield a cubic polynomial form in the likelihood function. Indeed, suppose one performs three independent standard sampling steps yielding results $x_1,x_2,x_3\in\{0,1\}$, and produces a binary outcome $z\in\{0,1\}$ classically by sampling from a distribution $\mathbb{P}(z|x_1,x_2,x_3)$. Then the likelihood function takes the form of
\begin{align}
\label{eq:classical_elf}
    \mathbb{P}(z|\Pi)
    &=\sum_{x_1,x_2,x_3}\mathbb{P}(z|x_1,x_2,x_3)\mathbb{P}(x_1,x_2,x_3|\Pi)\nonumber\\
    &=\sum_{i=0}^3\alpha_i
    \begin{pmatrix}
    3 \\ i
    \end{pmatrix}
    \left(\frac{1+\Pi}{2}\right)^i\left(\frac{1-\Pi}{2}\right)^{3-i},
\end{align}
where each $\alpha_i\in[0,1]$ is a parameter that can be tuned classically through changing the distribution $\mathbb{P}(z|x_1,x_2,x_3)$. More specifically, $\alpha_i=\sum_{x_1x_2x_3:h(x_1x_2x_3)=i}\mathbb{P}(z|x_1,x_2,x_3)$ where $h(x_1x_2x_3)$ is the Hamming weight of the bit string $x_1x_2x_3$. Suppose we want $\mathbb{P}(z=0|\Pi)$ to be equal to $\mathbb{P}(d=0|\Pi)$ in Eq. \eqref{eq:standardlf}. This implies that $\alpha_0=1$, $\alpha_1=-2$, $\alpha_2=3$ and $\alpha_3=0$, which is clearly beyond the classical tunability of the likelihood function in Eq. \eqref{eq:classical_elf}. This evidence suggests that the likelihood function arising from the quantum scheme in Eq. \eqref{eq:standardlf} is beyond classical means.

As the number of circuit layers $L$ is increased, the time per sample $T$ grows linearly in $L$.
This linear growth in circuit layer number, along with the quadratic growth in Fisher information leads to a lower bound on the expected runtime,
\begin{align}
    t_{\varepsilon} \in \Omega\left(\frac{1}{L\varepsilon^2}\right),
\end{align}
assuming a fixed-$L$ estimation strategy with an unbiased estimator.
In practice, the operations implemented on the quantum computer are subject to error.
Fortunately, Bayesian inference can incorporate such errors into the estimation process.
As long as the influence of errors on the form of the likelihood function is accurately modeled, the principal effect of such errors is only to slow the rate of information gain.
Error in the quantum circuit accumulates as we increase the number of circuit layers $L$.
Consequently, beyond a certain number of circuit layers, we will receive diminishing returns with respect to gains in Fisher information (or the reduction in runtime).
The estimation algorithm should then seek to balance these competing factors in order to optimize the overall performance.

The introduction of error poses another issue for estimation.
Without error, the Fisher information gain per sample in the enhanced sampling case with $L=1$ is greater than or equal to $9$ for all $\Pi$.
As shown in Figure \ref{fig:L1ELF}, with the introduction of even a small degree of error, the values of $\Pi$ where the likelihood function is flat incur a dramatic drop in Fisher information.
We refer to such regions as estimation \emph{dead spots}.
This observation motivates the concept of engineering likelihood functions (ELF) to increase their statistical power.
By promoting the $P$ and $R_0$ operations to generalized reflections $U(x)=\exp(-\i xP)$ and $R_0(y)=\exp(-\i y R_0)$, we can choose rotation angles such that the information gain is boosted around such dead spots.
We will find that even for deeper enhanced sampling circuits, engineering likelihood functions allows us to mitigate the effect of estimation dead spots.

\begin{figure*}[!ht]
\begin{minipage}{.48\textwidth}
\includegraphics[width=1\linewidth]{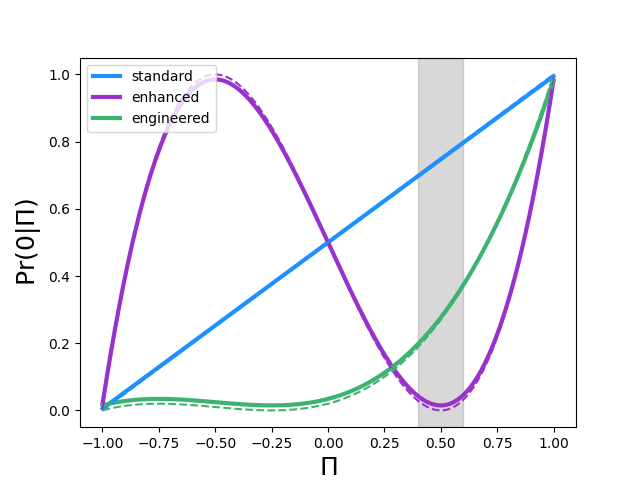}\\
{Likelihood function comparison}
\end{minipage}
\begin{minipage}{.48\textwidth}
\includegraphics[width=1\linewidth]{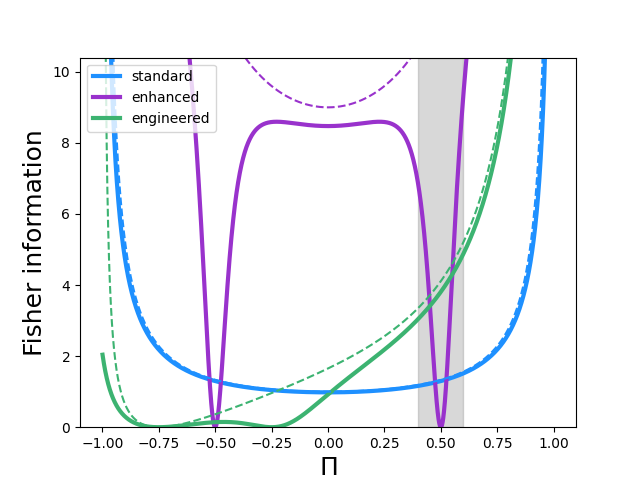}\\
{Fisher information comparison}
\end{minipage}
\caption{These figures demonstrate improvement in information gain if the likelihood function is engineered. Dotted curves are the noiseless case, solid curves incorporate a 1\% error per ansatz application (using the notation of Eq. \eqref{eq:noisy_elf}, $\bar{p}=0.99$, $p=0.99^2=0.9801$, and $L=1$). The likelihood function of enhanced sampling yields a large Fisher information for many values of $\Pi$ relative to that of standard sampling. The introduction of even a small degree of error in the quantum circuits causes the Fisher information to become zero where the enhanced sampling likelihood function is flat (indicated by the grey bands).
This can be verified by evaluating Eq. \eqref{eq:fisherinfo_def} using the expression for the model of the noisy likelihood function in Eq. \eqref{eq:noisy_elf}.
By tuning the generalized reflection angles ($(x_1, x_2)=(-0.6847,0.6847)$), we can engineer a likelihood function to boost the information gain in the estimation ``dead spot'' (gray region).
}
\label{fig:L1ELF}
\end{figure*}

\section{Engineered likelihood functions}
\label{sec:elf}
In this section, we propose the methodology of engineering likelihood functions for amplitude estimation. We first introduce the quantum circuits for drawing samples that correspond to engineered likelihood functions, and then describe how to tune the circuit parameters and carry out Bayesian inference with the resultant likelihood functions.

\subsection{Quantum circuits for engineered likelihood functions}
\label{subsec:qc_elf}
Our objective is to design a procedure for estimating the expectation value
\begin{align}
\Pi = \cosp{\theta} =\bra{A} P \ket{A},
\label{eq:defPi}
\end{align}
where $\ket{A}=A\ket{0^n}$ in which $A$ is an $n$-qubit unitary operator, $P$ is an $n$-qubit Hermitian operator with eigenvalues $\pm 1$, and $\theta=\arccosp{\Pi}$ is introduced to facilitate Bayesian inference later on. In constructing our estimation algorithms, we assume that we are able to perform the following primitive operations. First, we can prepare the computational basis state $\ket{0^n}$ and apply an ansatz circuit $A$ to it, obtaining $\ket{A}=A\ket{0^n}$. Second, we can implement the unitary operator $U(x)=\exp(-\i xP)$ for any angle $x \in \mathbb{R}$. Finally, we can perform the measurement of $P$ which is modeled as a projection-valued measure $\{\frac{I+P}{2},\frac{I-P}{2}\}$ with respective outcome labels $\{0,1\}$. We will also make use of the unitary operator $V(y)=A R_0(y) A^{\dagger}$, where $R_0(y)=\exp(-\i y(2\ketbra{0^n}-\identity))$ and $y \in \mathbb{R}$. Following the convention (see e.g. \cite{low2016methodology}), we will call $U(x)$ and $V(y)$ the \emph{generalized reflections} about the $+1$ eigenspace of $P$ and the state $\ket{A}$, respectively, where $x$ and $y$ are the \emph{angles} of these generalized reflections, respectively.

We use the ancilla-free \footnote{
We call this scheme ``ancilla-free'' (AF) since it does not involve any ancilla qubits. In Appendix \ref{sec:ancilla_based}, we consider a different scheme named the ``ancilla-based'' (AB) scheme that involves one ancilla qubit.
} quantum circuit in  Figure~\ref{fig:circuit_diagram} to generate the \emph{engineered likelihood function} (ELF), that is the probability distribution of the outcome $d \in \{0, 1\}$ given the unknown quantity $\theta$ to be estimated. The circuit consists of a sequence of generalized reflections. Specifically, after preparing the ansatz state $\ket{A}=A\ket{0^n}$, we apply $2L$ generalized reflections $U(x_1)$, $V(x_2)$, $\dots$, $U(x_{2L-1})$, $V(x_{2L})$ to it, varying the rotation angle $x_j$ in each operation. For convenience, we will call $V(x_{2j})U(x_{2j-1})$ the $j$-th layer of the circuit, for $j=1, 2, \dots, L$. The output state of this circuit is 
\begin{align} 
Q(\vec{x})\ket{A}=V(x_{2L})U(x_{2L-1})\ldots V(x_2)U(x_1)\ket{A},
\end{align}
where $\vec x = (x_1,x_2,\ldots, x_{2L-1},x_{2L}) \in \mathbb{R}^{2L}$ is the vector of tunable parameters. Finally, we perform the projective measurement $\{\frac{I+P}{2}, \frac{I-P}{2}\}$ on this state, receiving an outcome $d \in \{0,1\}$. 

\begin{figure*}[!ht]
\includegraphics[trim=40 40 60 300, clip=true, width=0.8\linewidth]{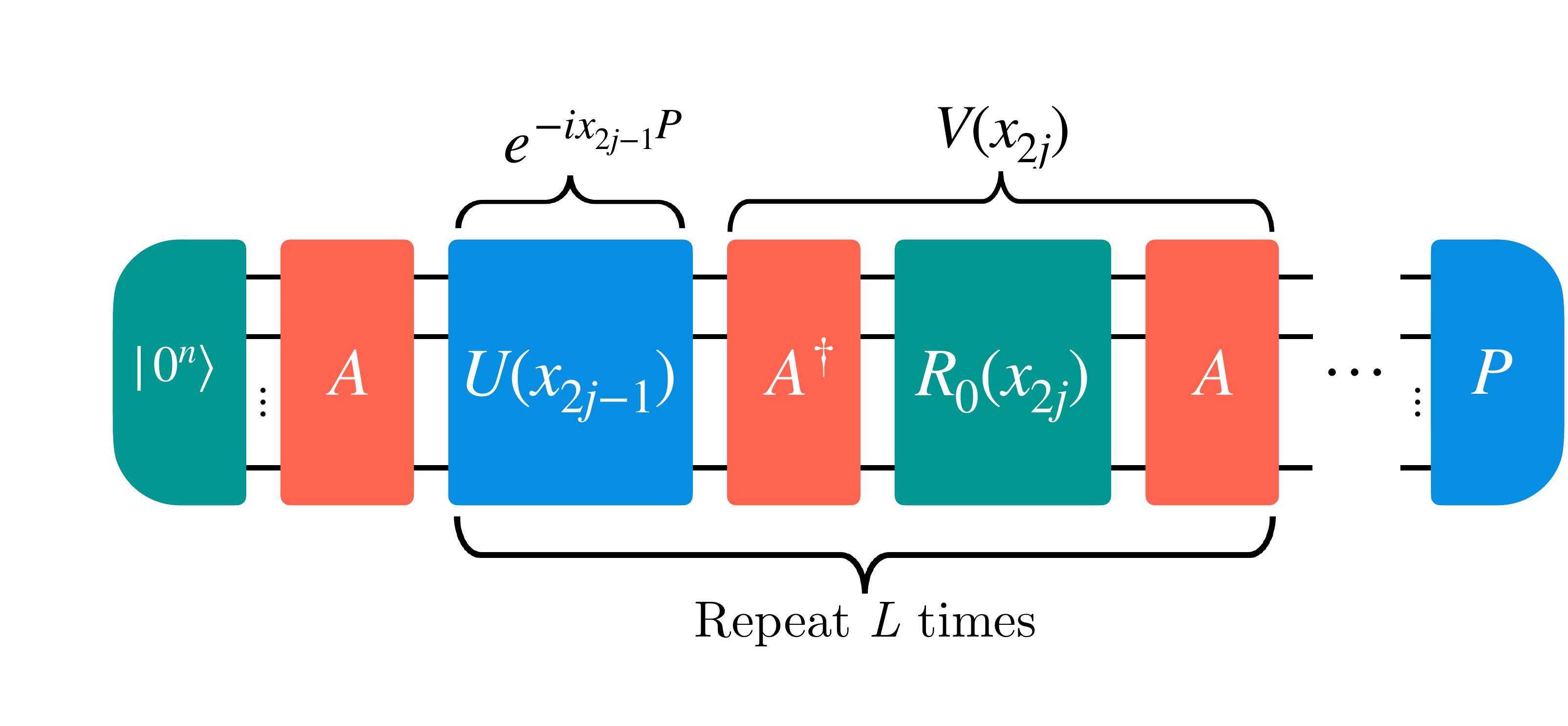}
\caption{This figure illustrates the operations used for generating samples that correspond to an engineered likelihood function. 
$A$ is the state preparation circuit, $P$ is the observable of interest, and $R_0(x_{i+1})$ is a generalized reflection about the state $\ket{0^n}$.
The blocks represent unitary transformations, while the caps at the left and right indicate state preparation and measurement, respectively.
The outcomes of measurement of $P$ yield information about the expectation value $\Pi=\langle A|P|A\rangle$.
The case of $L=0$ simply prepares $\ket{A}$ and measures $P$.
This corresponds to the standard sampling method used in VQE.
Even with an error-prone implementation, we can enhance the information gain rate by applying a sequence of generalized reflections before the measurement.
In such \emph{enhanced sampling}, the likelihood of outcomes depends more sensitively on $\Pi$.
These circuit elements are color-coded to highlight the commonalities in the way the features $P$ (blue), $A$ (red), and $\ket{0^n}$ (green) enter in the likelihood function.
} 
\label{fig:circuit_diagram}
\end{figure*}

As in Grover's search algorithm, the generalized reflections $U(x_{2j-1})$ and $V(x_{2j})$ ensure that the quantum state remains in two-dimensional subspace $S:=\mathrm{span}\{\ket{A}, P\ket{A}\}$ \footnote{To ensure that $S$ is two-dimensional, we assume that $\Pi \neq \pm 1$, i.e. $\theta \neq 0$ or $\pi$.} for any $j$. Let $\ket{A^\bot}$ be the state (unique, up to a phase) in $S$ that is orthogonal to $\ket{A}$, i.e.
\begin{align}
    \ket{A^\bot} = \frac{P\ket A - \bra A P \ket A \ket A}{\sqrt{1-\bra A P \ket A^2}}.
\label{eq:abot}
\end{align}
To help the analysis, we will view this two-dimensional subspace as a qubit, writing $\ket{A}$ and $|A^\bot\rangle$ as $\ket{\bar{0}}$ and $\ket{\bar{1}}$, respectively. 
Let $\bar{X}$, $\bar{Y}$, $\bar{Z}$ and $\bar{I}$ be the Pauli operators and identity operator on this virtual qubit, respectively. Then, focusing on the subspace $S=\mathrm{span}\{\ket{\bar{0}}, \ket{\bar{1}}\}$, we can rewrite $P$ as 
\begin{align}
P(\theta) = \cos(\theta) \bar{Z} + \sin(\theta) \bar{X},
\label{eq:matrixRep}
\end{align}
and rewrite the generalized reflections $U(x_{2j-1})$ and $V(x_{2j})$ as
\begin{align}
U(\theta; x_{2j-1})&=\cosp{x_{2j-1}} \bar{I} 
-\i \sinp{x_{2j-1}} [\cosp{\theta} \bar{Z} 
\nonumber\\
&\quad + \sinp{\theta} \bar{X}]
\label{eq:u_theta_x}
\end{align}
and 
\begin{align}
V(x_{2j})=\cosp{x_{2j}} \bar{I} - \i \sinp{x_{2j}} \bar{Z},
\label{eq:v_x}
\end{align}
where $x_{2j-1}, x_{2j} \in \mathbb{R}$ are tunable parameters. Then the unitary operator $Q(\vec x)$ implemented by the $L$-layer circuit becomes
\begin{align}
Q(\theta; \vec x)=V(x_{2L})U(\theta; x_{2L-1})\ldots V(x_2)U(\theta; x_1).
\label{eq:q_theta_vecx}
\end{align}
Note that in this picture, $\ket{A}=\ket{\bar{0}}$ is fixed, while $P=P(\theta)$, $U(x)=U(\theta; x)$ and $Q(\vec x)=Q(\theta; \vec x)$ depend on the unknown quantity $\theta$. It turns out to be more convenient to design and analyze the estimation algorithms in this ``logical" picture than in the original ``physical" picture. Therefore, we will stick to this picture for the remainder of this paper.

The engineered likelihood function (i.e. the probability distribution of measurement outcome $d \in \{0, 1\}$) depends on the output state $\rho(\theta; \vec x)$ of the circuit and the observable $P(\theta)$. Precisely, it is 
\begin{align}
\mathbb{P}(d|\theta; \vec x)=\dfrac { 
1+(-1)^d \Delta(\theta; \vec x) }{2},
\label{eq:lf_af}
\end{align}
where
\begin{align}
\Delta(\theta; \vec x) = \bra{\bar{0}} Q^\dagger(\theta;\vec x)P(\theta)Q(\theta;\vec x) \ket{\bar{0}}
\label{eq:delta_af}
\end{align}
is the $\emph{bias}$ of the likelihood function (from now on, we will use $\mathbb{P}'(d|\theta; \vec x)$ and $\Delta'(\theta;\vec x)$ to denote the partial derivatives of $\mathbb{P}(d|\theta; \vec x)$ and $\Delta(\theta; \vec x)$ with respect to $\theta$, respectively). In particular, if $\vec x=(\frac{\pi}{2}, \frac{\pi}{2}, \dots, \frac{\pi}{2}, \frac{\pi}{2})$, then we have
$\Delta(\theta; \vec x)=\cosp{(2L+1)\theta}$. Namely, the bias of the likelihood function for this $\vec x$ is the Chebyshev polynomial of degree $2L+1$ (of the first kind) of $\Pi$. For this reason, we will call the likelihood function for this $\vec x$ the \emph{Chebyshev likelihood function} (CLF). In Section \ref{sec:simulatin_results} we will explore the performance gap between CLFs and general ELFs.

In reality, quantum devices are subject to noise. To make the estimation process robust against errors, we incorporate the following exponential decay noise model into the likelihood function \cite{wiebe2016efficient} \footnote{In practice, the establishment of the noise model requires a procedure for \emph{calibrating the likelihood function} for the specific device being used. With respect to Bayesian inference, the parameters of this model are known as \emph{nuisance parameters} \cite{jaynes2003probability, royall2000probability}; the target parameter does not depend directly on them, but they determine how the data relates to the target parameter and, hence, should be incorporated into the inference process. We will explore likelihood function calibration in future work. For the remainder of this article, we will assume that the noise model has been calibrated to sufficient precision so as to render the effect of model error negligible. In Section \ref{sec:outlook} we touch upon the relationship between model error and estimation error.}. 
Recently this exponential decay noise model of the likelihood function has been used in several related works \cite{brown2020quantum, uno2020modified} and validated in small scale experiments \cite{tanaka2020amplitude}. It has also been used in an 18-qubit experiment on spectrum estimation \cite{aleiner2020accurately}. Furthermore, it is closely related to the noise analysis carried out in \cite{arute2019quantum}.
Validation at large scales remains an important line of research, which we emphasize in Section \ref{sec:outlook} and leave to future work.
Letting $p$ be the exponential decay factor, we have
\begin{align}
\mathbb{P}(d|\theta; \bar{p}, p, \vec x)=\frac 12 \left[ 
1+(-1)^d \bar{p}p^L\Delta(\theta; \vec x) \right],
\label{eq:noisy_elf}
\end{align}
where $\bar{p}$ accounts for SPAM error (c.f. Appendix \ref{sec:app_noise_model}) and $\Delta(\theta,\vec{x})$ is the bias of the ideal likelihood function as defined in Eq.~(\ref{eq:delta_af}). From now on, we will use $f=\bar{p}p^L$ as the \emph{fidelity} of the whole process for generating the ELF, and call $p$ the layer fidelity. Moreover, for convenience, we will write $\mathbb{P}(d|\theta; \bar{p}, p, \vec x)$ simply as $\mathbb{P}(d|\theta; f, \vec x)$ (we will also use $\mathbb{P}'(d|\theta; f, \vec x)$ to denote the partial derivative of $\mathbb{P}(d|\theta; f, \vec x)$ with respect to $\theta$). Note that the effect of noise on the ELF is that it rescales the bias by a factor of $f$. This implies that the less errored the generation process is, the steeper the resultant ELF is, as one would expect.

Before moving on to the discussion of Bayesian inference with ELFs, it is worth mentioning the following property of engineered likelihood functions, as it will play a pivotal role in Section \ref{sec:optalg}. In \cite{koh2020framework}, we introduced the concepts of \emph{trigono-multilinear} and \emph{trigono-multiquadratic} functions. Basically, a multivariable function $f: \mathbb{R}^k \to \mathbb{C}$ is trigono-multilinear if for any $j \in \{1,2,\dots,k\}$, we can write $f(x_1,x_2,\dots,x_k)$ as
\begin{align}
f(x_1,x_2,\dots,x_k) = C_j(\vec x_{\neg j})\cosp{x_j}+S_j(\vec x_{\neg j})\sinp{x_j},
\end{align}
for some (complex-valued) functions $C_j$ and $S_j$ of $\vec x_{\neg j}:=(x_1,\dots,x_{j-1}, x_{j+1}, x_k)$, and we call $C_j$ and $S_j$ the \emph{cosine-sine-decomposition (CSD)} coefficient functions of $f$ with respect to $x_j$. Similarly, a multivariable function $f: \mathbb{R}^k \to \mathbb{C}$ is trigono-multiquadratic if for any $j \in \{1,2,\dots,k\}$, we can write $f(x_1,x_2,\dots,x_k)$ as
\begin{align}
f(x_1,x_2,\dots,x_k) &= C_j(\vec x_{\neg j})\cosp{2x_j} +S_j(\vec x_{\neg j})\sinp{2x_j} \nonumber \\
&\quad +B_j(\vec x_{\neg j}),
\end{align}
for some (complex-valued) functions $C_j$, $S_j$ and $B_j$ of $\vec x_{\neg j}:=(x_1,\dots,x_{j-1}, x_{j+1}, x_k)$, and we call $C_j$, $S_j$ and $B_j$ the \emph{cosine-sine-bias-decomposition (CSBD)} coefficient functions of $f$ with respect to $x_j$. The concepts of trigono-multilinearity and trigono-multiquadraticity can be  naturally generalized to linear operators. Namely, a linear operator is trigono-multilinear (or trigono-multiquadratic) in a set of variables if each entry of this operator (written in an arbitrary basis) is trigono-multilinear (or trigono-multiquadratic) in the same variables. Now Eqs.~(\ref{eq:u_theta_x}), (\ref{eq:v_x}) and (\ref{eq:q_theta_vecx}) imply that $Q(\theta;\vec x)$ is a trigono-multilinear operator of $\vec x$. Then it follows from Eq.~(\ref{eq:delta_af}) that $\Delta(\theta; \vec x)$ is a trigono-multiquadratic function of $\vec x$. Furthermore, we will show in Section \ref{subsec:proxyopt} that the CSBD coefficient functions of $\Delta(\theta; \vec x)$ with respect to any $x_j$ can be evaluated in $O(L)$ time, and this greatly facilitates the construction of the algorithms in Section \ref{subsec:proxyopt} for tuning the circuit parameters $\vec x=(x_1, x_2, \dots, x_{2L-1}, x_{2L})$.

\subsection{Bayesian inference with engineered likelihood functions}
\label{subsec:bi_elf}

With the model of (noisy) engineered likelihood functions in place, we are now ready to describe our methodology for tuning the circuit parameters $\vec x$ and performing Bayesian inference with the resultant likelihood functions for amplitude estimation.

Let us begin with a high-level overview of our algorithm for estimating $\Pi=\cosp{\theta}=\bra{A} P \ket{A}$. For convenience, our algorithm mainly works with $\theta=\arccosp{\Pi}$ instead of $\Pi$. We use a Gaussian distribution to represent our knowledge of $\theta$ and make this distribution gradually concentrate around the true value of $\theta$ as the inference process proceeds. We start with an initial distribution of $\Pi$ (which can be generated by standard sampling or domain knowledge) and convert it to the initial distribution of $\theta$. Then we repeat the following procedure until a convergence criterion is satisfied. At each round, we first find the circuit parameters $\vec x$ that maximize the information gain from the measurement outcome $d$ in certain sense (based on our current knowledge of $\theta$). Then we run the quantum circuit in Figure \ref{fig:circuit_diagram} with the optimized parameters $\vec x$ and receive a measurement outcome $d \in \{0, 1\}$. Finally, we update the distribution of $\theta$ by using Bayes' rule, conditioned on the received outcome $d$. Once this loop is finished, we convert the final distribution of $\theta$ to the final distribution of $\Pi$, and set the mean of this distribution as the final estimate of $\Pi$. See Figure \ref{fig:ELF_diagram} for the conceptual diagram of this algorithm.

\tikzstyle{orangeBlock} = [rectangle, rounded corners, minimum width=2.8cm, minimum height=1.25cm,text centered, draw=black,fill=orange!12]
\tikzstyle{blueBlock} = [rectangle, rounded corners, minimum width=2.6cm, minimum height=1.25cm,text centered, draw=black,fill=blue!12]
\tikzstyle{blueDiamond} = [diamond, rounded corners, minimum width=2.6cm, minimum height=1cm,text centered, draw=black,fill=blue!12]
\tikzstyle{redBlock} = [rectangle, rounded corners, minimum width=2.8cm, minimum height=1.25cm,text centered, draw=black,fill=red!12]
\tikzstyle{arrow} = [->,>=stealth]

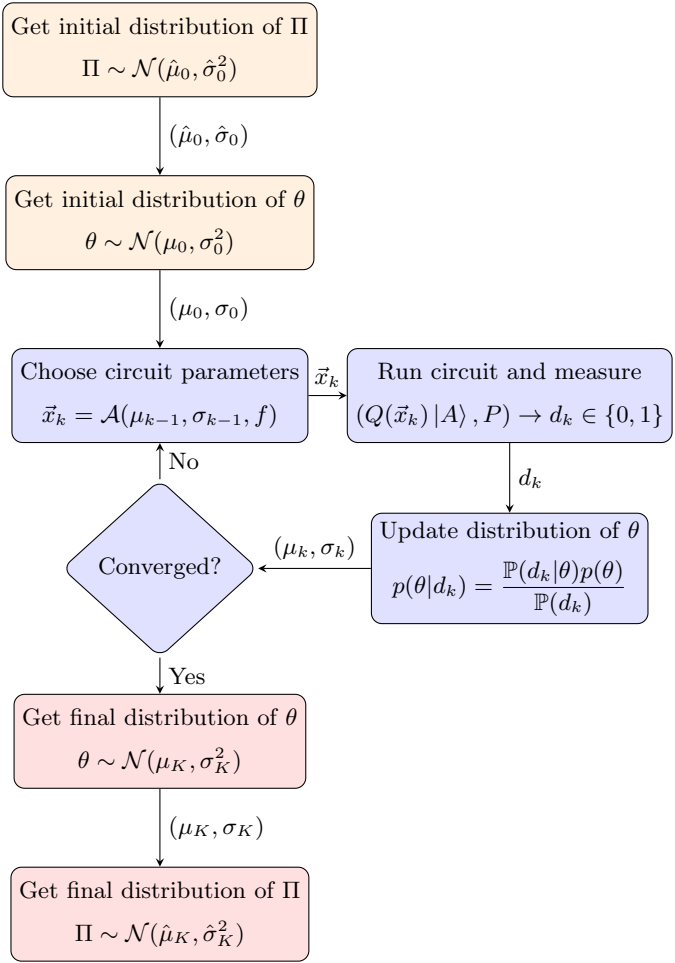
\begin{figure}[!ht]
\begin{tikzpicture}[node distance=2.3cm, auto]
  \node [redBlock, align=center] (05) {Get final distribution of $\Pi$ \\[0.2cm] 
  $\Pi \sim \mathcal{N}(\hat{\mu}_K, \hat{\sigma}_K^2)$ };
  \node [redBlock, above of=05, align=center] (00) {Get final distribution of $\theta$ \\[0.2cm] 
  $\theta \sim \mathcal{N}(\mu_K, \sigma_K^2)$ };
  \node [blueDiamond, above of=00, align=center] (01) {Converged?};
  \node [blueBlock, above of=01, align=center] (02) {Choose circuit parameters \\[0.2cm] $\vec{x}_k = \mathcal A(\mu_{k-1},\sigma_{k-1},f)$};
  \node [orangeBlock, above of=02, align=center] (03) {Get initial distribution of $\theta$ \\[0.2cm] $\theta \sim \mathcal N(\mu_0, \sigma_0^2)$};   
  \node [orangeBlock, above of=03, align=center] (04) {Get initial distribution of $\Pi$ \\[0.2cm] $\Pi \sim \mathcal N(\hat{\mu}_0, \hat{\sigma}_0^2)$};   
  \node [blueBlock, right of=01, align=center, xshift=2.35cm] (11) {
  Update distribution of $\theta$ \\[0.2cm] 
  $\displaystyle{p(\theta|d_k) = \frac{\mathbb{P}(d_k|\theta)p(\theta)}{\mathbb{P}(d_k)}}$
  };
  \node [blueBlock, above of=11, align=center] (12) {Run circuit and measure \\[0.2cm]
  $\displaystyle{(Q(\vec{x}_k)\ket A,P) \rightarrow d_k \in \{0,1\}}$};
  
  \draw [arrow] (01) -- node[anchor=west] {Yes} (00);
  \draw [arrow] (01) -- node[anchor=west] {No} (02);
  \draw [arrow] (03) -- node[anchor=west] {$(\mu_0, \sigma_0)$} (02);
  \draw [arrow] (04) -- node[anchor=west] {$(\hat{\mu}_0, \hat{\sigma}_0)$} (03);
  \draw [arrow] (02) -- node[anchor=south] {$\vec{x}_k$} (12);
  \draw [arrow] (11) -- node[anchor=south] {$(\mu_k,\sigma_k)$} (01);
  \draw [arrow] (00) -- node[anchor=west] {$(\mu_K,\sigma_K)$} (05);
  \draw [arrow] (12) -- node[anchor=west] {$d_k$} (11);   
\end{tikzpicture}
\caption{High-level flowchart of the algorithm for estimating $\Pi=\cosp{\theta}=\bra{A} P \ket{A}$. Here $f$ is the fidelity of the process for generating the ELF. This algorithm mainly works with $\theta$ instead of $\Pi$, and there are conversions between the distributions of $\theta$ and $\Pi$ at the beginning and end of the algorithm. The final estimate of $\Pi$ is $\hat{\mu}_K$. Note that only the ``Run circuit and measure'' step involves a quantum device.}
\label{fig:ELF_diagram}
\end{figure}

Next, we describe each component of the above algorithm in more detail. Throughout the inference process, we use a Gaussian distribution to keep track of our belief of the value of $\theta$. Namely, at each round, $\theta$ has prior distribution
\begin{align}
\label{eq:gaussianDistribution}
    p(\theta) = p(\theta ; \mu, \sigma) := \frac{1}{\sqrt{2 \pi} \sigma} \e^{-\frac{(\theta-\mu)^{2}}{2 \sigma^{2}}}
\end{align}
for some prior mean $\mu\in \mathbb R$ and prior variance $\sigma^2 \in \mathbb R^+$. After receiving the measurement outcome $d$, we compute the posterior distribution of $\theta$ by using Bayes' rule
\begin{align}
 p(\theta|d; f, \vec x)=\frac{\mathbb{P}(d|\theta; f, \vec{x})p(\theta)}{\mathbb{P}(d; f, \vec x)},
 \label{eq:bayes_update}
\end{align}
where the normalization factor, or model evidence, is defined as $\mathbb{P}(d; f, \vec x)=\int \mathbb{P}(d|\theta; f, \vec{x})p(\theta) \d \theta$ (recall that $f$ is the fidelity of the process for generating the ELF). Although the true posterior distribution will not be a Gaussian, we will approximate it as such. Following the methodology in \cite{granade2012robust}, we replace the true posterior with a Gaussian distribution of the same mean and variance \footnote{Although we can compute the mean and variance of the posterior distribution $p(\theta|d; f, \vec x)$ directly by definition, this approach is time-consuming, as it involves numerical integration. Instead, we accelerate this process by taking advantage of certain property of engineered likelihood functions. See Section \ref{sec:bi_elf_af} for more details.}, and set it as the prior of $\theta$ for the next round. We repeat this measurement-and-Bayesian-update procedure until the distribution of $\theta$ is sufficiently concentrated around a single value.

Since the algorithm mainly works with $\theta$ and we are eventually interested in $\Pi$, we need to make conversions between the estimators of $\theta$
 and $\Pi$. This is done as follows. Suppose that at round $k$ 
the prior distribution of $\theta$ is $\mathcal{N}(\mu_k, \sigma_k^2)$ and the prior distribution of $\Pi$ is $\mathcal{N}(\hat{\mu}_k, \hat{\sigma}_k^2)$ (note that $\mu_k$, $\sigma_k$, $\hat{\mu}_k$ and $\hat{\sigma}_k$ are random variables as they depend on the history of random measurement outcomes up to round $k$). The estimators of $\theta$ and $\Pi$ at this round are $\mu_k$ and $\hat{\mu}_k$, respectively. Given the distribution $\mathcal{N}(\mu_k, \sigma_k^2)$ of $\theta$, 
we compute the mean $\hat{\mu}_k$ and variance $\hat{\sigma}_k^2$ of $\cosp{\theta}$, and set $\mathcal{N}(\hat{\mu}_k, \hat{\sigma}_k^2)$ as the distribution of $\Pi$. This step can be done analytically, as if $X \sim \mathcal{N}(\mu, \sigma^2)$, then 
\begin{align}
\mathbb{E}[\cosp{X}] &= e^{-\frac{\sigma^2}{2}} \cosp{\mu}, 
\label{eq:conversion_mean}\\
\mathrm{Var}[\cosp{X}] &= \dfrac{\lrb{1-e^{-\sigma^2}} \lrb{1-e^{-\sigma^2}\cosp{2\mu}}}{2}.
\label{eq:conversion_var}
\end{align}
Conversely, given the distribution $\mathcal{N}(\hat{\mu}_k, \hat{\sigma}_k^2)$ of $\Pi$, we compute the mean $\mu_k$ and variance $\sigma_k^2$ of $\arccosp{\Pi}$ (clipping $\Pi$ to $[-1, 1]$), and set $\mathcal{N}(\mu_k, \sigma_k^2)$ as the distribution of $\theta$. This step is done numerically. Even though the $\cos$ or $\arccos$ function of a Gaussian variable is not truly Gaussian, we approximate it as such and find that this has negligible impact on the performance of the algorithm. 

Our method for tuning the circuit parameters $\vec x$ is as follows. Ideally, we want to choose them carefully so that the mean squared error (MSE) of the estimator $\mu_k$ of $\theta$ decreases as fast as possible as $k$ grows. In practice, however, it is hard to compute this quantity directly, and we must resort to a proxy of its value. The MSE of an estimator is a sum of the variance of the estimator and the squared bias of the estimator. We find that for large enough $k$, the squared bias of $\mu_k$ is smaller than its variance, i.e. $\operatorname{Bias}(\mu_k)^2=(\mathbb{E}[\mu_k] - \theta^*)^2 < \operatorname{Var}(\mu_k)$, where $\theta^*$ is the true value of $\theta$. We also find that for large enough $k$, the variance $\sigma_k^2$ of $\theta$ is often close to the variance of $\mu_k$, i.e. $\sigma_k^2 \approx \operatorname{Var}(\mu_k)$ with high probability (see Appendix \ref{sec:evidences} for evidences for these claims). Combining these facts, we know that for large enough $k$, $\operatorname{MSE}(\mu_k) =\mathbb{E}[(\mu_k - \theta^*)^2] \le 2 \sigma_k^2$ with high probability. So we will find the parameters $\vec x$ that minimize the variance $\sigma_k^2$ of $\theta$ instead.

Specifically, suppose $\theta$ has prior distribution $\mathcal N(\mu, \sigma^2)$. Upon receiving the measurement outcome $d \in \{0, 1\}$, the expected posterior variance \cite{koh2020framework} of $\theta$ is
\begin{align}
\mathbb{E}_{\mathrm d}[\operatorname{Var}(\theta | d; f, \vec x)] 
=\sigma^2\left(1-\sigma^2\frac{f^2(\partial_{\mu} b(\mu,\sigma;\vec x))^2}{1-f^2 (b(\mu,\sigma;\vec x))^2}\right),
\end{align}
where 
\begin{align}
    b(\mu,\sigma;\vec x) &= \int_{-\infty}^\infty  \ p(\theta;\mu,\sigma) \Delta(\theta;\vec x) \d \theta
    \label{eq:b}
\end{align}
in which $\Delta(\theta;\vec x)$ is the bias of the ideal likelihood function as defined in Eq.~(\ref{eq:delta_af}), and $f$ is the fidelity of the process for generating the likelihood function. We introduce an important quantity for engineering likelihood functions that we refer to as the \emph{variance reduction factor},
\begin{align}
    \mathcal{V}(\mu, \sigma;f, \vec{x}):=\frac{f^2(\partial_{\mu} b(\mu,\sigma;\vec x))^2}{1-f^2(b(\mu,\sigma;\vec x))^2}.
    \label{eq:varReductionFactor}
\end{align}
Then we have 
\begin{align}
    \mathbb{E}_{\mathrm d}[\operatorname{Var}(\theta | d; f, \vec x)] = \sigma^2 \lrbb{1- \sigma^2 \mathcal{V}(\mu, \sigma;f, \vec{x})}.
\end{align}
The larger $\mathcal{V}$ is, the faster the variance of $\theta$ decreases on average. Furthermore, to quantify the growth rate (per time step) of the inverse variance of  $\theta$, we introduce the following quantity 
\begin{align}
\label{eq:fisherrate}
R(\mu, \sigma; f, \vec{x})&:=\frac{1}{T(L)}\left(\frac{1}{\mathbb{E}_{\mathrm d}[\operatorname{Var}(\theta | d; f, \vec x)]}-\frac{1}{\sigma^2}\right) \\
 &= \frac{1}{T(L)}\frac{{\mathcal{V}}(\mu, \sigma;f, \vec{x})}{1-\sigma^2{\mathcal{V}}(\mu, \sigma;f, \vec{x})},
\end{align}
where $T(L)$ is the duration of the $L$-layer circuit in Figure \ref{fig:circuit_diagram} (recall that $\vec x \in \mathbb{R}^{2L}$). Note that $R$ is a monotonic function of $\mathcal{V}$ for $\mathcal{V}\in (0, 1)$. Therefore, when $L$ is fixed, we can maximize $R$ (with respect to $\vec{x}$) by maximizing $\mathcal{V}$. In addition, when $\sigma$ is small, $R$ is approximately proportional to $\mathcal{V}$, i.e. ${R} \approx \mathcal{V}/T(L)$. For the remainder of this work, we will assume that the ansatz circuit $A$ and its inverse $A^{\dagger}$ contribute most significantly to the duration of the overall circuit. So we will take $T(L)$ to be proportional to the number of times $A$ or $A^{\dagger}$ is invoked in the circuit, setting $T(L)=2L+1$, where time is in units of $A$'s duration.

So now we need to find the parameters $\vec x=(x_1, x_2, \dots, x_{2L}) \in \mathbb{R}^{2L}$ that maximize the variance reduction factor $\mathcal{V}(\mu, \sigma; f, \vec x)$ for given $\mu \in \mathbb{R}$, $\sigma \in \mathbb{R}^+$ and $f \in [0, 1]$. This optimization problem turns out to be difficult to solve in general. Fortunately, in practice, we may assume that the prior variance $\sigma^2$ of $\theta$ is small (e.g. at most $0.01$), and in this case, $\mathcal V(\mu, \sigma; f, \vec x)$ can be approximated by the Fisher information of the likelihood function $\mathbb{P}(d|\theta; f, \vec x)$ at $\theta=\mu$, as shown in Appendix \ref{sec:proxy_comparison}, i.e.
\begin{align}
\mathcal{V}(\mu, \sigma; f, \vec x) \approx \mathcal{I}(\mu; f, \vec x), \quad \mbox{when $\sigma$ is~small},
\label{eq:v_fi_af}
\end{align}
where 
\begin{align}
\mathcal{I}(\theta; f, \vec x) 
&= \mathbb{E}_d \lrbb{\lrb{\frac{\partial}{\partial {\theta}} \log{\mathbb{P}(d| \theta; f, \vec x)}}^2}
\\
&= \dfrac{f^2 (\Delta'(\theta; \vec x))^2}{1 - f^2 (\Delta(\theta; \vec x))^2}
\label{eq:fi_af_def}
\end{align}
is the Fisher information of the two-outcome likelihood function $\mathbb{P}(d| \theta; f, \vec x)$ as defined in Eq.~(\ref{eq:noisy_elf}). Therefore, rather than directly optimizing the variance reduction factor $\mathcal{V}(\mu, \sigma; f, \vec x)$, we optimize the Fisher information $\mathcal{I}(\mu; f, \vec x)$, which can be done efficiently by the algorithms in Section \ref{subsubsec:opt_fi_af}. Furthermore, when the fidelity $f$ of the process for generating the ELF is low, we have $\mathcal{I}(\theta; f, \vec x) \approx f^2 (\Delta'(\theta; \vec x))^2$. It follows that
\begin{align}
\mathcal{V}(\mu, \sigma; f, \vec x) &\approx f^2 (\Delta'(\mu; \vec x))^2, \quad \mbox{when $\sigma$ and $f$ are small}.  
\label{eq:v_slope_af}
\end{align}
So in this case, we can simply optimize $|\Delta'(\mu; \vec x)|$, which is proportional to the slope of the likelihood function $\mathbb{P}(d|\theta; f, \vec x)$ at $\theta=\mu$, and this task can be accomplished efficiently by the algorithms in Section \ref{subsubsec:opt_slope_af}.

Finally, we make a prediction on how fast the MSE of the estimator $\hat{\mu}_k$ of $\Pi$ decreases as $k$ grows, under the assumption that the number $L$ of circuit layers is fixed during the inference process. Note that $\operatorname{MSE}(\hat{\mu}_k) = \Theta(\frac{1}{k})$ as $k \to \infty$ in this case. Let $\theta^*$ and $\Pi^*$ be the true values of $\theta$ and $\Pi$, respectively. Then as $k \to \infty$, we have $\mu_k \to \theta^*$, $\sigma_k \to 0$, $\hat{\mu}_k \to \Pi^*$ and $\hat{\sigma}_k \to 0$ with high probability. When this event happens, we get that for large $k$,
\begin{align}
\dfrac{1}{\sigma_{k+1}^2} - \dfrac{1}{\sigma_k^2}
 \approx \mathcal{I}(\mu_k; f, \vec x_k).
\end{align}
Consequently, by Eq.~(\ref{eq:conversion_var}), we know that for large $k$,
\begin{align}
\dfrac{1}{\hat{\sigma}_{k+1}^2} - \dfrac{1}{\hat{\sigma}_k^2}
 \approx \dfrac{\mathcal{I}(\mu_k; f, \vec x_k)}{\sinpt{\mu_k}},
\end{align}
where $\mu_k \approx \arccosp{\hat{\mu}_k}$. Since $\operatorname{Bias}(\hat{\mu}_k)^2 \ll \operatorname{Var}(\hat{\mu}_k) \approx \hat{\sigma}_k^2$ for large $k$, we predict that 
\begin{align}
\operatorname{MSE}(\hat{\mu}_k) \approx \dfrac{ 1-\hat{\mu}_k^2}{k\mathcal{I}(\arccosp{\hat{\mu}_k}; f, \vec x_k)}.
\end{align}
This means that the asymptotic growth rate (per time step) of the inverse MSE of $\hat{\mu}_k$ should be roughly 
\begin{align}
\hat{R}_0(\Pi^*; f, \vec x) := \dfrac{\mathcal{I}(\arccosp{\Pi^*}; f, \vec x)}{(2L+1)(1-(\Pi^*)^2)},
\label{eq:r0_factor}
\end{align}
where $\vec x \in \mathbb{R}^{2L}$ is chosen such that $\mathcal{I}(\arccosp{\Pi^*}; f, \vec x)$ is maximized. We will compare this rate with the empirical growth rate (per time step) of the inverse MSE of $\hat{\mu}_k$ in Section \ref{sec:simulatin_results}.

\section{Efficient heuristic algorithms for circuit parameter tuning and Bayesian inference}
\label{sec:optalg}

In this section, we present heuristic algorithms for tuning the parameters $\vec x$ of the circuit in Figure~\ref{fig:circuit_diagram} and describe how to efficiently carry out Bayesian inference with the resultant likelihood functions. 

\subsection{Efficient maximization of proxies of the variance reduction factor}
\label{subsec:proxyopt}
Our algorithms for tuning the circuit parameters $\vec x$ are based on maximizing
two proxies of the variance reduction factor $\mathcal{V}$ -- the Fisher information and slope of the likelihood function $\mathbb{P}(d|\theta; f, \vec x)$. All of these algorithms require efficient procedures for evaluating the CSBD coefficient functions of $\Delta(\theta; \vec x)$ and $\Delta'(\theta; \vec x)$ with respect to $x_j$ for $j=1,2,\dots, 2L$. Recall that we have shown in Section \ref{subsec:qc_elf} that $\Delta(\theta;\vec x)$ is trigono-multiquadratic in $\vec x$. Namely, for any $j \in \{1,2,\dots, 2L\}$, there exist functions $C_j(\theta; \vec x_{\neg j})$, $S_j(\theta; \vec x_{\neg j})$ and $B_j(\theta; \vec x_{\neg j})$ of $\vec x_{\neg j}:=(x_1, \dots, x_{j-1}, x_{j+1}, \dots, x_{2L})$ such that
\begin{align}
\Delta(\theta; \vec x) &= C_j(\theta; \vec x_{\neg j}) \cosp{2 x_j} + S_j(\theta; \vec x_{\neg j}) \sinp{2 x_j} \nonumber \\
&\quad + B_j(\theta; \vec x_{\neg j}).
\end{align}
It follows that
\begin{align}
\Delta'(\theta; \vec x) &= C_j'(\theta; \vec x_{\neg j}) \cosp{2 x_j} + S_j'(\theta; \vec x_{\neg j}) \sinp{2 x_j} \nonumber \\
&\quad + B'_j(\theta; \vec x_{\neg j})
\end{align}
is also trigono-multiquadratic in $\vec x$, where $C'_j(\theta;\vec x_{\neg j}) = \partial_{\theta} C_j(\theta; \vec x_{\neg j})$, $S'_j(\theta;\vec x_{\neg j}) = \partial_{\theta} S_j(\theta; \vec x_{\neg j})$, $B'_j(\theta;\vec x_{\neg j}) = \partial_{\theta} B_j(\theta; \vec x_{\neg j})$ are the partial derivatives of $C_j(\theta; \vec x_{\neg j})$, $S_j(\theta; \vec x_{\neg j})$, $B_j(\theta; \vec x_{\neg j})$
with respect to $\theta$, respectively. It turns out that given $\theta$ and $\vec x_{\neg j}$, $C_j(\theta; \vec x_{\neg j})$, each of $S_j(\theta; \vec x_{\neg j})$, $B_j(\theta; \vec x_{\neg j})$,
$C'_j(\theta; \vec x_{\neg j})$, $S'_j(\theta; \vec x_{\neg j})$ 
and $B'_j(\theta; \vec x_{\neg j})$ can be computed in $O(L)$ time.

\begin{lemma}
Given $\theta$ and $\vec x_{\neg j}$, each of
$C_j(\theta; \vec x_{\neg j})$, $S_j(\theta; \vec x_{\neg j})$, 
$B_j(\theta; \vec x_{\neg j})$,
$C'_j(\theta; \vec x_{\neg j})$, $S'_j(\theta; \vec x_{\neg j})$ 
and $B'_j(\theta; \vec x_{\neg j})$ can be computed in $O(L)$ time.
\label{lem:eval_coeff_func_af}
\end{lemma}

\begin{proof}
See Appendix \ref{sec:eval_coeff_func_af}.
\end{proof}

\subsubsection{Maximizing the Fisher information of the likelihood function}
\label{subsubsec:opt_fi_af}

We propose two algorithms for maximizing the Fisher information of the likelihood function $\mathbb{P}(d|\theta; f, \vec x)$ at a given point $\theta=\mu$ (i.e. the prior mean of $\theta$). Namely, our goal is to find $\vec x \in \mathbb{R}^{2L}$ that maximizes
\begin{align}
\mathcal{I}(\mu; f, \vec x) = \dfrac{f^2 (\Delta'(\mu; \vec x))^2}{1 - f^2 \Delta(\mu; \vec x)^2}.
\end{align}

The first algorithm is based on \emph{gradient ascent}. Namely, it starts with a random initial point, and keeps taking steps proportional to the gradient of $\mathcal{I}$ at the current point, until a convergence criterion is satisfied. Specifically, let $\vec x^{(t)}$ be the parameter vector at iteration $t$. We update it as follows:
\begin{align}
\vec x^{(t+1)} = \vec x^{(t)} + \delta(t) \nabla{\mathcal{I}(\mu; f, \vec x)}|_{\vec x=\vec x^{(t)}}.
\end{align}
where $\delta: \mathbb{Z}^{\ge 0} \to \mathbb{R}^+$ is the step size schedule\footnote{In the simplest case, $\delta(t)=\delta$ is constant. But in order to achieve better performance, we might want $\delta(t) \to 0$ as $t \to \infty$.}. This requires the calculation of the partial derivative of $\mathcal{I}(\mu; f, \vec x)$ with respect to each $x_j$, which can be done as follows. We first use the procedures in Lemma \ref{lem:eval_coeff_func_af}  to compute $C_j:=C_j(\mu; \vec x_{\neg j})$, $S_j:=S_j(\mu; \vec x_{\neg j})$, $B_j:=B_j(\mu; \vec x_{\neg j})$, $C'_j:=C'_j(\mu; \vec x_{\neg j})$, $S'_j:=S'_j(\mu; \vec x_{\neg j})$ and $B'_j:=B'_j(\mu; \vec x_{\neg j})$ for each $j$. Then we get
\begin{align}
\Delta &:= \Delta(\mu; \vec x) \\
&= C_j \cosp{2x_j} + S_j \sinp{2x_j} + B_j, \\
\Delta' &:= \Delta'(\mu; \vec x) \\
& = C'_j \cosp{2x_j} + S'_j \sinp{2x_j} + B'_j, \\
\chi_j &:= \dfrac{\partial \Delta(\mu; \vec x)}{\partial x_j} \\
&= 2 \lrbb{-C_j \sinp{2 x_j} + S_j \cosp{2x_j}},  \\
\chi'_j &:= \dfrac{\partial \Delta'(\mu; \vec x)}{\partial x_j} \\
&= 2 \lrbb{-C'_j \sinp{2 x_j} + S'_j \cosp{2x_j}}.
\end{align}
Knowing these quantities, we can compute the partial derivative of $\mathcal{I}(\mu; f, \vec x)$ with respect to $x_j$ as follows:
\begin{align}
\gamma_j &:= {\dfrac{\partial \mathcal{I}(\mu; f, \vec x)}{\partial x_j}} \\
&= \dfrac{2f^2 \lrbb{(1-f^2 \Delta^2) \Delta' \chi'_j
+f^2 \Delta \chi_j (\Delta')^2}
}{\lrb{1-f^2 \Delta^2}^2}.
\end{align}
Repeat this procedure for $j=1,2,\dots, 2L$. Then we obtain $\nabla{\mathcal{I}(\mu; f, \vec x)}=(\gamma_1, \gamma_2, \dots, \gamma_{2L})$. Each iteration of the algorithm takes $O(L^2)$ time. The number of iterations in the algorithm depends on the initial point, the termination criterion and the step size schedule $\delta$. See Algorithm \ref{alg:ga_opt_v0_af} for more details.

The second algorithm is based on \emph{coordinate ascent}. Unlike gradient ascent, this algorithm does not require step sizes, and allows each variable to change dramatically in a single step. As a consequence, it may converge faster than the previous algorithm. Specifically, this algorithm starts with a random initial point, and successively maximizes the objective function $\mathcal{I}(\mu; f, \vec x)$ along coordinate directions, until a convergence criterion is satisfied. At the $j$-th step of each round, it solves the following single-variable optimization problem for a coordinate $x_j$:
\begin{align}
\argmax_{z} \dfrac{f^2\lrb{C'_j \cosp{2z} + S'_j \sinp{2z} + B'_j
}^2}{1-f^2\lrb{C_j \cosp{2z}+ S_j \sinp{2z} + B_j}^2},
\end{align}
where $C_j= C_j(\mu; \vec x_{\neg j})$, $S_j = S_j(\mu; \vec x_{\neg j})$, $B_j= B_j(\mu; \vec x_{\neg j})$, $C'_j= C'_j(\mu; \vec x_{\neg j})$, $S'_j = S'_j(\mu; \vec x_{\neg j})$, $B'_j= B'_j(\mu; \vec x_{\neg j})$ can be computed in $O(L)$ time by the procedures in Lemma \ref{lem:eval_coeff_func_af}. This single-variable optimization problem can be tackled by standard gradient-based methods, and we set $x_j$ to be its solution. Repeat this procedure for $j=1,2,\dots, 2L$. This algorithm produces a sequence $\vec x^{(0)}$, $\vec x^{(1)}$, $\vec x^{(2)}$, $\dots$, such that $\mathcal{I}(\mu; f, \vec x^{(0)}) 
\le \mathcal{I}(\mu; f, \vec x^{(1)}) \le \mathcal{I}(\mu; f, \vec x^{(2)}) \le \dots$. Namely, the value of $\mathcal{I}(\mu; f, \vec x^{(t)})$ increases monotonically as $t$ grows. Each round of the algorithm takes $O(L^2)$ time. The number of rounds in the algorithm depends on the initial point and the termination criterion. See Algorithm \ref{alg:ca_opt_v0_af} for more details.

\begin{algorithm*}[ht]
 \KwIn{The prior mean $\mu$ of $\theta$, the number $L$ of circuit layers, the fidelity $f$ of the process for generating the ELF, the step size schedule $\delta: \mathbb{Z}^{\ge 0} \to \mathbb{R}^+$, the error tolerance $\epsilon$ for termination.}
 \KwOut{A set of parameters $\vec x = (x_1, x_2, \dots, x_{2L}) \in \mathbb{R}^{2L}$ that are a local maximum point of the function $\mathcal{I}(\mu; f, \vec x)$.}

Choose random initial point $\vec x^{(0)}=(x^{(0)}_1, x^{(0)}_2, \dots, x^{(0)}_{2L}) \in (-\pi, \pi]^{2L}$; \\
$t \leftarrow 0$; \\
\While{True}{

\For{$j \leftarrow 1$ \KwTo $2L$}{
    Let $\vec x^{(t)}_{\neg j} = (x^{(t)}_1, \dots, x^{(t)}_{j-1}, x^{(t)}_{j+1}, \dots, x^{(t)}_{2L})$; \\
    Compute $C^{(t)}_j:= C_j(\mu; \vec x^{(t)}_{\neg j})$, $S^{(t)}_j:= S_j(\mu; \vec x^{(t)}_{\neg j})$, $B^{(t)}_j:= B_j(\mu; \vec x^{(t)}_{\neg j})$, $C'^{(t)}_j:= C'_j(\mu; \vec x^{(t)}_{\neg j})$, $S'^{(t)}_j:= S'_j(\mu; \vec x^{(t)}_{\neg j})$, $B'^{(t)}_j:= B'_j(\mu; \vec x^{(t)}_{\neg j})$ by using the procedures in Lemma \ref{lem:eval_coeff_func_af}; \\
    Compute $\Delta(\mu; \vec x)$, $\Delta'(\mu; \vec x)$ and their partial derivatives with respect to $x_j$ at $\vec x=\vec x^{(t)}$ as follows:
    \begin{align}
    &\Delta^{(t)} := \Delta(\mu; \vec x^{(t)}) = C^{(t)}_j \cosp{2x_j} + S^{(t)}_j \sinp{2x_j} + B^{(t)}_j, \\
    &\Delta'^{(t)} := \Delta'(\mu; \vec x^{(t)}) = C'^{(t)}_j \cosp{2x_j} + S'^{(t)}_j \sinp{2x_j} + B'^{(t)}_j, \\
    &\chi^{(t)}_j := \dfrac{\partial \Delta(\mu; \vec x)}{\partial x_j}|_{\vec x = \vec x^{(t)}}
    = 2 \lrb{-C^{(t)}_j \sinp{2x_j} + S^{(t)}_j \cosp{2x_j}},  \\
    &\chi'^{(t)}_j := \dfrac{\partial \Delta'(\mu; \vec x)}{\partial x_j}|_{\vec x = \vec x^{(t)}}
    = 2 \lrb{-C'^{(t)}_j \sinp{2x_j} + S'^{(t)}_j \cosp{2x_j}}; 
    \end{align}
    Compute the partial derivative of $\mathcal{I}(\mu; f, \vec x)$ with respect to $x_j$ at $\vec x=\vec x^{(t)}$ as follows:
    \begin{align}
        \gamma^{(t)}_j := {\dfrac{\partial \mathcal{I}(\mu; f, \vec x)}{\partial x_j}}|_{\vec x=\vec x^{(t)}} = \dfrac{2f^2 \lrbb{(1-f^2 (\Delta^{(t)})^2) \Delta'^{(t)} \chi'^{(t)}_j
        +f^2 \Delta^{(t)} \chi^{(t)}_j (\Delta'^{(t)})^2}
        }{\lrbb{1-f^2 (\Delta^{(t)})^2}^2}
    \end{align}
} 
Set $\vec x^{(t+1)}=\vec x^{(t)} + \delta(t) \nabla{\mathcal{I}(\mu; f, \vec x)} |_{\vec x=\vec x^{(t)}}$, where
$\nabla{\mathcal{I}(\mu; f, \vec x)}|_{\vec x=\vec x^{(t)}}=(\gamma^{(t)}_1, \gamma^{(t)}_2, \dots, \gamma^{(t)}_{2L})$;\\
\If{$|\mathcal{I}(\mu; f, \vec x^{(t+1)}) - \mathcal{I}(\mu; f, \vec x^{(t)})| < \epsilon$}{break;}
$t \leftarrow t+1$;
}
Return $\vec x^{(t+1)}=(x_1^{(t+1)}, x_2^{(t+1)}, \dots, x_{2L}^{(t+1)})$ as the optimal parameters.
\caption{Gradient ascent for Fisher information maximization in the ancilla-free case}
\label{alg:ga_opt_v0_af}
\end{algorithm*}

\begin{algorithm*}[ht]
 \KwIn{The prior mean $\mu$ of $\theta$, the number $L$ of circuit layers, the fidelity $f$ of the process for generating the ELF, the error tolerance $\epsilon$ for termination.}
 \KwOut{A set of parameters $\vec x = (x_1, x_2, \dots, x_{2L}) \in (-\pi, \pi]^{2L}$ that are a local maximum point of the function $\mathcal{I}(\mu; f, \vec x)$.}

Choose random initial point $\vec x^{(0)}=(x^{(0)}_1, x^{(0)}_2, \dots, x^{(0)}_{2L}) \in (-\pi, \pi]^{2L}$; \\
$t \leftarrow 1$; \\
\While{True}{

\For{$j \leftarrow 1$ \KwTo $2L$}{
    Let $\vec x^{(t)}_{\neg j} = (x^{(t)}_1, \dots, x^{(t)}_{j-1}, x^{(t-1)}_{j+1}, \dots, x^{(t-1)}_{2L})$; \\
    Compute $C^{(t)}_j:= C_j(\mu; \vec x^{(t)}_{\neg j})$, $S^{(t)}_j:= S_j(\mu; \vec x^{(t)}_{\neg j})$, $B^{(t)}_j:= B_j(\mu; \vec x^{(t)}_{\neg j})$, $C'^{(t)}_j:= C'_j(\mu; \vec x^{(t)}_{\neg j})$, $S'^{(t)}_j:= S'_j(\mu; \vec x^{(t)}_{\neg j})$, $B'^{(t)}_j:= B'_j(\mu; \vec x^{(t)}_{\neg j})$ by using the procedures in Lemma \ref{lem:eval_coeff_func_af}; \\
    Solve the single-variable optimization problem 
    $$\argmax_{z} \dfrac{f^2\lrb{C'^{(t)}_j \cosp{2z} + S'^{(t)}_j \sinp{2z} + B'^{(t)}_j
}^2}{1-f^2\lrb{C^{(t)}_j \cosp{2z}+ S^{(t)}_j \sinp{2z} + B^{(t)}_j}^2}$$
    by standard gradient-based methods and set $x^{(t)}_j$ to be its solution;
} 
\If{$|\mathcal{I}(\mu; f, \vec x^{(t)}) - \mathcal{I}(\mu; f, \vec x^{(t-1)})| < \epsilon$}{break;}
$t \leftarrow t+1$;
}
Return $\vec x^{(t)}=(x_1^{(t)}, x_2^{(t)}, \dots, x_{2L}^{(t)})$ as the optimal parameters.

\caption{Coordinate ascent for Fisher information maximization in the ancilla-free case}
\label{alg:ca_opt_v0_af}
\end{algorithm*}

We have used Algorithms \ref{alg:ga_opt_v0_af} and \ref{alg:ca_opt_v0_af} to find the parameters $\vec x_{\theta} \in \mathbb{R}^{2L}$ that maximize $\mathcal{I}(\theta; f, \vec x)$ \footnote{Since Algorithms \ref{alg:ga_opt_v0_af} and \ref{alg:ca_opt_v0_af} only output a local maximum point of $\mathcal{I}(\theta; f, \vec x)$ for given $\theta$ and $f$, we need to run them multiple times with random initial points to find a global maximum point of the same function. We find that this does not require many trials. For example, for $L \le 6$, a global-optimal solution can be found in $10$ trials with high probability. Similar statements hold for the other algorithms for parameter tuning in this paper.} for various $\theta \in (0, \pi)$ (fixing $f$) and obtained Figure \ref{fig:fisher_information_af}. This figure indicates that the Fisher information of ELF is larger than that of CLF for the majority of $\theta \in (0, \pi)$. Consequently, the estimation algorithm based on ELF is more efficient than the one based on CLF, as will be demonstrated in Section \ref{sec:simulatin_results}.

\begin{figure}[!ht]
\includegraphics[width=0.9\linewidth]{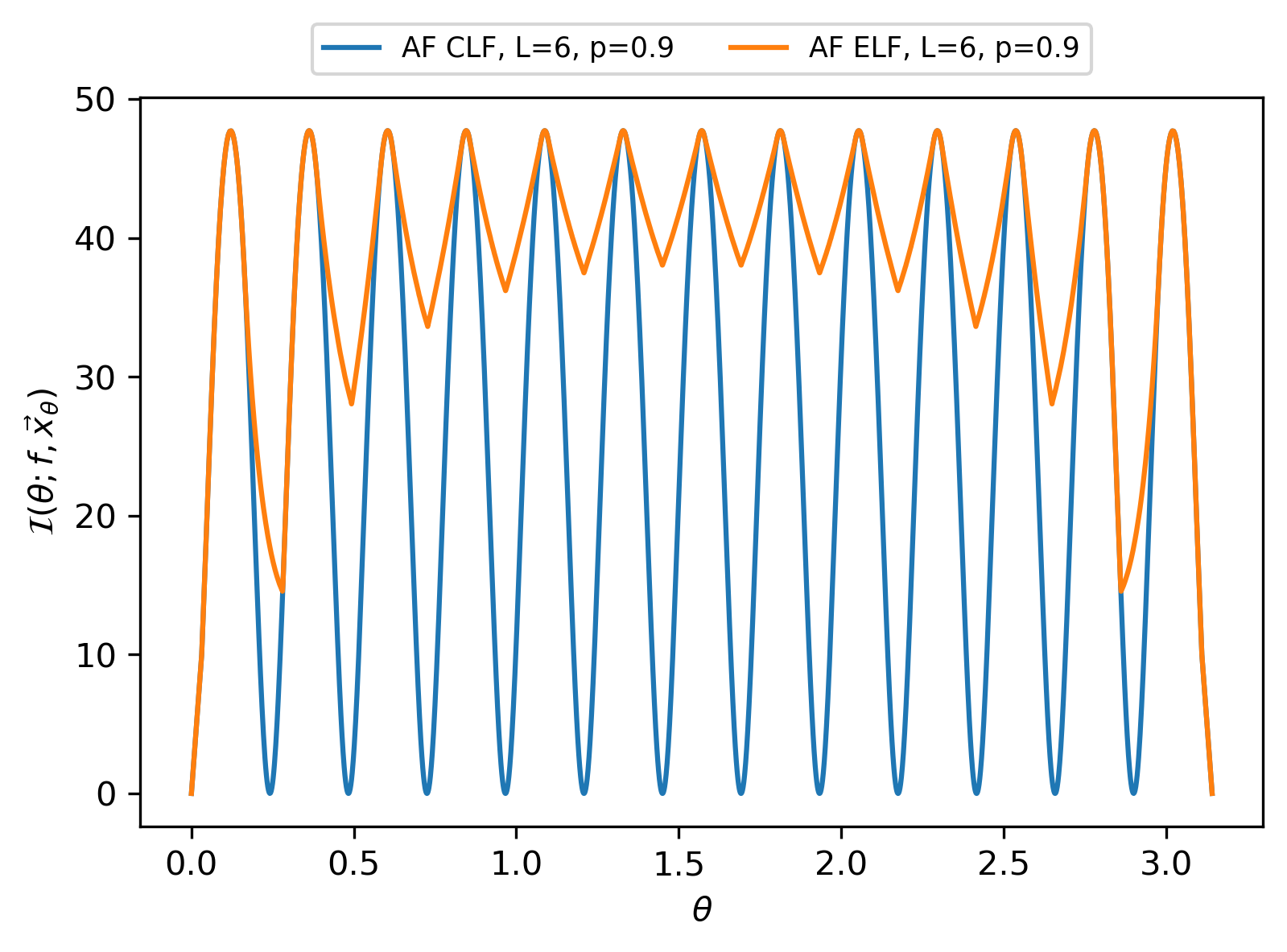}
\caption{This figure compares the Fisher information of ELF and CLF for various $\theta \in (0, \pi)$, when the number of circuit layers is $L=6$, the fidelity of each layer is $p=0.9$, and there is no SPAM error (i.e. $\bar{p}=1)$. For ELF, $\vec x_{\theta}$ is a global maximum point of $\mathcal{I}(\theta; f, \vec x)$ for given $\theta$ and $f=\bar{p}p^L=0.531441$. For CLF, $\vec x_{\theta}=(\pi/2, \pi/2, \dots, \pi/2)$ is fixed. One can see that the Fisher information of ELF is larger than that of CLF for the majority of $\theta \in (0, \pi)$. Furthermore, the Fisher information of CLF changes dramatically for different $\theta$'s (in fact, it is exactly $0$ when $\theta = j \pi/(2L+1)$ for $j=0, 1, \dots, 2L+1$), whereas the Fisher information of ELF is less sensitive to the value of $\theta$.}
\label{fig:fisher_information_af}
\end{figure}

\subsubsection{Maximizing the slope of the likelihood function}
\label{subsubsec:opt_slope_af}
We also propose two algorithms for maximizing the slope of the likelihood function $\mathbb{P}(d|\theta; f, \vec x)$ at a given point $\theta=\mu$ (i.e. the prior mean of $\theta$). Namely, our goal is to find $\vec x \in \mathbb{R}^{2L}$ that maximizes $|\mathbb{P}'(\mu; f, \vec x)| = f |\Delta'(\mu; \vec x)|/2$. These algorithms are similar to Algorithms \ref{alg:ga_opt_v0_af} and \ref{alg:ca_opt_v0_af} for Fisher information maximization, in the sense that they are also based on gradient ascent and coordinate ascent, respectively. They are formally described in Algorithms \ref{alg:ga_opt_slope_af} and \ref{alg:ca_opt_slope_af} in Appendix \ref{sec:alg_opt_slope_af}, respectively. We have used them to find the parameters $\vec x_{\theta} \in \mathbb{R}^{2L}$ that maximize $|\Delta'(\theta; \vec x)|$ for various $\theta \in (0, \pi)$ and obtained Figure \ref{fig:slope_af}. This figure implies that the slope-based ELF is steeper than CLF for the majority of $\theta \in (0, \pi)$ and hence has more statistical power than CLF (at least) in the low-fidelity setting by Eq.~\eqref{eq:v_slope_af}.

\begin{figure}[!ht]
\includegraphics[width=0.9\linewidth]{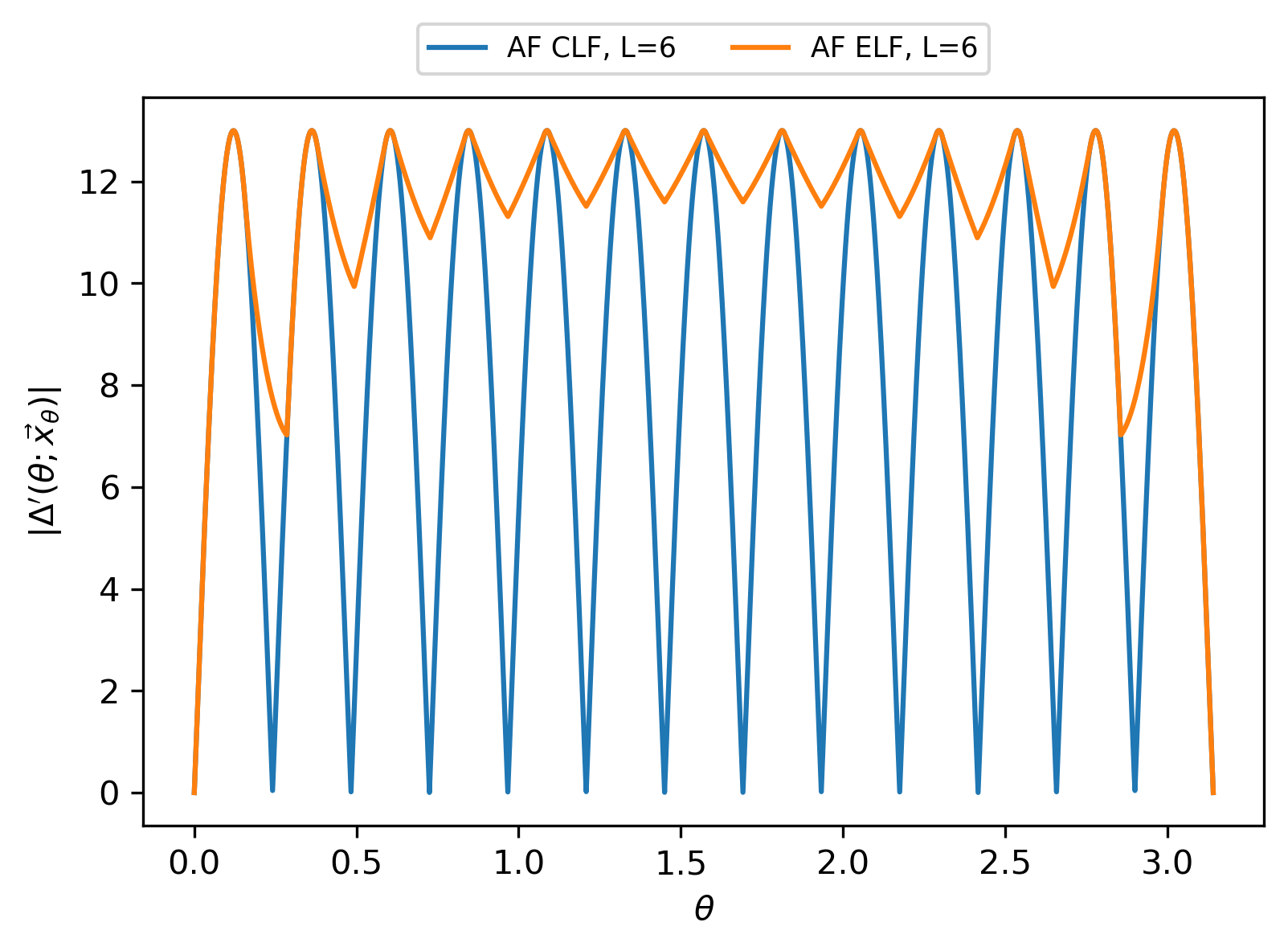}
\caption{This figure compares the values of $|\Delta'(\theta; \vec x_{\theta})|$ for slope-based ELF and CLF for various $\theta \in (0, \pi)$, when the number of circuit layers is $L=6$. For ELF, $\vec x_{\theta}$ is a global maximum point of $|\Delta'(\theta; \vec x)|$ for given $\theta$. For CLF, $\vec x_{\theta}=(\pi/2, \pi/2, \dots, \pi/2)$ is fixed. This figure implies that the slope-based ELF is steeper than CLF for the majority of $\theta \in (0, \pi)$. Furthermore, the slope of CLF changes dramatically for different $\theta$'s (in fact, it is exactly $0$ when $\theta = j \pi/(2L+1)$ for $j=0, 1, \dots, 2L+1$), whereas the slope of ELF is less sensitive to the value of $\theta$.}
\label{fig:slope_af}
\end{figure}

\subsection{Approximate Bayesian inference with engineered likelihood functions}
\label{sec:bi_elf_af}
With the algorithms for tuning the circuit parameters $\vec x$ in place, we now describe how to efficiently carry out Bayesian inference with the resultant likelihood functions. In principle, we can compute the posterior mean and variance of $\theta$ directly after receiving a measurement outcome $d$. But this approach is time-consuming, as it involves numerical integration. By taking advantage of certain property of the engineered likelihood functions, we can greatly accelerate this process.

Suppose $\theta$ has prior distribution $\mathcal{N}(\mu, \sigma^2)$, where $\sigma \ll 1/L$, and the fidelity of the process for generating the ELF is $f$. We find that the parameters $\vec x=(x_1, x_2, \dots, x_{2L})$ that maximize $\mathcal{I}(\mu; f, \vec x)$ (or $|\Delta'(\mu; \vec x)|$) satisfy the following property: When $\theta$ is close to $\mu$, i.e. $\theta \in [\mu-O(\sigma), \mu+O(\sigma)]$, we have
\begin{align}
\mathbb{P}(d|\theta; f, \vec x) \approx \dfrac{1+(-1)^d f \sinp{r\theta+b}}{2}    
\label{eq:lfapprox}
\end{align}
for some $r, b \in \mathbb{R}$. Namely, $\Delta(\theta; \vec x)$ can be approximated by a sinusoidal function in this region of $\theta$. Figure~\ref{fig:fit_lf_af} illustrates one such example.

We can find the best-fitting $r$ and $b$ by solving the following least squares problem:
\begin{align}
(r^*, b^*) = \argmin_{r, b} \sum_{\theta \in \Theta}  \left|\arcsinp{\Delta(\theta;\vec x)}-r\theta-b \right|^2,
\label{eq:lstsq_af}
\end{align}
where $\Theta=\{\theta_1, \theta_2, \dots, \theta_k\} \subseteq [\mu-O(\sigma), \mu+O(\sigma)]$. This least-squares problem has the following analytical solution:
\begin{align}
\begin{pmatrix}
r^*\\
b^*
\end{pmatrix}
= A^+ z
= (A^TA)^{-1} A^T z,
\end{align}
where 
\begin{align}
A =
\begin{pmatrix}
\theta_1 &  1 \\
\theta_2 &  1 \\
\vdots & \vdots\\
\theta_k & 1
\end{pmatrix}, 
\quad
z = \begin{pmatrix}
\arcsinp{\Delta(\theta_1;\vec x)} \\
\arcsinp{\Delta(\theta_2;\vec x)} \\
\vdots \\
\arcsinp{\Delta(\theta_k;\vec x)}
\end{pmatrix}.
\end{align}
Figure~\ref{fig:fit_lf_af} demonstrates an example of the true and fitted likelihood functions.

Once we obtain the optimal $r$ and $b$, we can approximate the posterior mean and variance of $\theta$ by the ones for 
\begin{align}
\mathbb{P}(d|\theta; f) = \dfrac{1+(-1)^d f \sinp{r\theta+b}}{2},
\end{align}
which have analytical formulas. Specifically, suppose $\theta$ has prior distribution $\mathcal{N}(\mu_k, \sigma_k^2)$ at round $k$. Let $d_k$ be the measurement outcome and $(r_k, b_k)$ be the best-fitting parameters at this round. Then we approximate the posterior mean and variance of $\theta$ by
\begin{widetext}
\begin{align}
\mu_{k+1} & = 
\mu_k+\dfrac{(-1)^{d_k} f e^{-r_k^2\sigma_k^2/2}r_k \sigma_k^2\cosp{r_k\mu_k+b_k}}{1+(-1)^{d_k} f  e^{-r_k^2\sigma_k^2/2}\sinp{r_k\mu_k+b_k}},\label{eq:bayes_update_mu_af} \\
\sigma_{k+1}^2 & =  
\sigma_k^2 \left (1-\dfrac{f r_k^2 \sigma_k^2 e^{-r_k^2\sigma_k^2/2} [f e^{-r_k^2\sigma_k^2/2}+(-1)^{d_k}\sinp{r_k\mu_k+b_k}]}{[1+(-1)^{d_k} f e^{-r_k^2\sigma_k^2/2}\sinp{r_k\mu_k+b_k}]^2}\right ).
\label{eq:bayes_update_sigma_af}
\end{align}
\end{widetext}
After that, we proceed to the next round, setting $\mathcal{N}(\mu_{k+1}, \sigma_{k+1}^2)$ as the prior distribution of $\theta$ for that round.

\begin{figure}[!ht]
\includegraphics[width=0.9\linewidth]{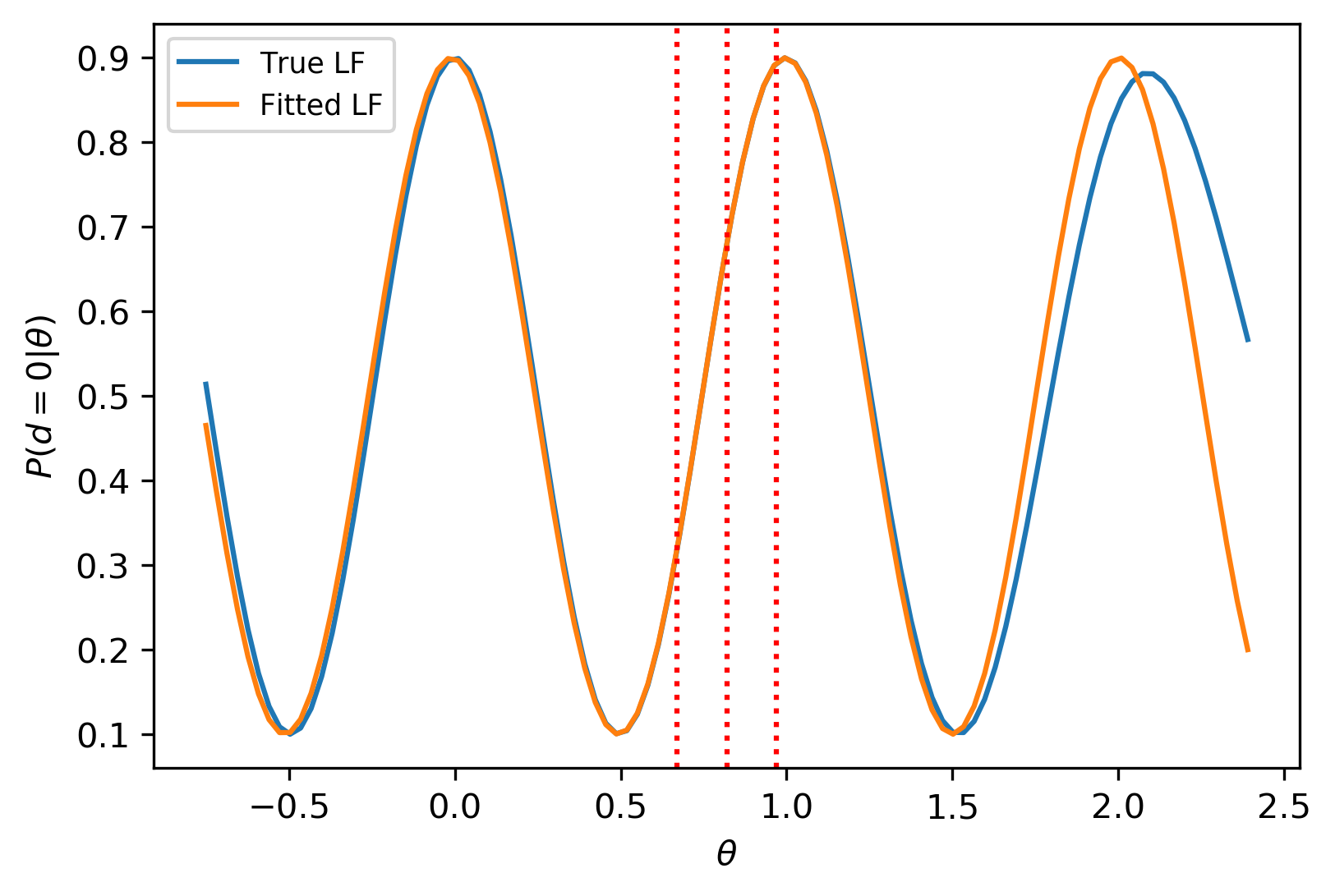}
\caption{The true and fitted likelihood functions when $L=3$, $f=0.8$, and $\theta$ has prior distribution $\mathcal{N}(0.82, 0.0009)$. The true likelihood function is generated by Algorithm \ref{alg:ca_opt_v0_af}. During the sinusoidal fitting of this function, we
set $\Theta=\{\mu-\sigma, \mu-0.8\sigma, \dots, \mu+0.8\sigma, \mu+\sigma\}$ (i.e. $\Theta$ contains $11$ uniformly distributed points in $[\mu-\sigma, \mu+\sigma]$) in Eq.~(\ref{eq:lstsq_af}). The fitted likelihood function is $\mathbb{P}(d|\theta)=(1+(-1)^d f\sinp{r \theta + b })/2$, where $r=6.24$ and $b=-4.65$. Note that the true and fitted likelihood functions are close for $\theta \in [0.67, 0.97]$. }
\label{fig:fit_lf_af}
\end{figure}

Note that, as Figure~\ref{fig:fit_lf_af} illustrates, the difference between the true and fitted likelihood functions can be large when $\theta$ is far from $\mu$, i.e. $|\theta-\mu|\gg \sigma$. But since the prior distribution $p(\theta) = \frac{1}{\sqrt{2\pi}\sigma}e^{-\frac{(\theta-\mu)^2}{2 \sigma^2}}$ decays exponentially in $|\theta-\mu|$, such $\theta$'s have little contribution to the computation of posterior mean and variance of $\theta$. So Eqs.~(\ref{eq:bayes_update_mu_af}) and (\ref{eq:bayes_update_sigma_af}) give highly accurate estimates of the posterior mean and variance of $\theta$, and their errors have negligible impact on the performance of the whole algorithm. 

\section{Simulation results}
\label{sec:simulatin_results}

In this section, we present the simulation results of Bayesian inference with engineered likelihood functions for amplitude estimation. These results demonstrate the advantages of engineered likelihood functions over unengineered ones, as well as the impacts of circuit depth and fidelity on their performance. 

\subsection{Experimental details}
\label{subsec:experimental_details}

In our experiments, we assume that the ansatz circuit $A$ and its inverse $A^{\dagger}$ contribute most significantly to the duration of the circuit (in Figure \ref{fig:circuit_diagram} or \ref{fig:elf_circuit_ab}) for generating the likelihood function. So when the number of circuit layer is $L$, the time cost of an inference round is roughly $2L+1$, where time is in units of $A$'s duration. Moreover, we assume that there is no SPAM error, i.e. $\bar{p}=1$, in the experiments.

Suppose we aim to estimate the expectation value $\Pi=\cosp{\theta}=\bra{A}P\ket{A}$. Let $\hat{\mu}_t$ be the estimator of $\Pi$ at time $t$. Note that $\hat{\mu}_t$ is a random variable, since it depends on the history of random measurement outcomes up to time $t$. We measure the performance of a scheme by the root-mean-squared error (RMSE) of $\hat{\mu}_t$, that is given by
\begin{align}
\operatorname{RMSE}_t := \sqrt{\operatorname{MSE}_t} = \sqrt{\mathbb{E}[(\hat{\mu}_t - \Pi)^2]}.    
\end{align}
We will investigate how fast $\operatorname{RMSE}_t$ decreases as $t$ grows for various schemes, including the ancilla-based Chebyshev likelihood function (AB CLF), ancilla-based engineered likelihood function (AB ELF), ancilla-free Chebyshev likelihood function (AF CLF), and ancilla-free engineered likelihood function (AF ELF).

In general, the distribution of $\hat{\mu}_t$ is difficult to characterize, and there is no analytical formula for $\operatorname{RMSE}_t$. To estimate this quantity, we simulate the inference process $M$ times, and collect $M$ samples $\hat{\mu}^{(1)}_t$, $\hat{\mu}^{(2)}_t$, $\dots$, $\hat{\mu}^{(M)}_t$ of $\hat{\mu}_t$, where $\hat{\mu}^{(i)}_t$ is the estimate of $\Pi$ at time $t$ in the $i$-th run, for $i=1,2,\dots,M$. Then we use the quantity
\begin{align}
\overline{\operatorname{RMSE}}_t := \sqrt{\dfrac{1}{M}\sum_{i=1}^M (\hat{\mu}^{(i)}_t - \Pi)^2}.
\end{align}
to approximate the true $\operatorname{RMSE}_t$. In our experiments, we set $M=300$ and find that this leads to a low enough empirical variance in the estimate to yield meaningful results.

In each experiment, we set the true value of $\Pi$ and choose a prior distribution of $\Pi$, and run the algorithm in Figure \ref{fig:ELF_diagram} to test its performance for estimating this quantity. The only quantum component of this algorithm, i.e. executing the circuit in Figure \ref{fig:circuit_diagram} and measuring the outcome $d \in \{0, 1\}$, is simulated by an efficient classical procedure, since the probability distribution of the outcome $d$ is characterized by Eq. \eqref{eq:noisy_elf} or \eqref{eq:noisy_elf_ab}, depending on whether the scheme is ancilla-free or ancilla-based. Namely, we synthesize the outcome $d$ by a classical pseudo-random number generator (based on Eq. \eqref{eq:noisy_elf} or \eqref{eq:noisy_elf_ab}) instead of simulating the noisy quantum circuit. This not only greatly accelerates our simulation, but also makes our results applicable to a wide range of scenarios (i.e. they are not specific to certain ansatz circuit $A$ or observable $P$). 

For engineering the likelihood function, we use Algorithm \ref{alg:ca_opt_v0_af} or \ref{alg:ca_opt_v0_ab} to optimize the circuit parameters $\vec x$, depending on whether the scheme is ancilla-free or ancilla-based. We find that Algorithms 1 and 2 generate the same likelihood function (as they both find the optimal parameters within a reasonable number of trials), and the same holds for Algorithms 5 and 6. So the simulation results in  Section \ref{sec:simulatin_results} will not change if we have used Algorithms 1 and 5 instead, as the AF ELF and AB ELF will remain intact. Meanwhile, we do not use Algorithm 3, 4, 7 or 8 in our experiments, because these algorithms are used to maximize the slope of the AF/AB likelihood function, which is a good proxy of the variance reduction factor $\mathcal{V}$ only when the fidelity $f$ is close to zero (see Eqs.~\eqref{eq:v_slope_af} and \eqref{eq:v_slope_ab}). In our experiments, $f$ is mostly between $0.5$ and $0.9$, and in such case, the slope is not a good proxy of the variance reduction factor and we need to maximize the Fisher information by Algorithms 1, 2, 5 or 6 instead. 

Our estimation algorithm needs to tune the circuit parameters at each round, and since it can have tens of thousands of rounds, this component can be quite time-consuming. We use the following method to mitigate this issue. Given the number $L$ of circuit layers and the fidelity $f$ of the process for generating the ELF, we first construct a ``lookup table" which consists of the optimal parameters $\vec x_{\Pi}$ for all $\Pi$'s in a discrete set $\mathcal{S}=\{\Pi_1, \Pi_2, \dots, \Pi_N \} \subset [-1, 1]$. Then at any stage of the estimation algorithm, we find the $\Pi_j \in \mathcal{S}$ that is closest to the current estimate of $\Pi$, and use the optimal parameters for $\Pi_j$ to approximate the optimal parameters for $\Pi$. In our experiments, we set $\mathcal{S}=\{-1.0, -0.9995, -0.999, \dots, 0.999, 0.9995, 1.0\}$ (i.e. $\mathcal{S}$ contains $4001$ uniformly distributed points in $[-1, 1]$) and find that this greatly accelerates the estimation algorithm without deteriorating its performance much.

For Bayesian update with an ELF, we use the method in Section \ref{sec:bi_elf_af} or Appendix \ref{sec:bi_elf_ab} to compute the posterior mean and variance of $\theta$, depending on whether the scheme is ancilla-free or ancilla-based. In particular, during the sinusoidal fitting of the ELF, we set $\Theta=\{\mu-\sigma, \mu-0.8\sigma, \dots, \mu+0.8\sigma, \mu+\sigma\}$ (i.e. $\Theta$ contains $11$ uniformly distributed points in $[\mu-\sigma, \mu+\sigma]$) in Eq.~(\ref{eq:lstsq_af}) or (\ref{eq:lstsq_ab}). We find that this is sufficient for obtaining a high-quality sinusoidal fit of the true likelihood function.

\subsection{Comparing the performance of various schemes}
\label{subsec:compare_schemes}

To compare the performance of various quantum-generated likelihood functions, including AB CLF, AB ELF, AF CLF and AF ELF, we run Bayesian inference with each of them, fixing the number of circuit layers $L=6$ and layer fidelity $p=0.9$ (note that this layer fidelity corresponds to a $12$-qubit experiment with two-qubit gate depth of $12$ and two-qubit gate fidelity $99.92\%$, which is almost within reach for today's quantum devices). Figures \ref{fig:normal_elf_clf_separation}, \ref{fig:normal_elf_clf_separation2},
\ref{fig:failed_ab_clf}, \ref{fig:failed_af_clf} and \ref{fig:failed_ab_af_clf} illustrate the performance of different schemes with respect to various true values of $\Pi$. These results suggest that:
\begin{itemize}
\item In both the ancilla-based and ancilla-free cases, ELF performs better than (or as well as) CLF. This means that by tuning the generalized reflection angles, we do enhance the rate of information gain and make the estimation of $\Pi$ more efficient.
\item AF ELF always performs better than AB ELF, whereas AF CLF may perform better or worse than AB CLF, depending on the true value of $\Pi$, but on average, AF CLF outperforms AB CLF. So overall the ancilla-free schemes are superior to the ancilla-based ones.
\item While $\operatorname{RMSE}_t \to 0$ as $t \to \infty$ for AB ELF and AF ELF, the same is not always true for AB CLF and AF CLF. In fact, the performance of AB CLF and AF CLF depends heavily on the true value of $\Pi$, while the performance of AB ELF and AF ELF is not much affected by this value. 
\end{itemize}

\begin{figure*}[!ht]
\begin{minipage}{.48\textwidth}
\center
\includegraphics[width=0.95\linewidth]{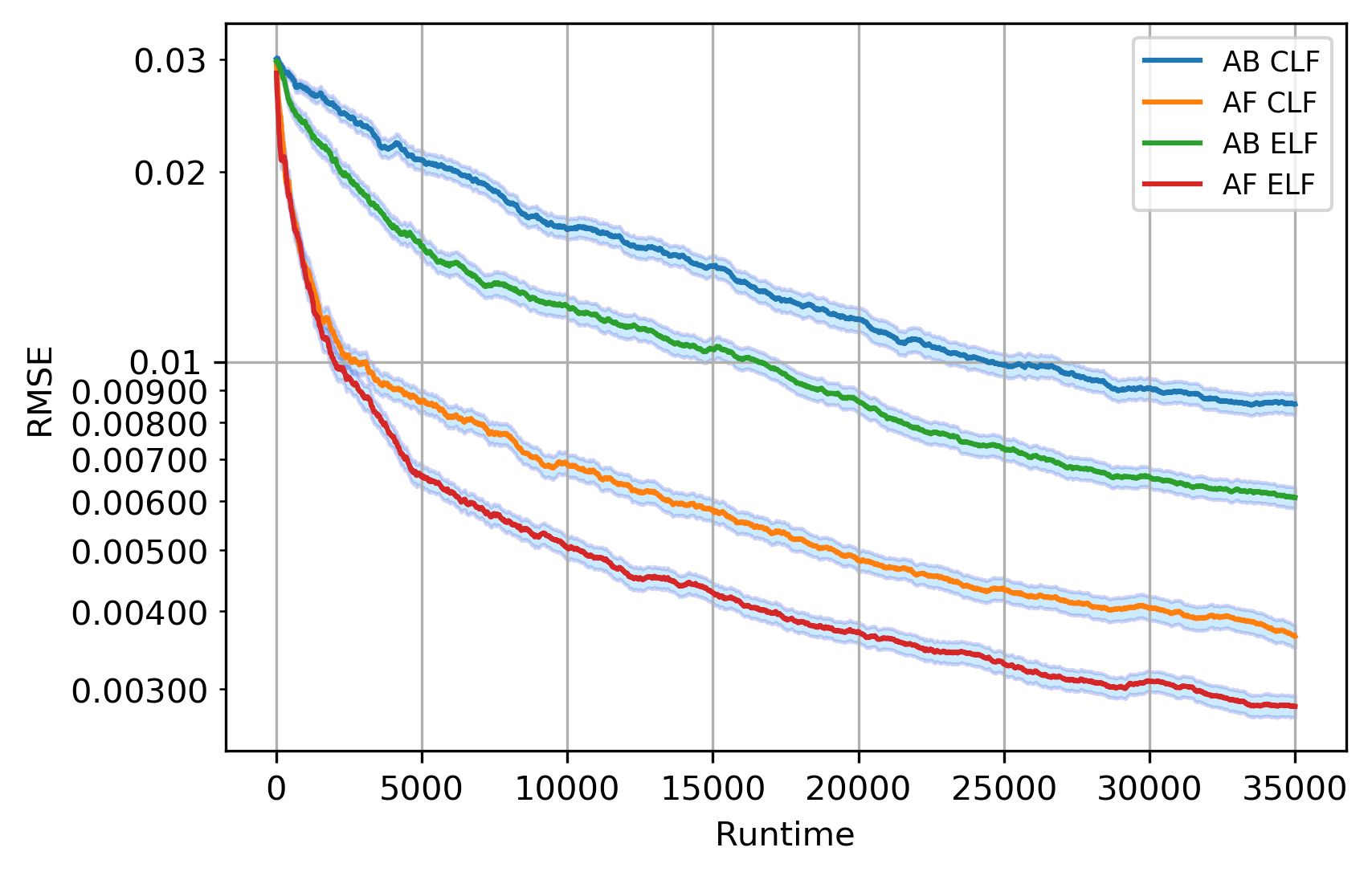}
\end{minipage}
\begin{minipage}{.48\textwidth}
\center
\includegraphics[width=0.95\linewidth]{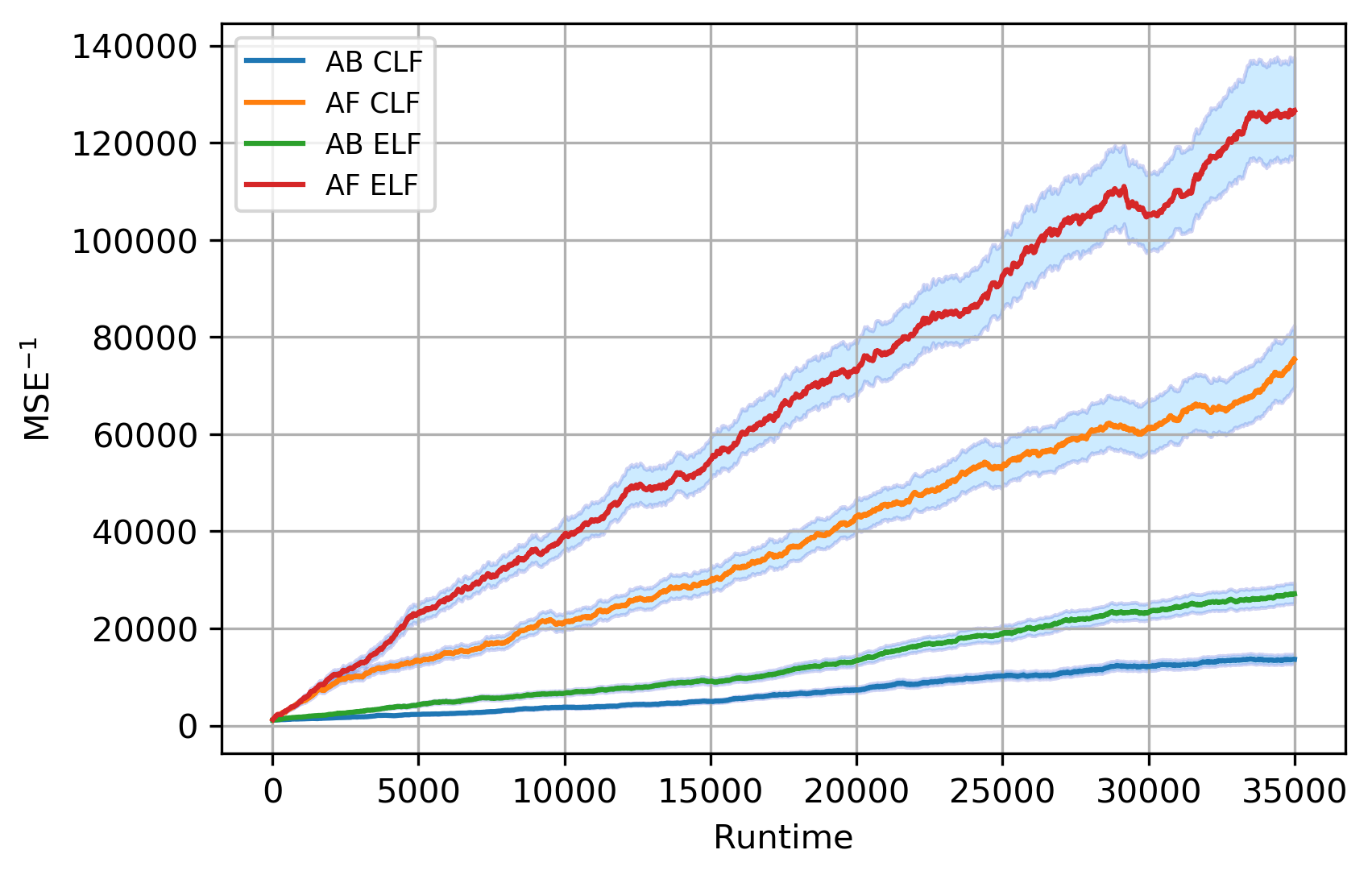}
\end{minipage}
\caption{This figure compares the performance of AB CLF, AB ELF, AF CLF and AF ELF when the expectation value $\Pi$ has true value $-0.4$ and prior distribution $\mathcal{N}(-0.43, 0.0009)$, the number of circuit layers is $L=6$, and the layer fidelity is $p=0.9$. Note that ELF outperforms CLF in both the ancilla-free and ancilla-based cases, and the ancilla-free schemes outperform the ancilla-based ones.}
\label{fig:normal_elf_clf_separation}
\end{figure*}

\begin{figure*}[!ht]
\begin{minipage}{.48\textwidth}
\center
\includegraphics[width=0.95\linewidth]{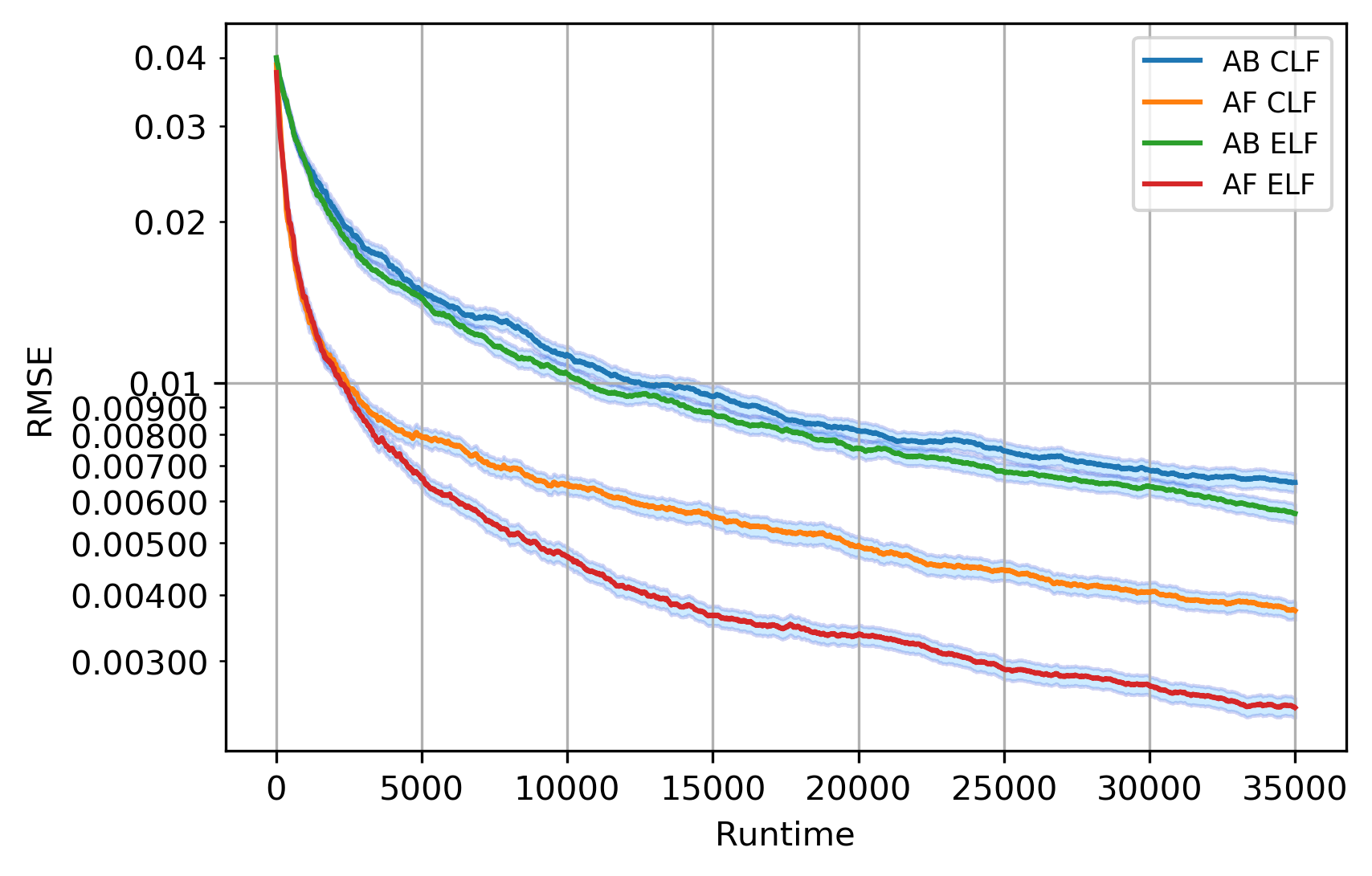}
\end{minipage}
\begin{minipage}{.48\textwidth}
\center
\includegraphics[width=0.95\linewidth]{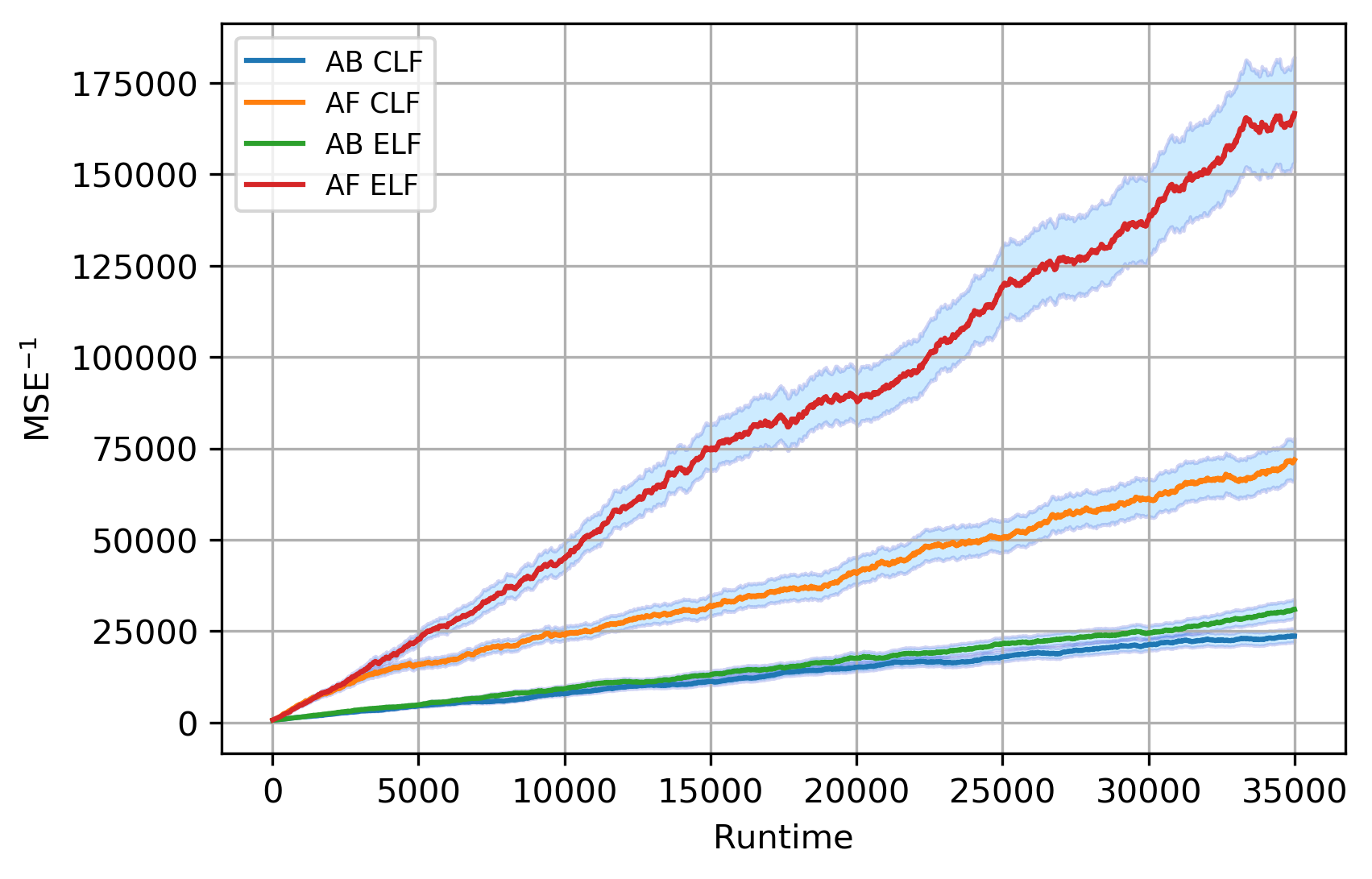}
\end{minipage}
\caption{This figure compares the performance of AB CLF, AB ELF, AF CLF and AF ELF when the expectation value $\Pi$ has true value $0.6$ and prior distribution  $\mathcal{N}(0.64, 0.0009)$, the number of circuit layers is $L=6$, and the layer fidelity is $p=0.9$. Note that AB ELF slightly outperforms AB CLF, while AF ELF outperforms AF CLF to a larger extent.}
\label{fig:normal_elf_clf_separation2}
\end{figure*}

\begin{figure*}[!ht]
\begin{minipage}{.48\textwidth}
\center
\includegraphics[width=0.95\linewidth]{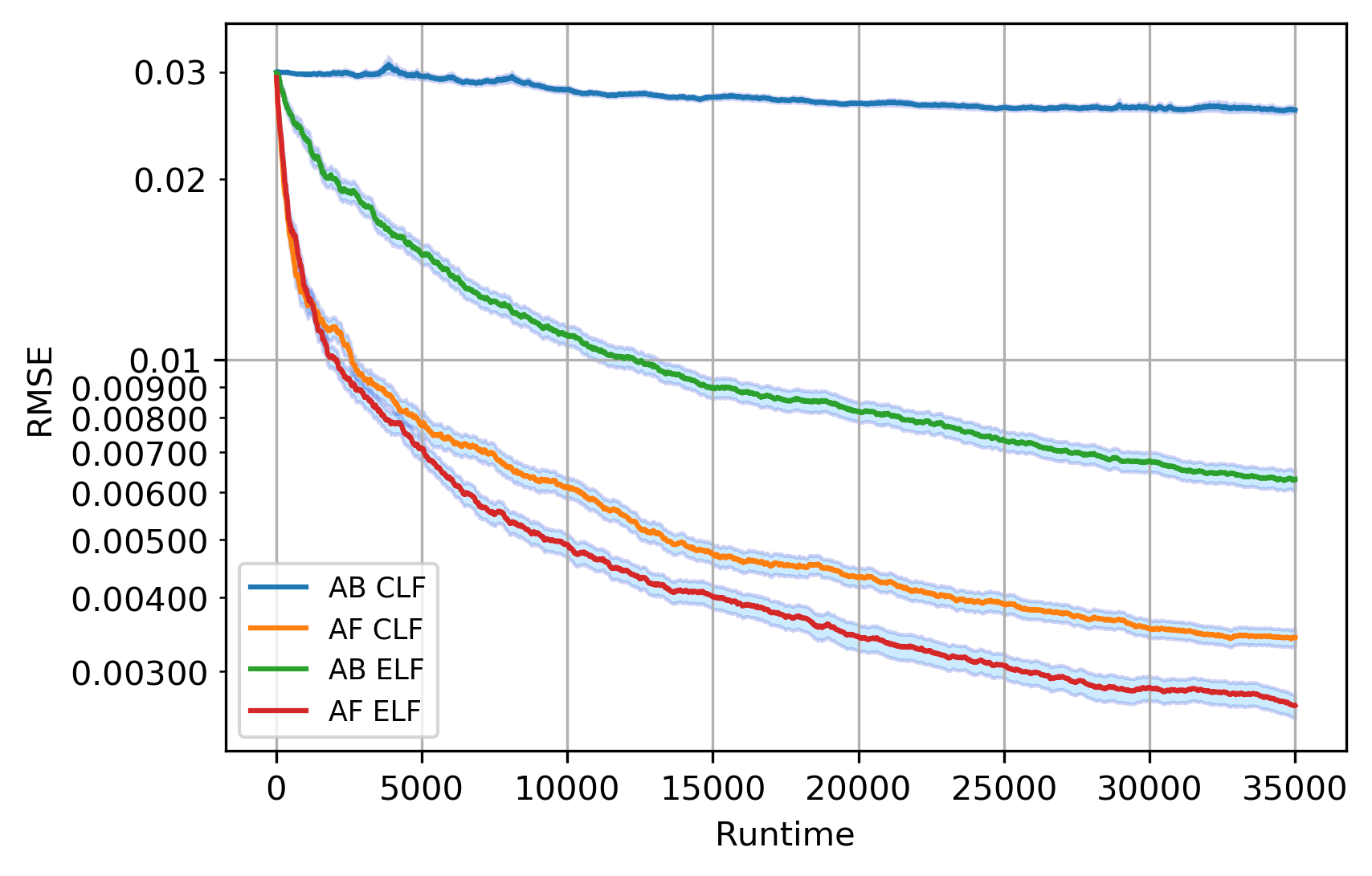}
\end{minipage}
\begin{minipage}{.48\textwidth}
\center
\includegraphics[width=0.95\linewidth]{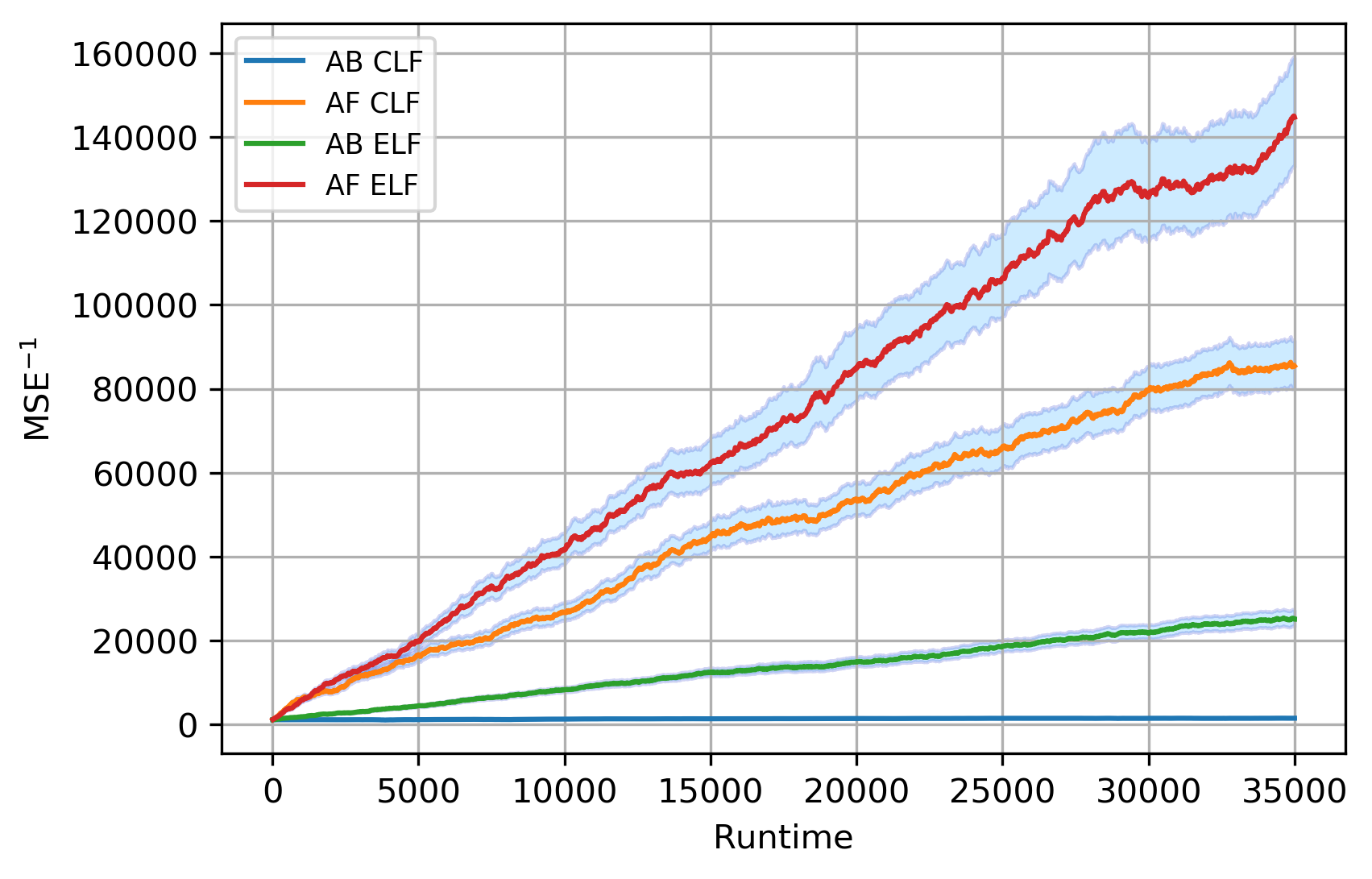}
\end{minipage}
\caption{This figure compares the performance of AB CLF, AB ELF, AF CLF and AF ELF when the expectation value $\Pi$ has true value $0.52$ and prior distribution $\mathcal{N}(0.49, 0.0009)$, the number of circuit layers is $L=6$, and the layer fidelity is $p=0.9$. Note that the RMSE of the estimator of $\Pi$ converges to $0$ for all schemes except AB CLF. AF ELF achieves the best performance.}
\label{fig:failed_ab_clf}
\end{figure*}

\begin{figure*}[!ht]
\begin{minipage}{.48\textwidth}
\center
\includegraphics[width=0.95\linewidth]{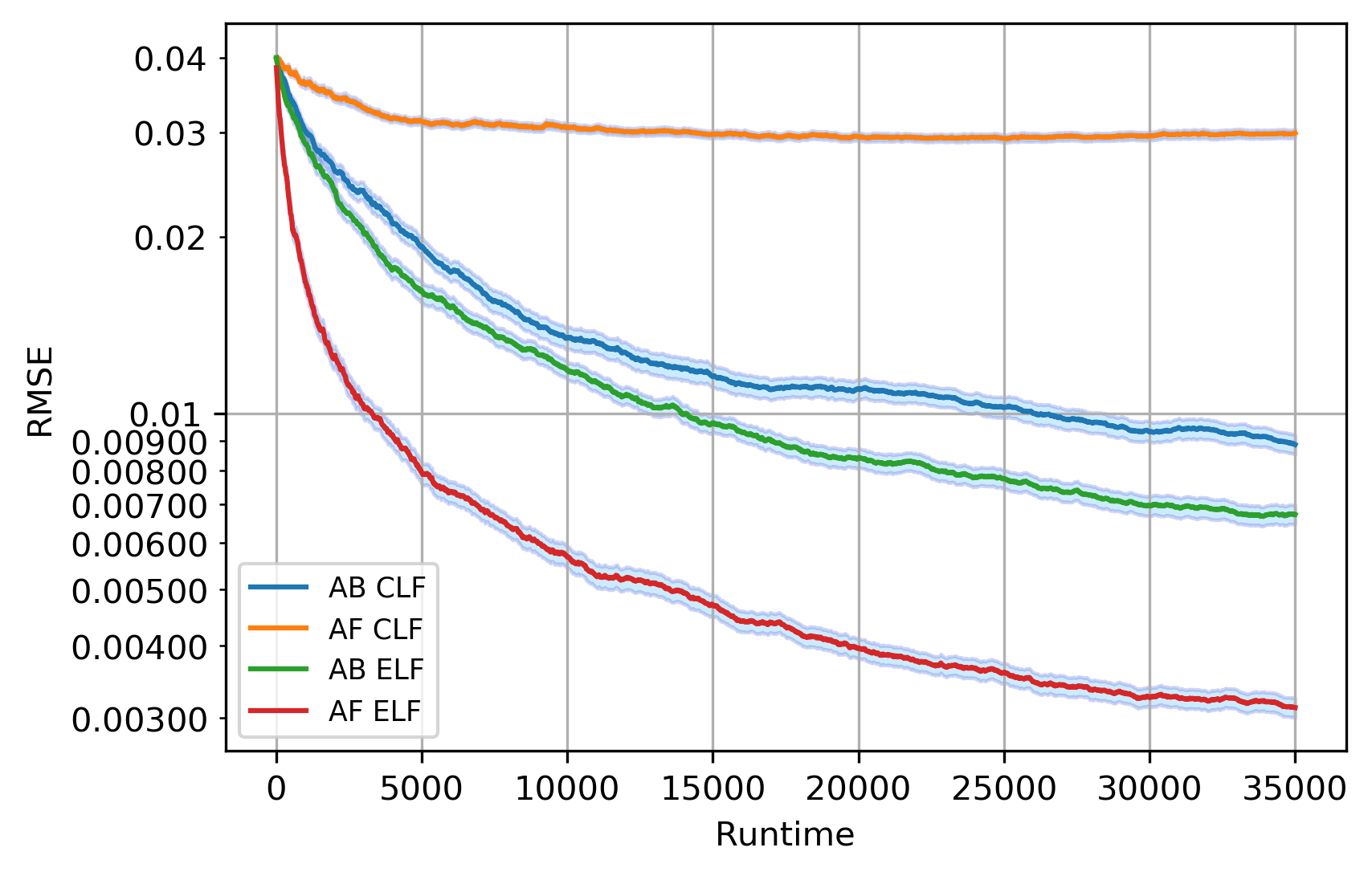}
\end{minipage}
\begin{minipage}{.48\textwidth}
\center
\includegraphics[width=0.95\linewidth]{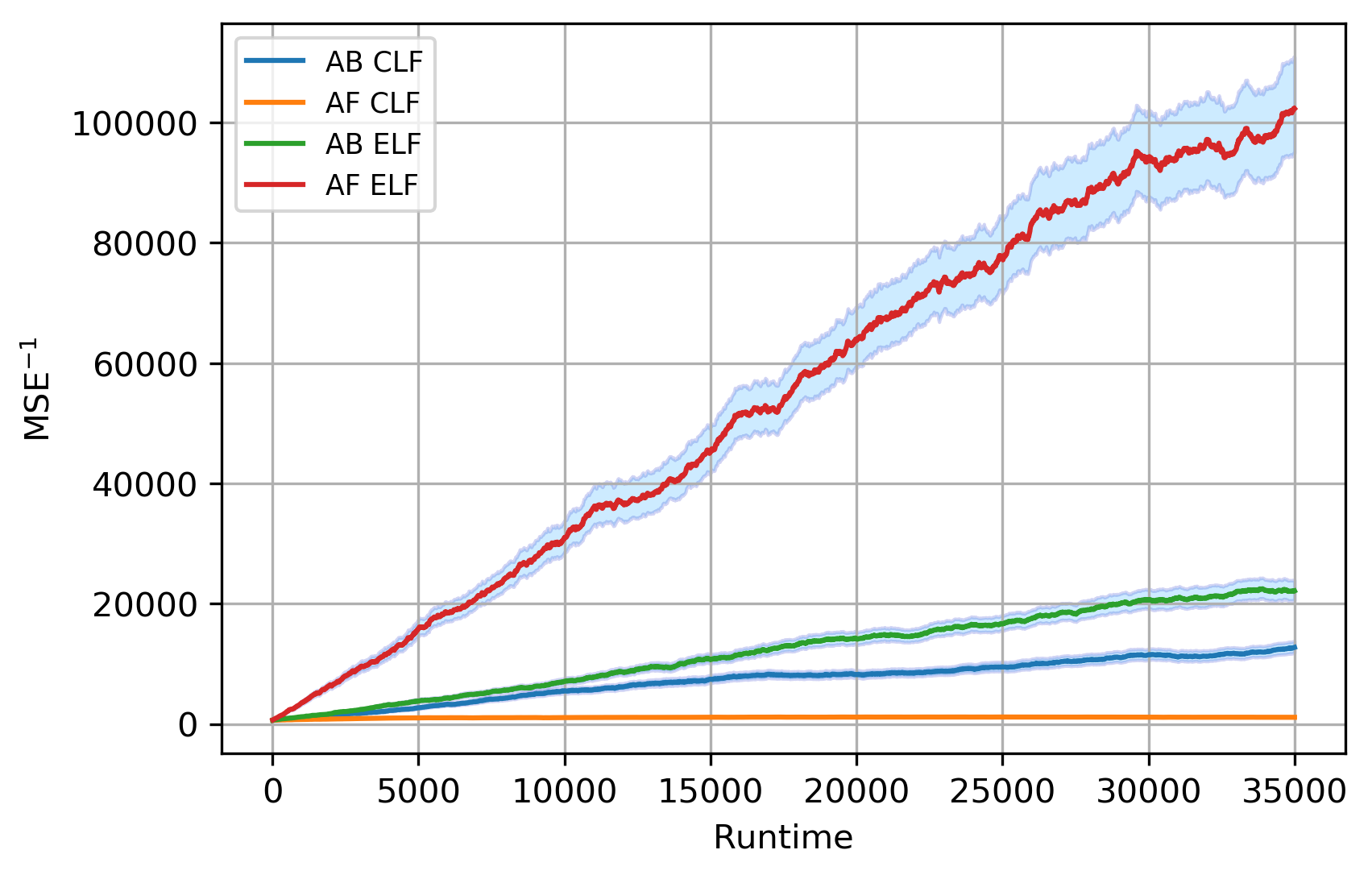}
\end{minipage}
\caption{This figure compares the performance of AB CLF, AB ELF, AF CLF and AF ELF when the expectation value $\Pi$ has true value $-0.1$ and prior distribution $\mathcal{N}(-0.14, 0.0009)$, the number of circuit layers is $L=6$, and the layer fidelity is $p=0.9$. Note that the RMSE of the estimator of $\Pi$ converges to $0$ for all schemes except AF CLF. AF ELF achieves the best performance.}
\label{fig:failed_af_clf}
\end{figure*}

\begin{figure*}[!ht]
\begin{minipage}{.48\textwidth}
\center
\includegraphics[width=0.95\linewidth]{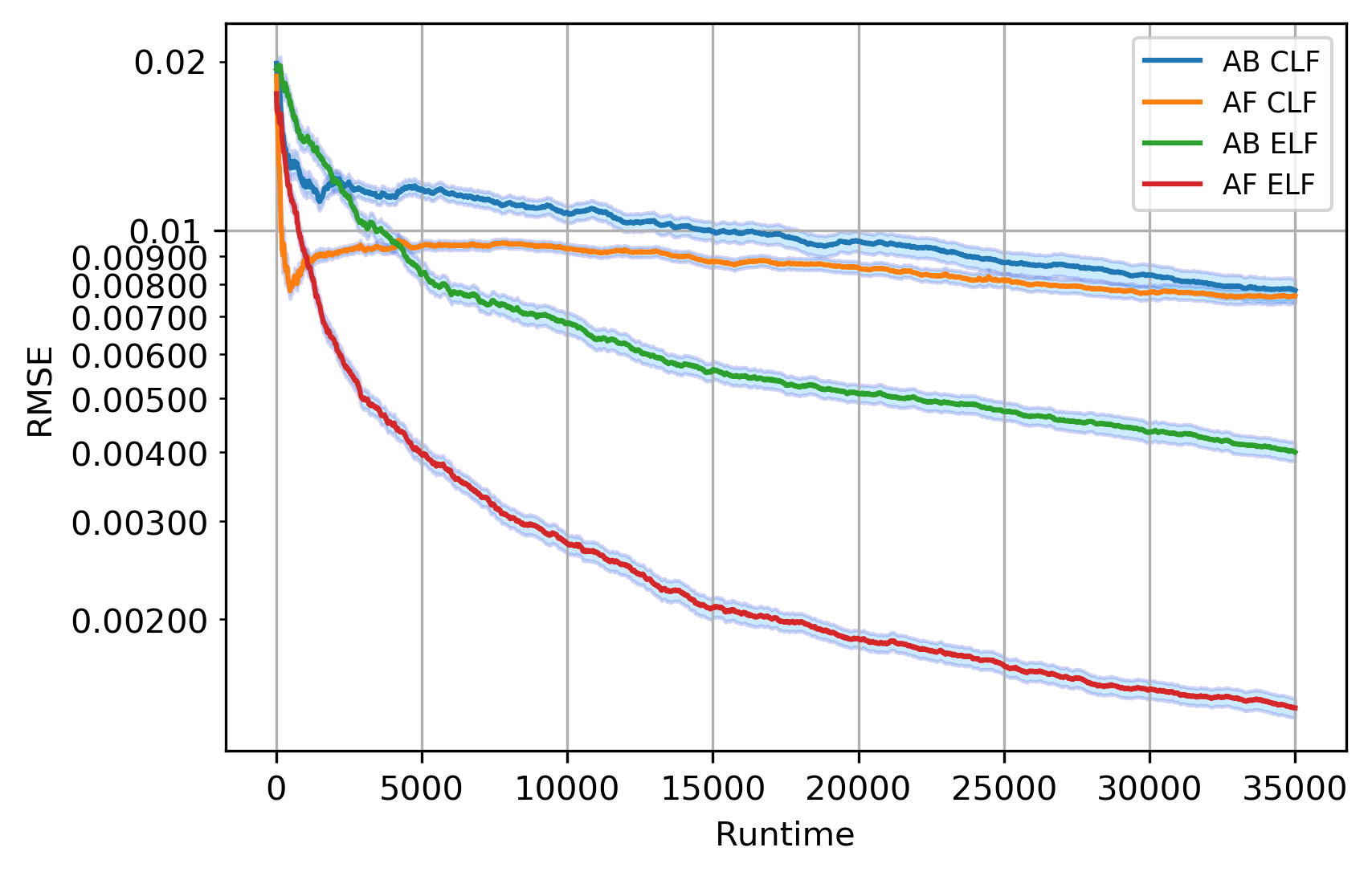}
\end{minipage}
\begin{minipage}{.48\textwidth}
\center
\includegraphics[width=0.95\linewidth]{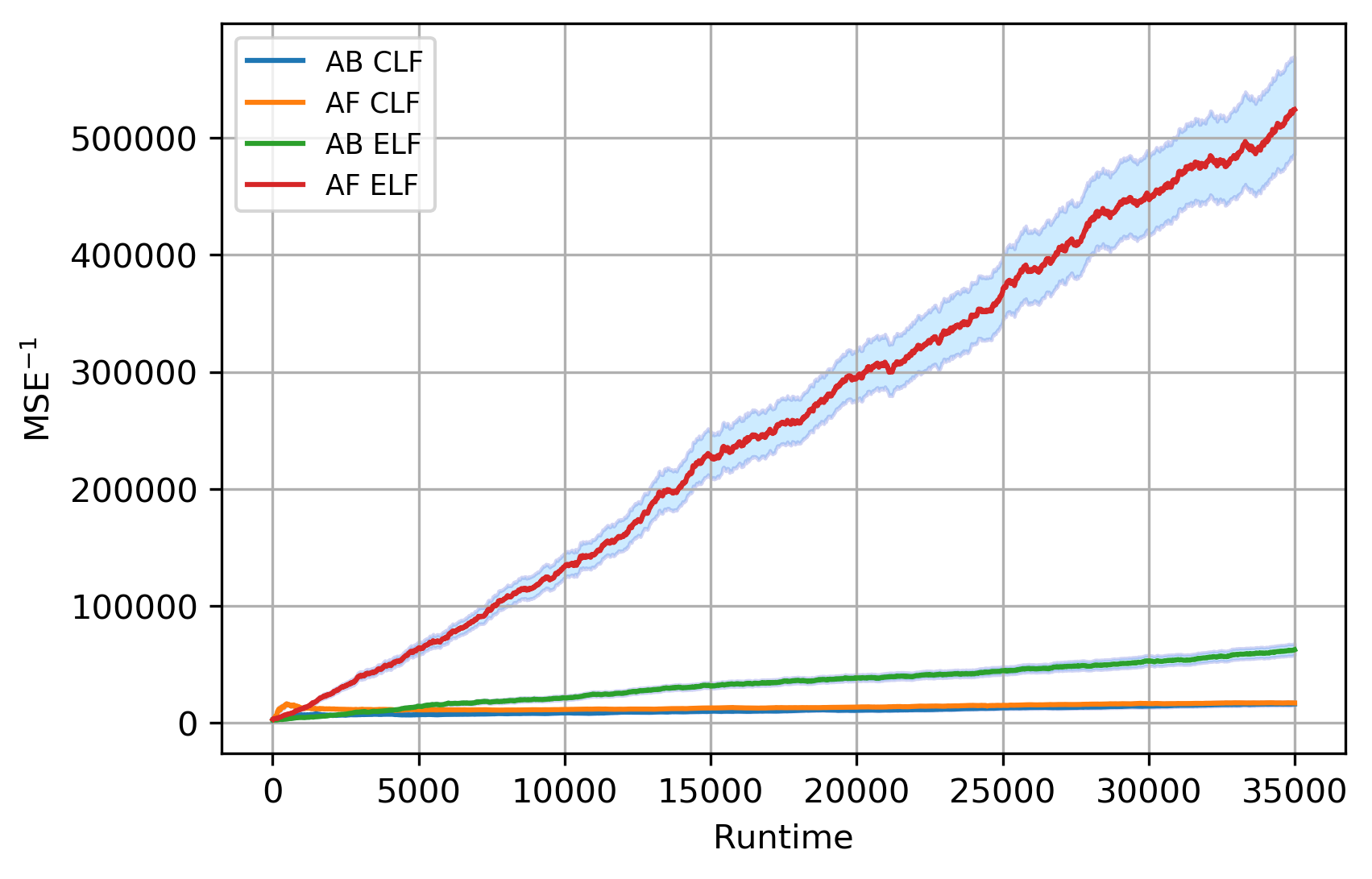}
\end{minipage}
\caption{This figure compares the performance of AB CLF, AB ELF, AF CLF and AF ELF when the expectation value $\Pi$ has true value $0.9$ and prior distribution $\mathcal{N}(0.92, 0.0009)$, the number of circuit layers is $L=6$, and the layer fidelity is $p=0.9$. Note that the RMSE of the estimator of $\Pi$ fails to converge to $0$ for AF CLF and decreases slowly for AB CLF. Both AB CLF and AF CLF are outperformed by AB ELF and AF ELF, with AF ELF achieving the best performance.}
\label{fig:failed_ab_af_clf}
\end{figure*}

We may compare the above results with Figure \ref{fig:R0-L_6_p_0.9} that illustrates the $\hat{R}_0$ factors (as defined in Eq.~(\ref{eq:r0_factor})) of AB CLF, AB ELF, AF CLF and AF ELF in the same setting. One can observe from this figure that:
\begin{itemize}
    \item The $\hat{R}_0$ factor of AB ELF is equal to or larger than that of AB CLF, and the same is true for AF ELF versus AB CLF. This explains why ELF outperforms CLF in both the ancilla-based and ancilla-free cases. 
    \item The $\hat{R}_0$ factor of AF ELF is larger than that of AB ELF. Meanwhile, The $\hat{R}_0$ factor of AF CLF can be larger or smaller than that of AB CLF, depending on the value of $\Pi$, but on average AF CLF has larger $\hat{R}_0$ factor than AB CLF. This explains the superiority of the ancilla-free schemes over the ancilla-based ones. 
    \item The $\hat{R}_0$ factors of AB ELF and AF ELF are bounded away from $0$ for all $\Pi \in [-1, 1]$ \footnote{Though not shown in Figure \ref{fig:R0-L_6_p_0.9}, the $\hat{R}_0$ factors of AB ELF and AF ELF actually diverge to $+\infty$ as $\Pi \to \pm 1$, and this is true for any $L \in \mathbb{Z}^{+}$.}. This explains why their performance is stable regardless of the true value of $\Pi$. On the other hand, the $\hat{R}_0$ factors of AB CLF and AF CLF change dramatically for different $\Pi$'s. In fact, the $\hat{R}_0$ factor of AB CLF is $0$ when $\Pi=\cosp{{j \pi}/{L}}$ for $j=0,1,\dots,L$, and the $\hat{R}_0$ factor of AF CLF is $0$ when $\Pi=\cosp{{j \pi}/{(2L+1)}}$ for $j=0,1,\dots,2L+1$. This means that if the true value of $\Pi$ is close to one of these ``dead spots", then its estimator will struggle to improve and hence the performance of AB/AF CLF will suffer (see Figures \ref{fig:failed_ab_clf}, \ref{fig:failed_af_clf} and \ref{fig:failed_ab_af_clf} for examples.). AB ELF and AF ELF, on the other hand, do not have this weakness. 
\end{itemize}

\begin{figure}[!ht]
\center
\includegraphics[width=0.9\linewidth]{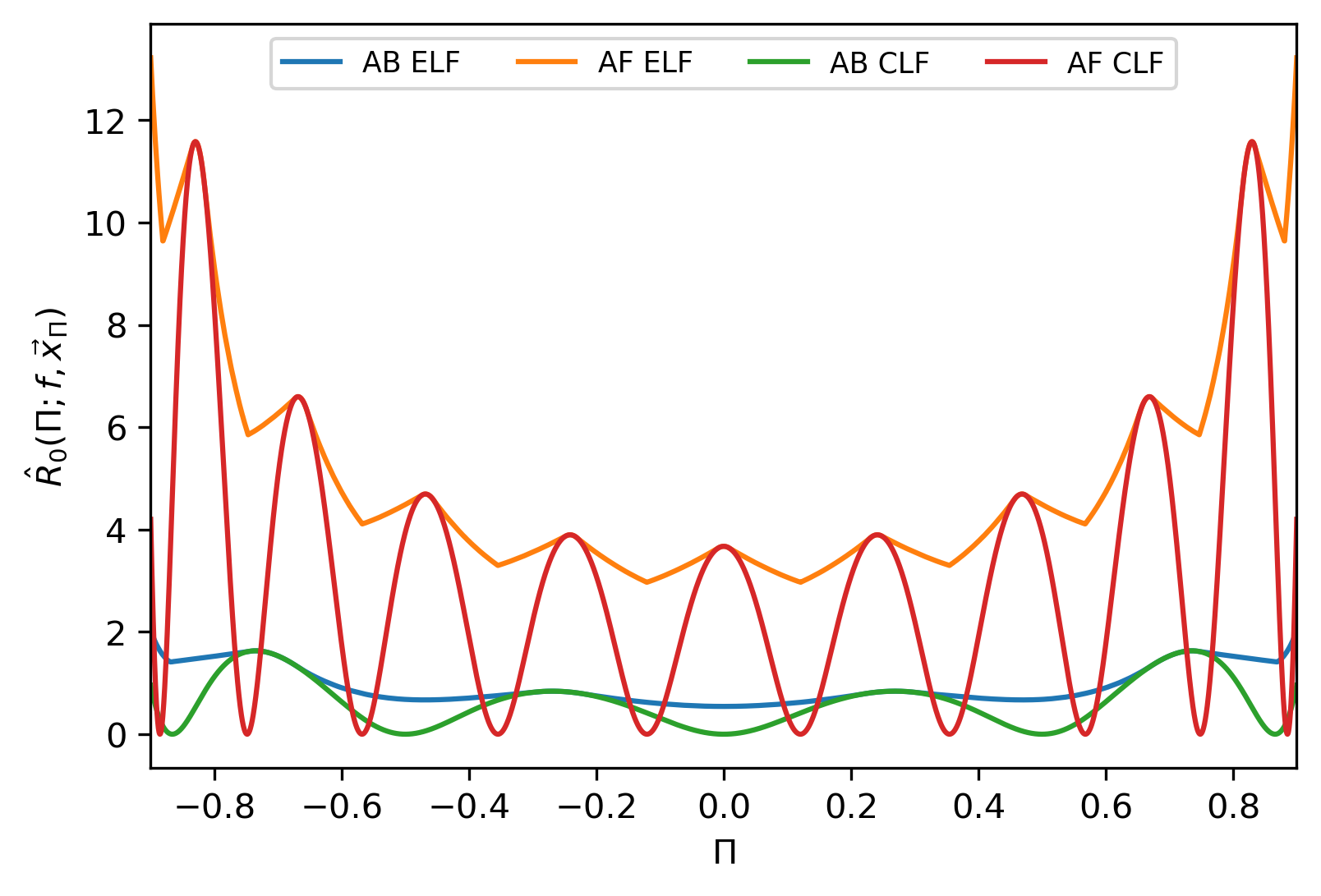}
\caption{This figure compares the $\hat{R}_0$ factors of AB CLF, AB ELF, AF CLF and AF ELF for $\Pi \in [-0.9, 0.9]$, when the number of circuit layers is $L=6$ and the layer fidelity is $p=0.9$. Here $\vec x_{\Pi} \in \mathbb{R}^{2L}$ is a global maximum point of $\hat{R}_0(\Pi; f, \vec x)$ for given $\Pi$ and $f=p^L=0.531441$. All of the four curves are plotted on a fine grid with interval length $0.0005$ (i.e. there are $3601$ equally-spaced grid points in [-0.9, 0.9]). Note that the $\hat{R}_0$ factors of AB CLF and AF CLF change dramatically for different $\Pi$'s. In fact, they can be close to $0$ for certain $\Pi$'s. By contrast, the $\hat{R}_0$ factors of AB ELF and AF ELF are bounded away from $0$ for all $\Pi$'s. }
\label{fig:R0-L_6_p_0.9}
\end{figure}

\subsection{Understanding the performance of Bayesian inference with ELFs}
\label{subsec:decay_rmse_elf}

Having established the improved performance of ELFs over CLFs, we now analyze the performance of Bayesian inference with ELFs in more detail. Note that Figures \ref{fig:normal_elf_clf_separation}, \ref{fig:normal_elf_clf_separation2},
\ref{fig:failed_ab_clf}, \ref{fig:failed_af_clf} and
\ref{fig:failed_ab_af_clf} suggest that the inverse MSEs of both AB-ELF-based and AF-ELF-based estimators of $\Pi$ grow linearly in time when the circuit depth is fixed. By fitting a linear model to the data, we obtain the empirical growth rates of these quantities, which are shown in
Tables \ref{table:inverse_mse_growth_rate_ab}
and \ref{table:inverse_mse_growth_rate_af}, respectively. We also compare these rates with the $\hat{R}_0$ factors of AB ELF and AF ELF in the same setting. It turns out that the $\hat{R}_0$ factor is a rough estimate of the true growth rate of the inverse MSE of an ELF-based estimator of $\Pi$, but it can be unreliable sometimes. We leave it as an open question to give a more precise characterization of the decay of the RMSEs of ELF-based estimators of $\Pi$ during the inference process.

\begin{table}[!ht]
\caption{The predicted and empirical growth rates of the inverse MSEs of AB-ELF-based estimators of $\Pi$ in the five experiments in Section \ref{subsec:compare_schemes}. In all of these experiments, the number of circuit layers is $L=6$ and the layer fidelity is $p=0.9$.}
\label{table:inverse_mse_growth_rate_ab}
\begin{ruledtabular}
\begin{tabular}{cccc} 
True value & Prior distribution & Predicated  & Empirical  \\
of $\Pi$ & of $\Pi$ & growth rate  & growth rate \\
& & of $\operatorname{MSE}_t^{-1}$ & of $\operatorname{MSE}_t^{-1}$ \\
 \hline
 $-0.4$ & $\mathcal{N}(-0.43, 0.0009)$  & 0.71 & 0.75  \\ 
 $0.6$ & $\mathcal{N}(0.64, 0.0009)$  & 0.91 & 0.86  \\ 
 $0.52$ & $\mathcal{N}(0.49, 0.0009)$  & 0.70  & 0.75  \\ 
 $-0.1$ & $\mathcal{N}(-0.14, 0.0009)$  & 0.60 & 0.68 \\ 
 $0.9$ & $\mathcal{N}(0.92, 0.0009)$  & 2.03 & 1.84  \\ 
\end{tabular}
\end{ruledtabular}
\end{table}

\begin{table}[!ht]
\caption{The predicted and empirical growth rates of the inverse MSEs of AF-ELF-based estimators of $\Pi$ in the five experiments in Section \ref{subsec:compare_schemes}. In all of these experiments, the number of circuit layers is $L=6$ and the layer fidelity is $p=0.9$.}
\label{table:inverse_mse_growth_rate_af}
\begin{ruledtabular}
\begin{tabular}{cccc} 
True value & Prior distribution & Predicated  & Empirical  \\
of $\Pi$ & of $\Pi$ & growth rate  & growth rate \\
& & of $\operatorname{MSE}_t^{-1}$ & of $\operatorname{MSE}_t^{-1}$ \\
 \hline
 $-0.4$ & $\mathcal{N}(-0.43, 0.0009)$  & 3.77 & 3.70 \\ 
 $0.6$ & $\mathcal{N}(0.64, 0.0009)$  & 4.73 & 4.68 \\ 
 $0.52$ & $\mathcal{N}(0.49, 0.0009)$  & 4.36 & 4.19 \\ 
 $-0.1$ & $\mathcal{N}(-0.14, 0.0009)$  & 3.07 & 3.08 \\ 
 $0.9$ & $\mathcal{N}(0.92, 0.0009)$  & 13.22 & 14.74 \\  
\end{tabular}
\end{ruledtabular}
\end{table}

\subsubsection{Analyzing the impact of layer fidelity on the performance of estimation}
\label{subsec:impact_layer_fidelity}

To investigate the influence of layer fidelity on the performance of estimation, we run Bayesian inference with AB/AF ELF for fixed circuit depth but varied layer fidelity. Specifically, we set the number $L$ of circuit layers to be $6$, and vary the layer fidelity $p$ from $0.75, 0.8, 0.85, 0.9$ to $0.95$.
Figures \ref{fig:layer_fidelity_on_ab_elf} and \ref{fig:layer_fidelity_on_af_elf} illustrate the simulation results in the ancilla-based and ancilla-free cases, respectively. As expected, higher layer fidelity leads to better performance of the algorithm. Namely, the less noisy the circuit is, the faster the RMSE of the estimator of $\Pi$ decreases. This is consistent with the fact that the $\hat{R}_0$ factors of AB ELF and AF ELF are monotonically increasing functions of $f=p^L$, as demonstrated by Figure \ref{fig:R0_various_fidelities}.

\begin{figure*}[!ht]
\begin{minipage}{.48\textwidth}
\center
\includegraphics[width=0.95\linewidth]{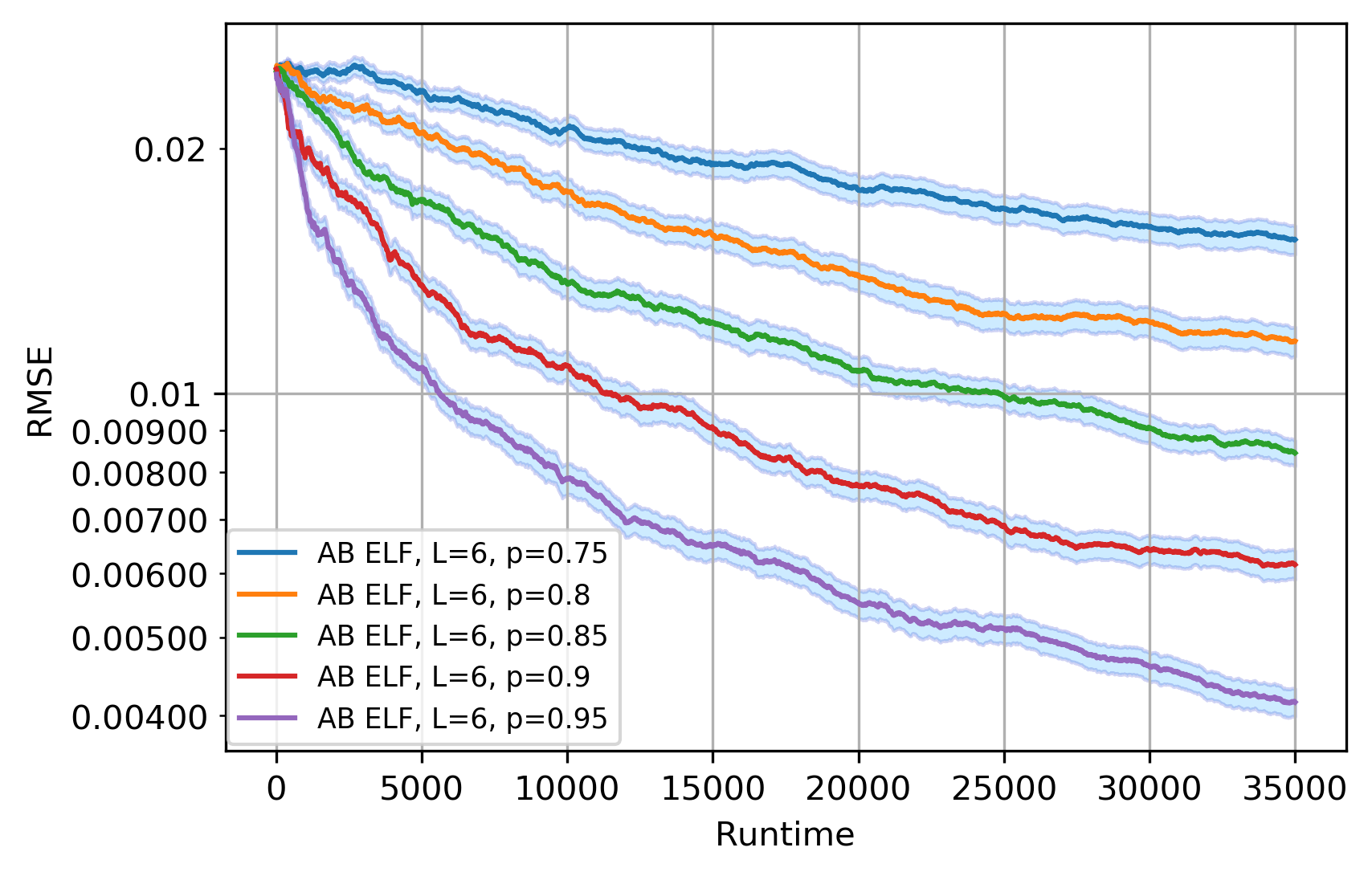}
\end{minipage}
\begin{minipage}{.48\textwidth}
\center
\includegraphics[width=0.95\linewidth]{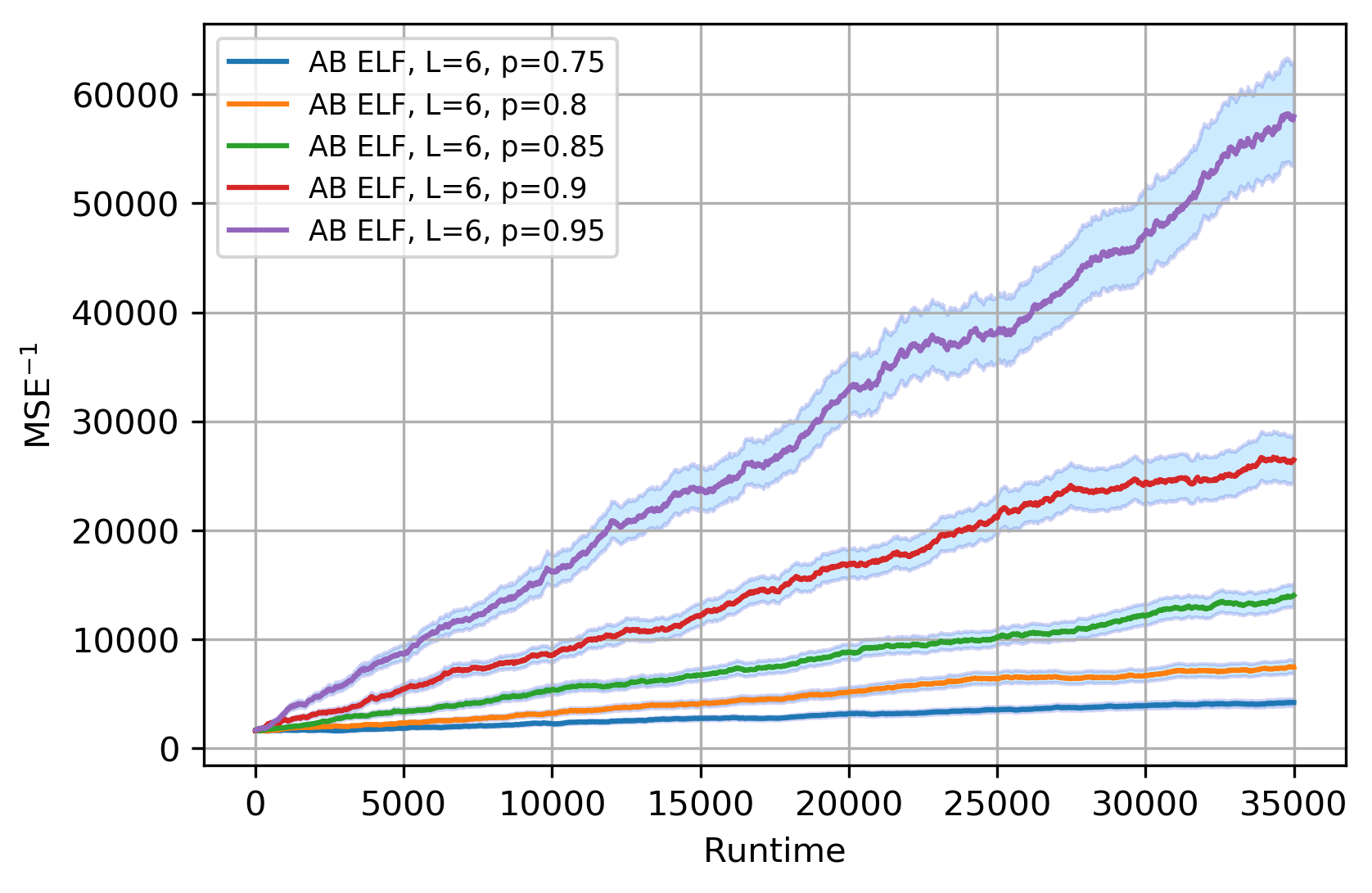}
\end{minipage}
\caption{This figure demonstrates the impact of layer fidelity on the performance of AB ELF. Here $\Pi$ has true value $0.18$ and prior distribution $\mathcal{N}(0.205, 0.0009)$, the number $L$ of circuit layers is $6$, and the layer fidelity $p$ is varied from $0.75, 0.8, 0.85, 0.9$ to $0.95$. Note that higher layer fidelity leads to better performance of estimation.}
\label{fig:layer_fidelity_on_ab_elf}
\end{figure*}

\begin{figure*}[!ht]
\begin{minipage}{.48\textwidth}
\center
\includegraphics[width=0.95\linewidth]{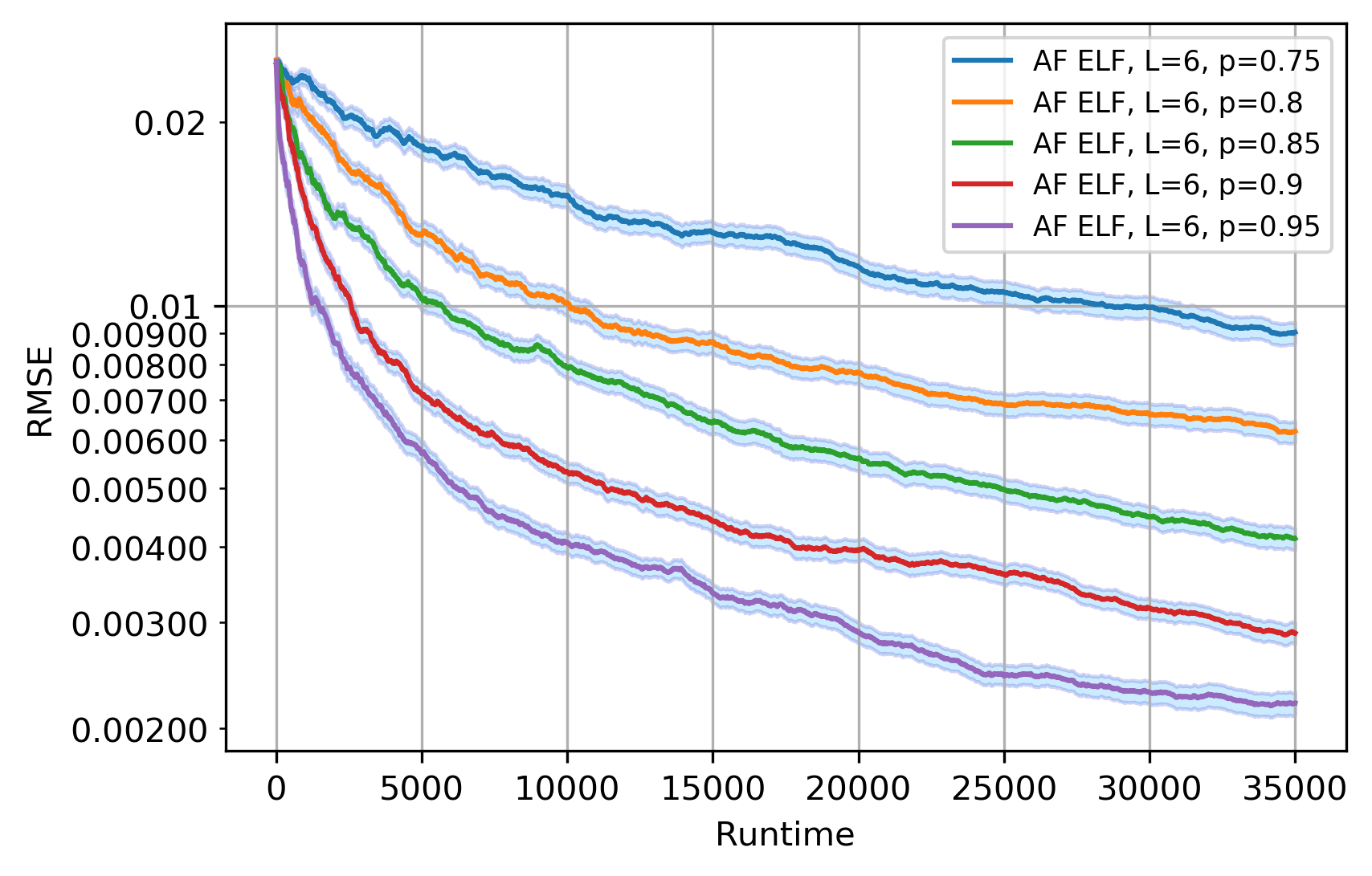}
\end{minipage}
\begin{minipage}{.48\textwidth}
\center
\includegraphics[width=0.95\linewidth]{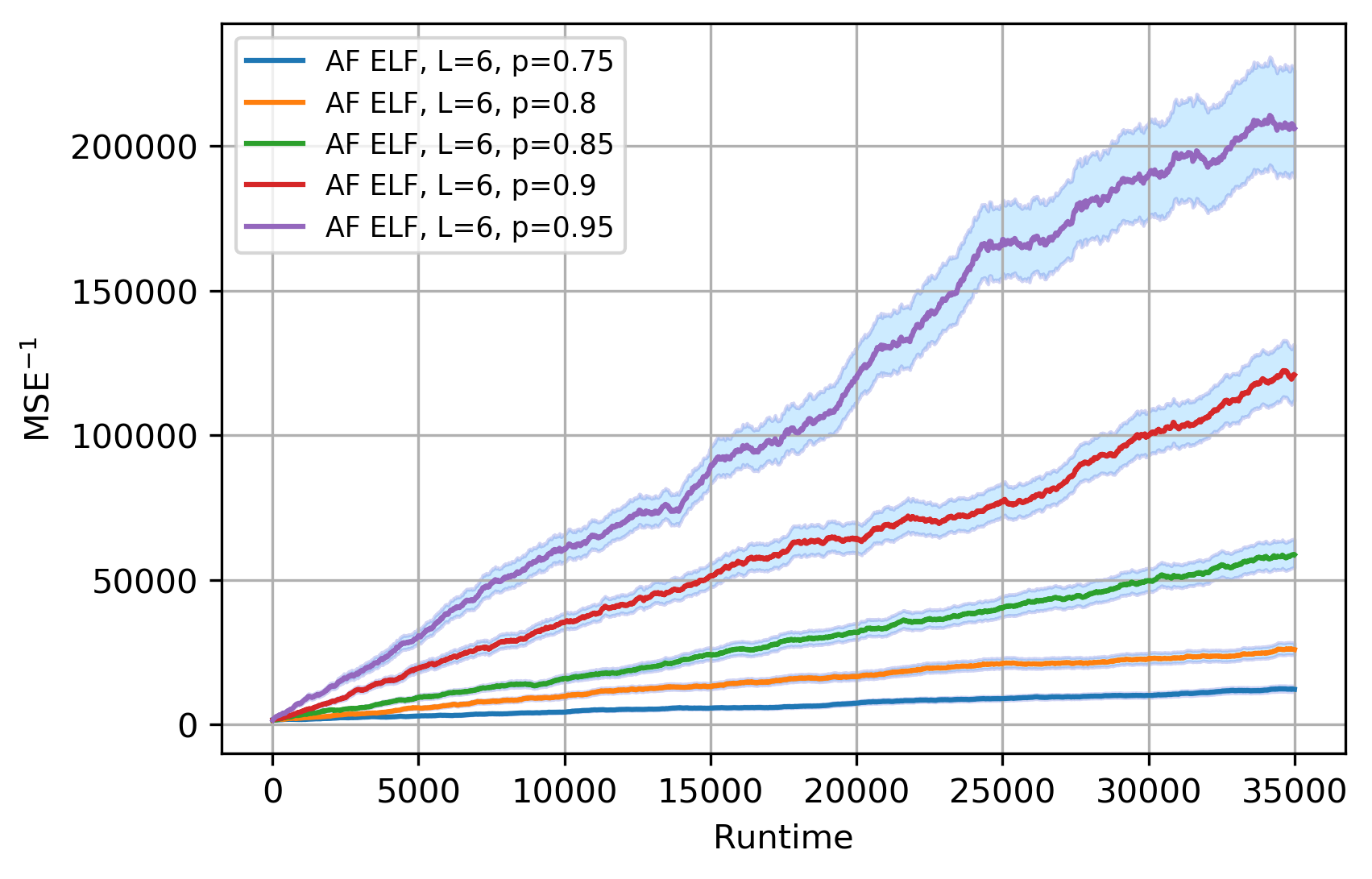}
\end{minipage}
\caption{This figure demonstrates the impact of layer fidelity on the performance of AF ELF. Here $\Pi$ has true value $0.18$ and prior distribution $\mathcal{N}(0.205, 0.0009)$, the number $L$ of circuit layers is $6$, and the layer fidelity $p$ is varied from $0.75, 0.8, 0.85, 0.9$ to $0.95$. Note that higher layer fidelity leads to better performance of estimation.}
\label{fig:layer_fidelity_on_af_elf}
\end{figure*}

\begin{figure*}[!ht]
\begin{minipage}{.48\textwidth}
\center
\includegraphics[width=0.95\linewidth]{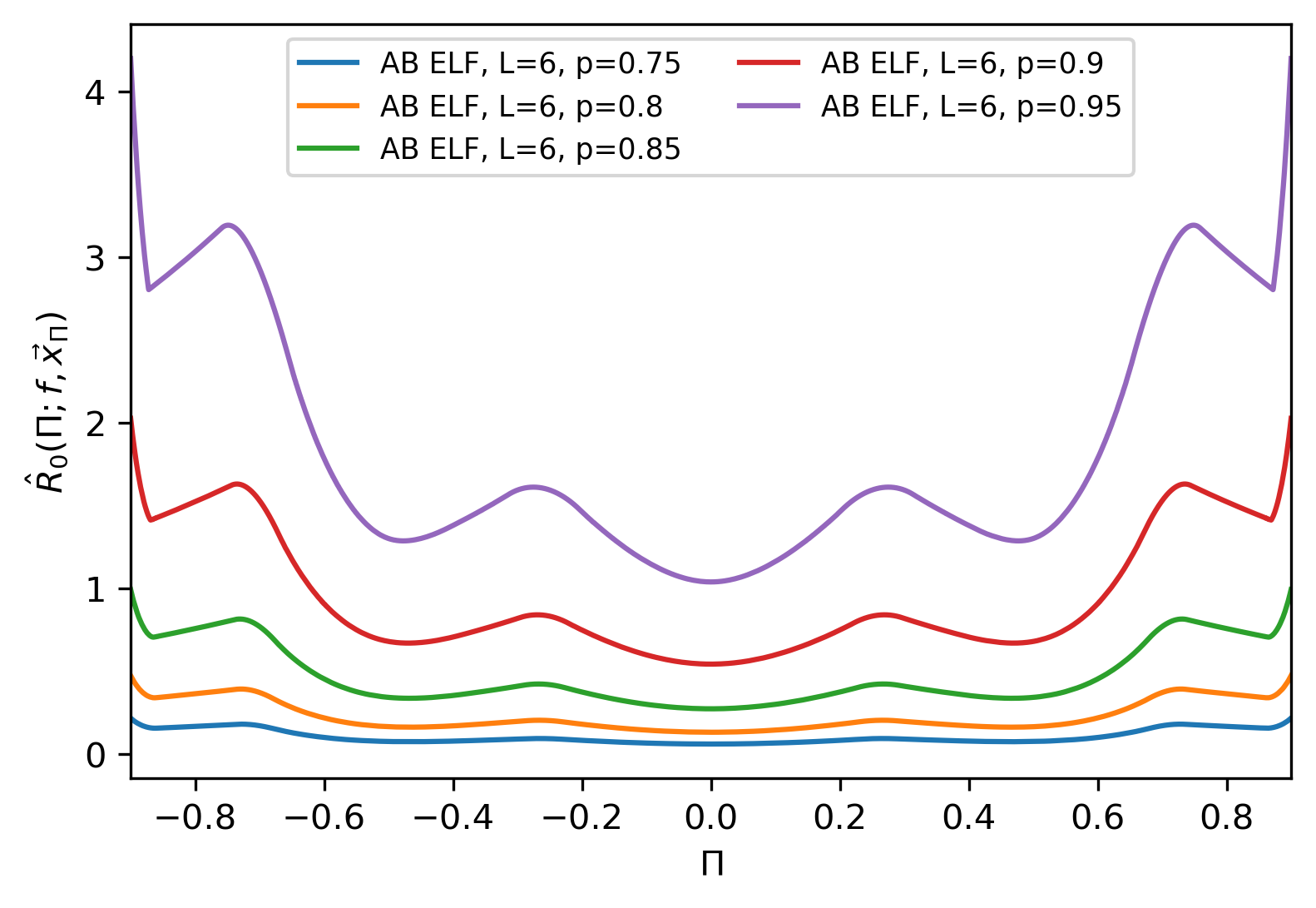}
\end{minipage}
\begin{minipage}{.48\textwidth}
\center
\includegraphics[width=0.95\linewidth]{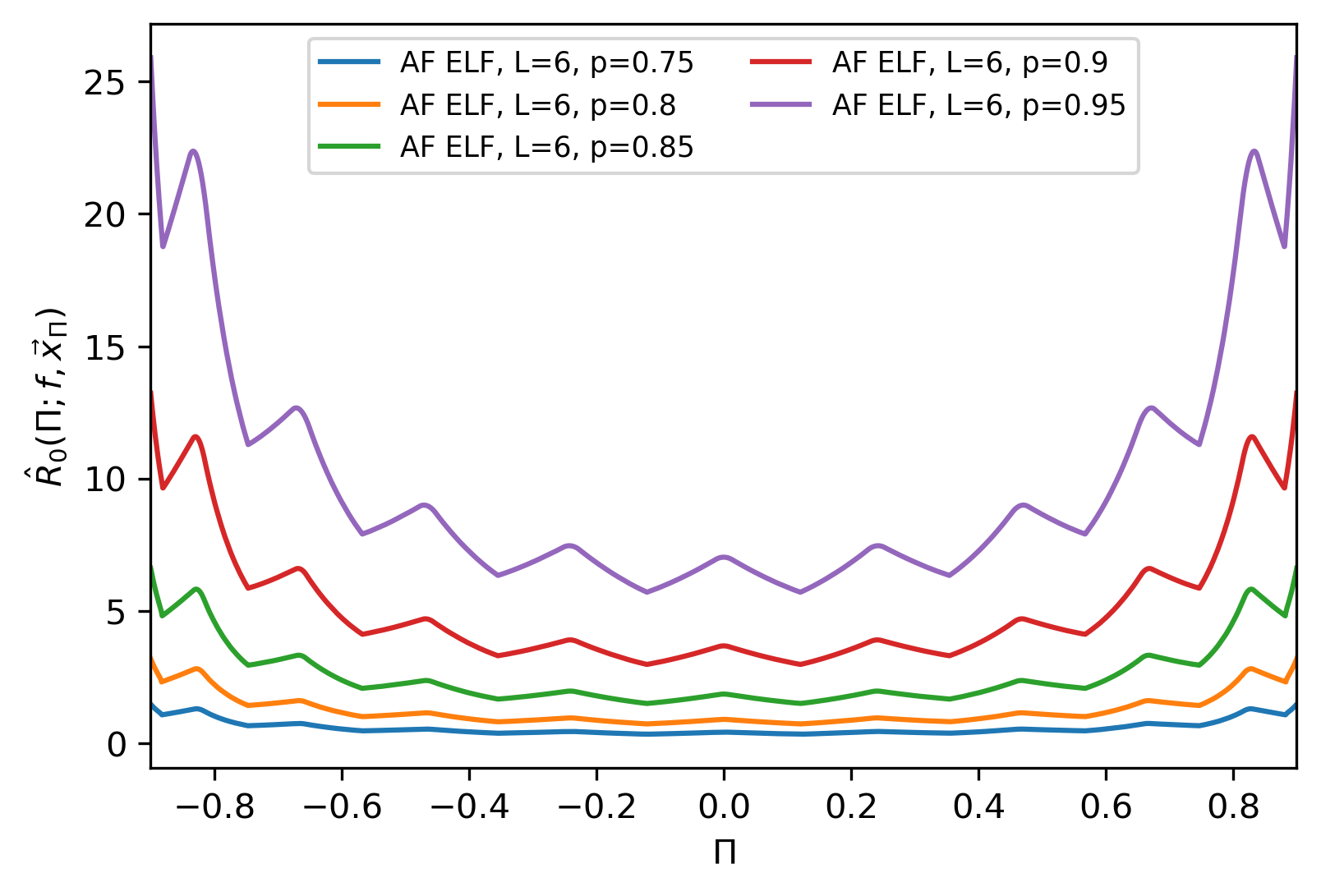}
\end{minipage}
\caption{This figure shows the $\hat{R}_0$ factors of AB ELF and AF ELF for $\Pi \in [-0.9, 0.9]$, when the number $L$ of circuit layers is fixed to $6$ and the layer fidelity $p$ is varied from $0.75, 0.8, 0.85, 0.9$ to $0.95$. Here $\vec x_{\Pi} \in \mathbb{R}^{2L}$ is a global maximum point of $\hat{R}_0(\Pi; f, \vec x)$ for given $\Pi$ and $f=p^L$. Note that higher layer fidelity leads to larger $\hat{R}_0$ factors of AB ELF and AF ELF.}
\label{fig:R0_various_fidelities}
\end{figure*}

\subsubsection{Analyzing the impact of circuit depth on the performance of estimation}
\label{subsec:impact_circuit_depth}
To investigate the influence of circuit depth on the performance of estimation, we run Bayesian inference with AB/AF ELF for fixed layer fidelity but varied circuit depth. Specifically, we set the layer fidelity $p$ to be $0.9$, and vary the number $L$ of circuit layers from $1$ to $5$. Figures \ref{fig:circuit_depth_on_ab_elf} and \ref{fig:circuit_depth_on_af_elf} illustrate the simulation results in the ancilla-based and ancilla-free cases, respectively. These results indicate that larger $L$ (i.e. deeper circuit) does not necessarily lead to better performance. The optimal choice of $L$ is indeed a subtle issue. This can be intuitively understood as follows. As $L$ increases, the likelihood function becomes steeper \footnote{More precisely, the slopes of AB ELF and AF ELF scale linearly in the number $L$ of circuit layers.} and hence gains more statistical power, if the circuit for generating it is noiseless. But on the other hand, the true fidelity of the circuit decreases exponentially in $L$ and the implementation cost of this circuit grows linearly in $L$. So one must find a perfect balance among these factors in order to maximize the performance of the estimation algorithm.

The above results are consistent with Figure \ref{fig:R0-various-depths} which illustrates the $\hat{R}_0$ factors of AB ELF and AF ELF in the same setting. Note that larger $L$ does not necessarily lead to larger $\hat{R}_0$ factor of AB/AF ELF. One can evaluate this factor for different $L$'s and choose the one that maximizes this factor. This often enables us to find a satisfactory (albeit not necessarily optimal) $L$. It remains an open question to devise an efficient strategy for determining the  $L$ that optimizes the performance of estimation given the layer fidelity $p$ and a prior distribution of $\Pi$.

\begin{figure*}[!ht]
\begin{minipage}{.48\textwidth}
\center
\includegraphics[width=0.95\linewidth]{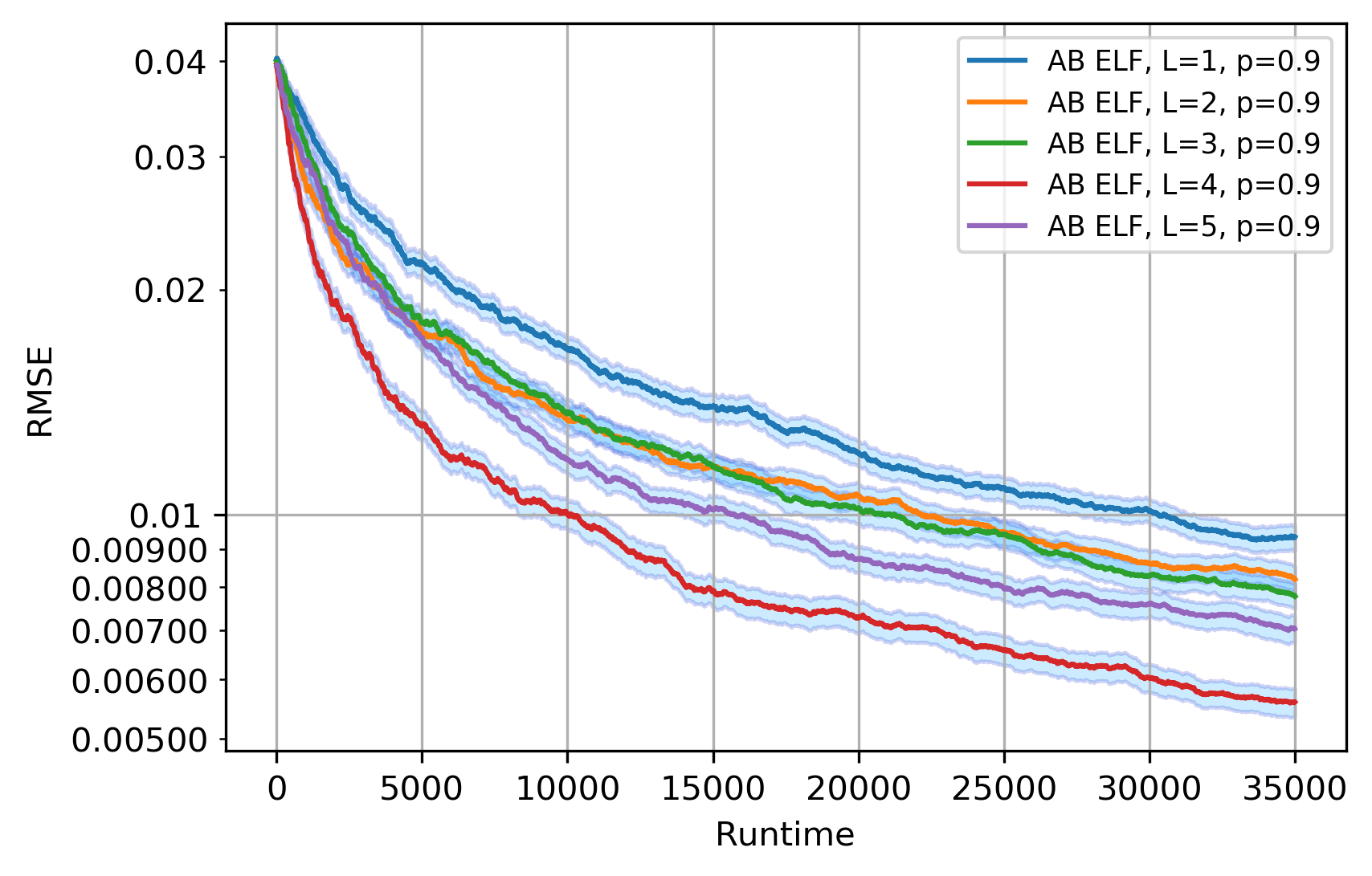}
\end{minipage}
\begin{minipage}{.48\textwidth}
\center
\includegraphics[width=0.95\linewidth]{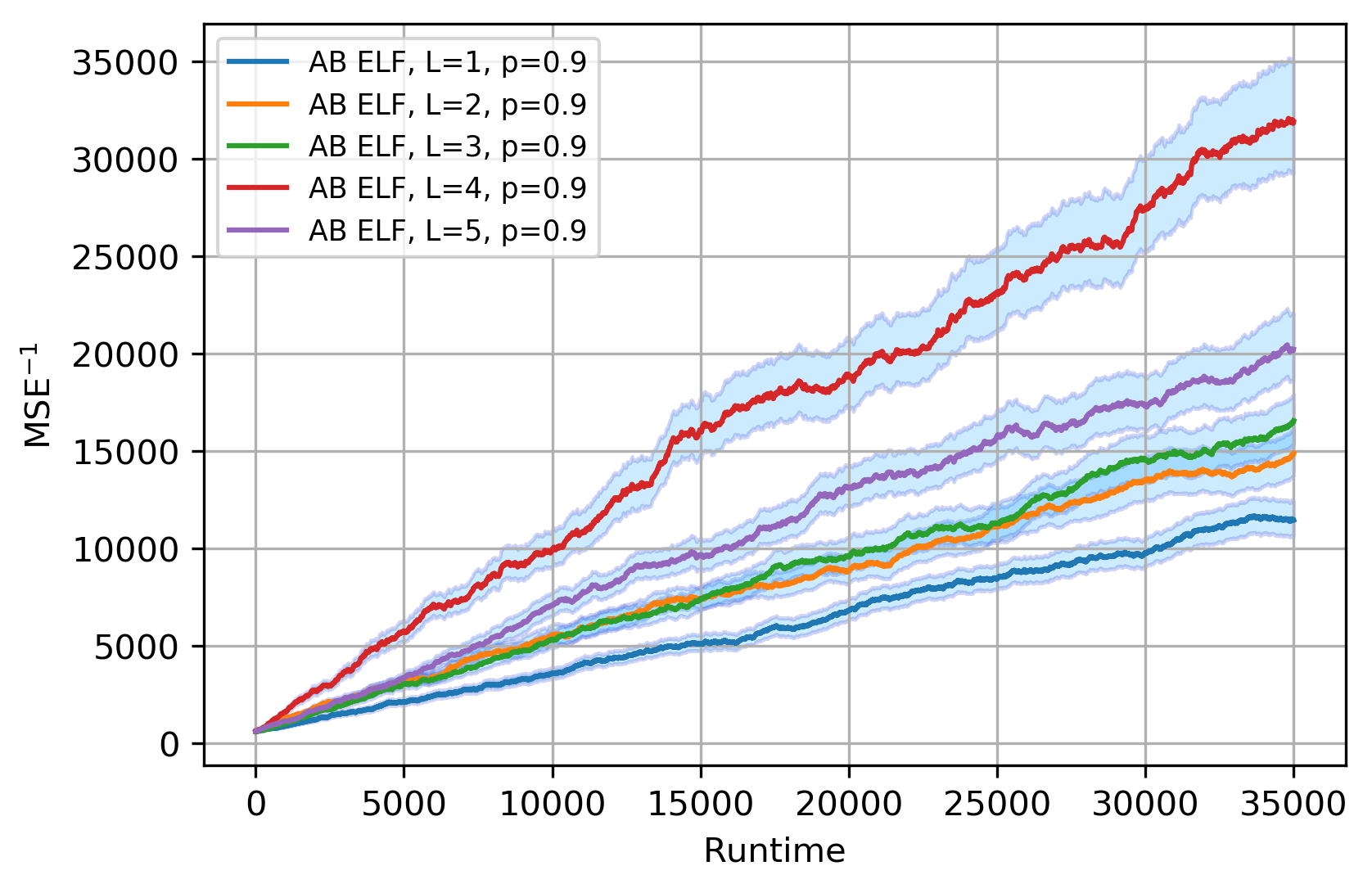}
\end{minipage}
\caption{This figure demonstrates the impact of circuit depth on the performance of AB ELF. Here $\Pi$ has true value $0.35$ and prior distribution $\mathcal{N}(0.39, 0.0016)$, the layer fidelity $p$ is $0.9$, and the number $L$ of layers is varied from $1$ to $5$. Note that the best performance is achieved by $L=4$ instead of $L=5$. }
\label{fig:circuit_depth_on_ab_elf}
\end{figure*}

\begin{figure*}[!ht]
\begin{minipage}{.48\textwidth}
\center
\includegraphics[width=0.95\linewidth]{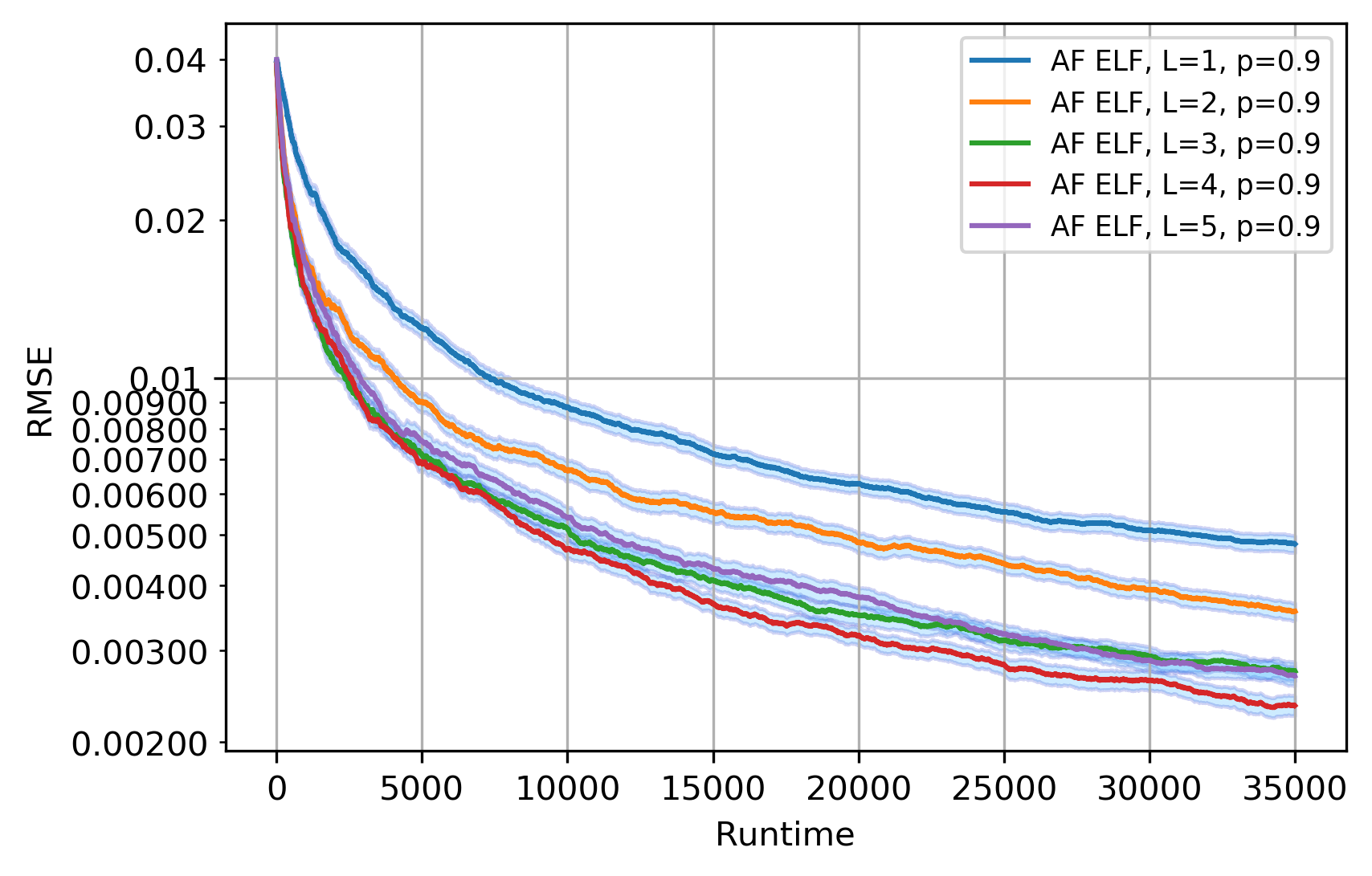}
\end{minipage}
\begin{minipage}{.48\textwidth}
\center
\includegraphics[width=0.95\linewidth]{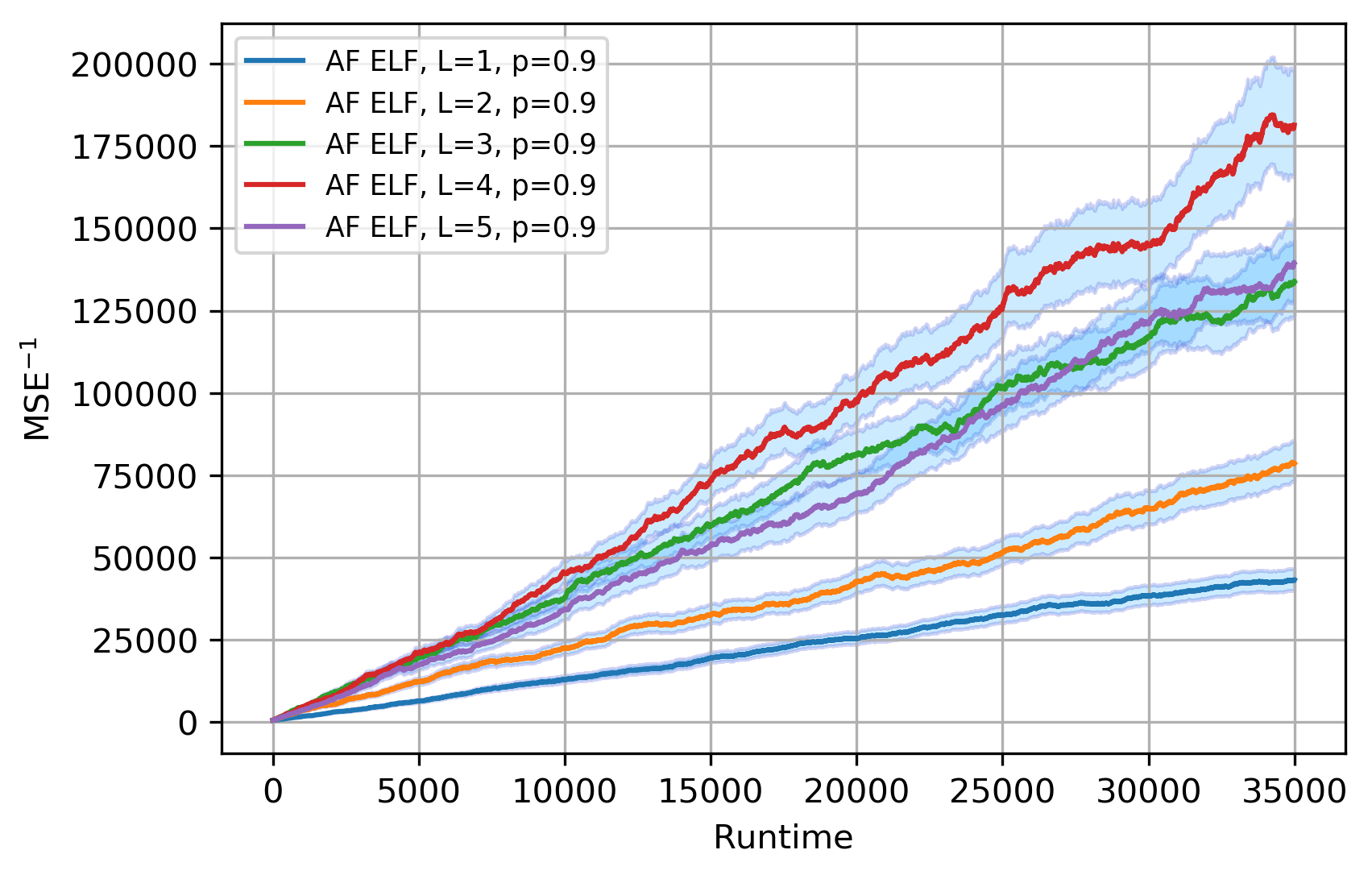}
\end{minipage}
\caption{This figure demonstrates the impact of circuit depth on the performance of AF ELF. Here $\Pi$ has true value $0.35$ and prior distribution $\mathcal{N}(0.39, 0.0016)$, the layer fidelity $p$ is $0.9$, and the number $L$ of layers is varied from $1$ to $5$. Note that the best performance is achieved by $L=4$ instead of $L=5$.}
\label{fig:circuit_depth_on_af_elf}
\end{figure*}

\begin{figure*}[!ht]
\begin{minipage}{.48\textwidth}
\center
\includegraphics[width=0.95\linewidth]{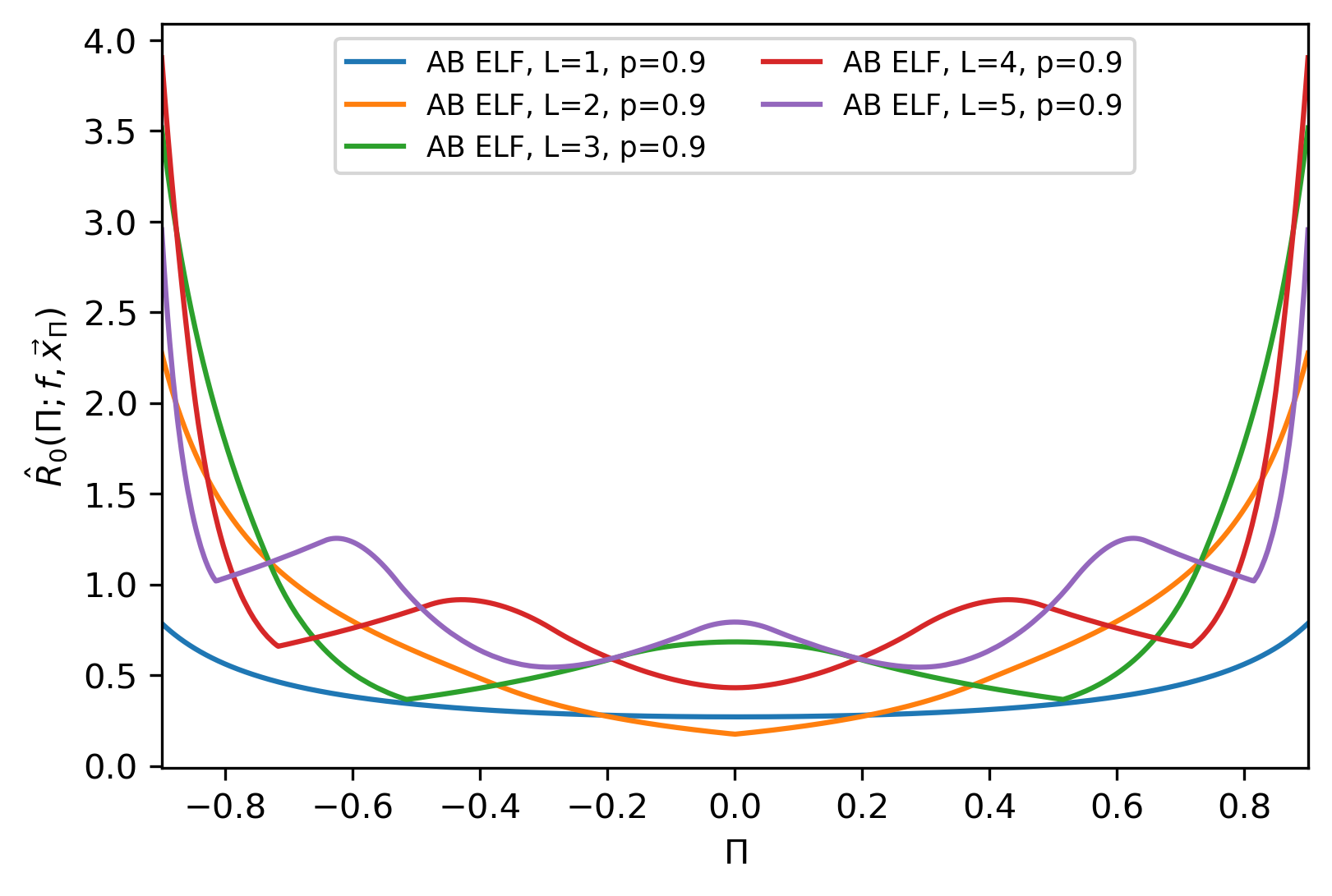}
\end{minipage}
\begin{minipage}{.48\textwidth}
\center
\includegraphics[width=0.95\linewidth]{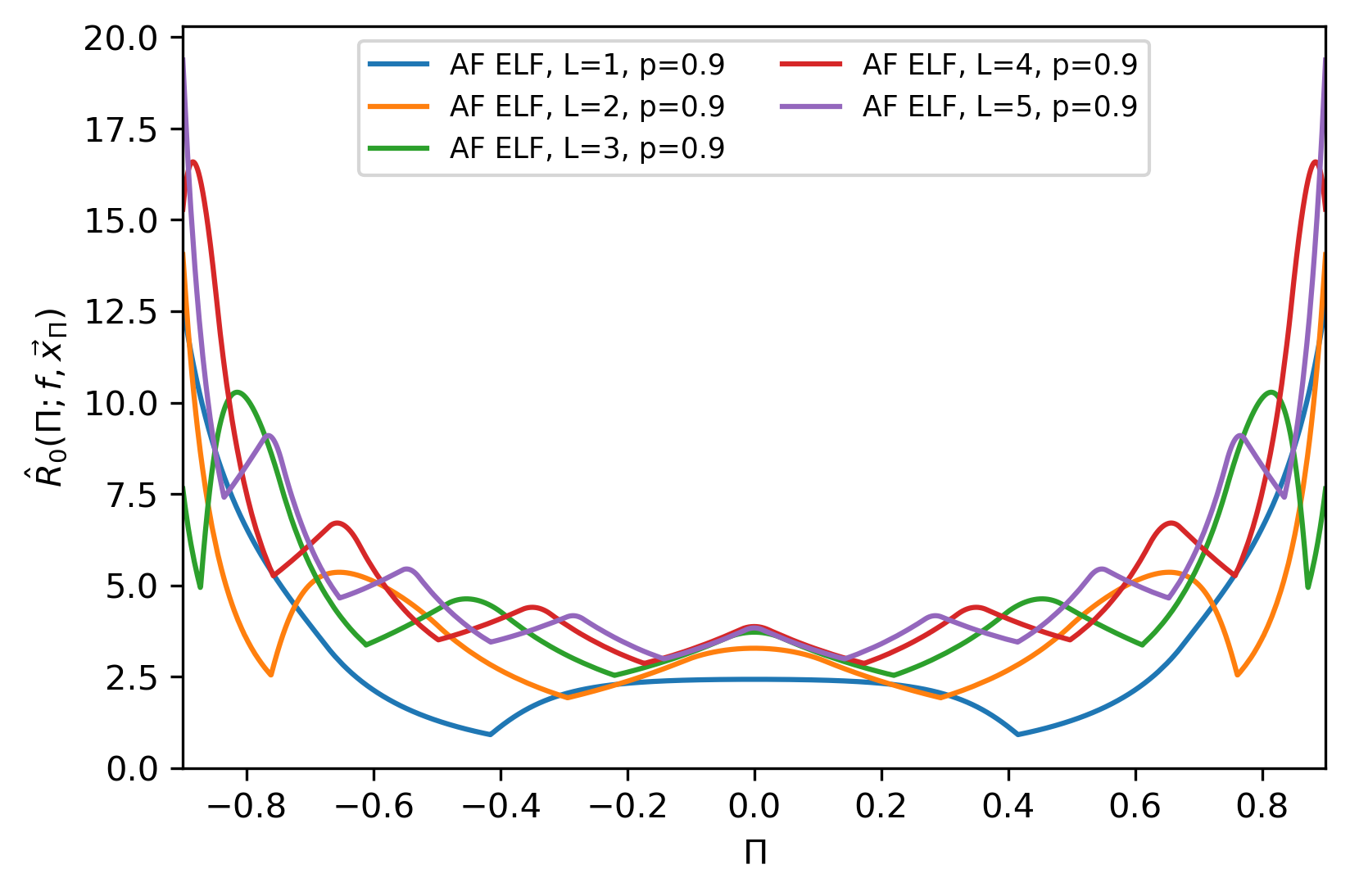}
\end{minipage}
\caption{This figure shows the $\hat{R}_0$ factors of AB ELF and AF ELF for $\Pi \in [-0.9, 0.9]$, when the number $L$ of circuit layers is varied from $1$ to $5$ and the layer fidelity $p$ is fixed to $0.9$. Here $\vec x_{\Pi} \in \mathbb{R}^{2L}$ is a global maximum point of $\hat{R}_0(\Pi; f, \vec x)$ for given $\Pi$ and $f=p^L$. Note that the $L$ that maximizes the $\hat{R}_0$ factor depends heavily on the value of $\Pi$ in both the ancilla-based and ancilla-free cases.}
\label{fig:R0-various-depths}
\end{figure*}

\section{A model for noisy algorithm performance}
\label{sec:runtimemodel}

Our aim is to build a model for the runtime needed to achieve a target mean-squared error in the estimate of $\Pi$ as it is scaled to larger systems and run on devices with better gate fidelities.
This model will be built on two main assumptions.
The first is that the growth rate of the inverse mean squared error is well-described by half the inverse variance rate expression (c.f. Eq.~(\ref{eq:fisherrate}) ).
The variance contribution to the MSE is the variance in the estimator, which is not necessarily the same as the posterior variance.
Appendix \ref{sec:evidences} shows that the variance closely tracks the variance of the estimator (c.f. Figure \ref{fig:two_variances}). We will use the posterior variance $\sigma^2_{i+1}$ in place of the estimator variance.
The half is due to the conservative estimate that the variance and squared bias contribute equally to the mean squared error. In Appendix \ref{sec:evidences} we show evidence supporting this assumption of a small bias (c.f. Figure  \ref{fig:bias_vs_std}).
The second assumption is an empirical lower bound on the variance reduction factor, which is motivated by numerical investigations of the Chebyshev likelihood function.

We carry out analysis for the MSE with respect to the estimate of $\theta$.
We will then convert the MSE of this estimate to an estimate of MSE with respect to $\Pi$.
Our strategy will be to integrate upper and lower bounds for the rate expression $R(\mu, \sigma; f, m)$ in Eq.~(\ref{eq:fisherrate}) to arrive at bounds for inverse MSE as a function of time.
To help our analysis we make the substitution $m=T(L)=2L+1$ and reparameterize the way noise is incorporated by introducing $\lambda$ and $\alpha$ such that $f^2=\bar{p}^2p^{2L}=e^{-\lambda(2L+1)-\alpha}=e^{-\lambda m-\alpha}$.

The upper and lower bounds on this rate expression are based on findings for the Chebyshev likelihood functions, where $\vec{x}=(\frac{\pi}{2})^{\oplus 2L}:=(\frac{\pi}{2}, \frac{\pi}{2}, \dots, \frac{\pi}{2}) \in \mathbb{R}^{2L}$.
Since the Chebyshev likelihood functions are a subset of the engineered likelihood functions, a lower bound on the Chebyshev performance gives a lower bound on the ELF performance.
We leave as a conjecture that the upper bound for this rate in the case of ELF is a small multiple (e.g. 1.5) of the upper bound we have established for the Chebyshev rate.

The Chebyshev upper bound is established as follows.
For fixed $\sigma$, $\lambda$, and $m$, one can show\footnote{
For the Chebyshev likelihood functions, we can express the variance reduction factor as $\mathcal V(\mu,\sigma; f,
    (\tfrac \pi 2)^{\oplus 2L}
    )=m_L^2/\left(1+\left(f^{-2} \e^{m_L^2 \sigma^2}-1\right) \csc^2(m_L\mu)\right)$
whenever $\sin(m_L\mu) \neq 0$. Then, $\csc^2(m_L \mu) \geq 1$ implies that $\mathcal V(\mu,\sigma; f,
    (\tfrac \pi 2)^{\oplus 2L}
    ) \leq f^2 m_L^2 \e^{-m_L^2\sigma^2}$. Here $(\frac{\pi}{2})^{\oplus 2L}=(\frac{\pi}{2}, \frac{\pi}{2}, \dots, \frac{\pi}{2}) \in \mathbb{R}^{2L}$.} that the variance reduction factor achieves a maximum value of $\mathcal{V}=m^2 \exp(-m^2\sigma^2-\lambda m -\alpha)$, occurring at $\mu=\pi/2$.
This expression is less than $m^2 \e^{-m^2 \sigma^2}$, which achieves a maximum of $(\e\sigma^2)^{-1}$ at $m =\tfrac 1{\sigma}$.
Thus, the factor $1/(1-\sigma^2\mathcal{V})$ cannot exceed $1/(1-e^{-1})\approx 1.582$.
Putting this all together, for fixed $\sigma$, $\lambda$, and $m$, the maximum rate is upper bounded as $R(\mu, \sigma; \lambda, \alpha, m)\leq  \frac{em}{e-1}\exp(-m^2\sigma^2-\lambda m -\alpha)$.
This follows from the fact that $R$ is monotonic in $\mathcal{V}$ and that $\mathcal{V}$ is maximized at $\mu=\pi/2$.
In practice, we will aim to choose a value of $L$ that maximizes the inverse variance rate.
The rate achieved by discrete $L$ cannot exceed the value we obtain when optimizing the above upper bound over continuous value of $m$.
This optimal value is realized for $1/m=\frac{1}{2}\left(\sqrt{\lambda^2+8\sigma^2}+\lambda\right)$.
We define $\bar{R}(\sigma; \lambda, \alpha)$ by evaluating $R(\pi/2, \sigma; \lambda, \alpha, m)$ at this optimum value,
\begin{widetext}
\begin{align}
    \label{eq:optrate}
    \bar{R}(\sigma;\lambda,\alpha)=\frac{2e^{-\alpha-1}}{\sqrt{\lambda^2+8\sigma^2}+\lambda}\exp\left(\frac{2\sigma^2}{4\sigma^2+\lambda^2+\lambda^2\sqrt{8\sigma^2/\lambda^2+1}}\right),
\end{align}
\end{widetext}
which gives the upper bound on the Chebyshev rate
\begin{align}
    R_C^*(\mu, \sigma; \lambda, \alpha) = \max_L R(\mu, \sigma; \lambda, \alpha, m) \leq \frac{e}{e-1}\bar{R}(\sigma;\lambda, \alpha).
\end{align}
We do not have an analytic lower bound on the Chebyshev likelihood performance.
We can establish an empirical lower bound based on numerical checks.
For any fixed $L$, the inverse variance rate is zero at the $2L+2$ points $\mu \in \{0, \pi/(2L+1), 2\pi/(2L+1), \ldots, 2L\pi/(2L+1), \pi\}$.
Since the rate is zero at these end points for all $L$, the global lower bound on $R_C^*$ is zero.
However, we are not concerned with the poor performance of the inverse variance rate near these end points.
When we convert the estimator from $\hat{\theta}$ to $\hat{\Pi}=\cos \hat{\theta}$, the information gain near these end point actually tends to a large value.
For the purpose of establishing useful bounds, we will restrict $\mu$ to be in the range $[0.1\pi, 0.9\pi]$.
In the numerical tests
\footnote{We searched over a uniform grid of 50000 values of $\theta$, $L$ values from $L^*/3$ to $3L^*$, where $L^*$ is to the optimized value used to arrive at Eq. (\ref{eq:optrate}), and $\sigma$ and $\lambda$ ranging over $[10^{-1}, 10^{-2}, \ldots, 10^{-5}]$. For each $(\sigma, \lambda)$ pair we found the $\theta$ for which the maximum inverse variance rate (over $L$) is a minimum. For all $(\sigma,\lambda)$ pairs checked, this worst-case rate was always between $0.4$ and $0.5$, with the smallest value found being $R=0.41700368\geq (e-1)^2/e^2$.}
we find that for all $\mu\in[0.1\pi, 0.9\pi]$, there is always a choice of $L$ for which the inverse variance rate is above $(e-1)^2/e^2\approx 0.40$ times the upper bound.
Putting these together, we have
\begin{align}
\label{eq:ratebounds}
\frac{e-1}{e}\bar{R}(\sigma; \lambda, \alpha) \leq R_C^*(\mu, \sigma; \lambda, \alpha) \leq \frac{e}{e-1}\bar{R}(\sigma;\lambda,\alpha).
\end{align}
It is important to note that by letting $m$ be continuous, certain values of $\sigma$ and $\lambda$ can lead to an optimal $m$ for which $L=(m-1)/2$ is negative.
Therefore, these results apply only in the case that $\lambda\leq 1$, which ensures that $m\geq1$. We expect this model to break down in the large-noise regime (i.e. $\lambda \geq 1$).

For now, we will assume that the rate tracks the geometric mean of these two bounds, i.e. $R^*_C(\sigma, \lambda, \mu)=\bar{R}(\sigma, \lambda)$, keeping in mind that the upper and lower bounds are small constant factors off of this.
We assume that the inverse variance grows continuously in time at a rate given by the difference quotient expression captured by the inverse-variance rate, $R^*=\frac{\d}{\d t}\frac{1}{\sigma^2}$. Letting $F=1/\sigma^2$ denote this inverse variance, the rate equation above can be recast as a differential equation for $F$,
\begin{widetext}
\begin{align}
    \frac{\d F}{\d t} = \frac{2e^{-\alpha-1}}{\lambda\sqrt{1+8/(F\lambda^2)}+\lambda}\exp\left(\frac{2}{4+\lambda^2F+\lambda^2F\sqrt{1+8/(F\lambda^2)}}\right).
    \label{eq:diffeq_df_dt}
\end{align}
\end{widetext}
Through this expression, we can identify both the Heisenberg limit behavior and shot-noise limit behavior.
For $F\ll 1/\lambda^2$, the differential equation becomes
\begin{align}
    \frac{\d F}{\d t} = \frac{e^{-\alpha-1/2}}{\sqrt{2}}\sqrt{F},
\end{align}
which integrates to a quadratic growth of the inverse squared error $F(t) \sim t^2$.
This is the signature of the Heisenberg limit regime.
For $F\gg 1/\lambda^2$, the rate approaches a constant,
\begin{align}
    \frac{\d F}{\d t} = \frac{e^{-\alpha-1}}{\lambda}.
\end{align}
This regime yields a linear growth in the inverse squared error $F(t) \sim t$, indicative of the shot-noise limit regime.

In order to make the integral tractable, we can replace the rate expression with integrable upper and lower bound expressions (to be used in tandem with our previous bounds).
Letting $x=\lambda^2 F$, these bounds are re-expressed as,
\begin{align}
    \frac{2e^{-\alpha-1}\lambda}{1+\frac{1}{\sqrt{12x}}+\frac{(x+4)}{\sqrt{x^2+8x}}} \geq \frac{\d x}{\d t} \geq \frac{2e^{-\alpha-1}\lambda}{1+\frac{1}{\sqrt{4x}}+\frac{(x+4)}{\sqrt{x^2+8x}}}.
\end{align}
From the upper bound we can establish a lower bound on the runtime, by treating time as a function of $x$ and integrating,
\begin{align}
    \int_{0}^{t} dt &\geq \int_{x_0}^{x_f}\d x\frac{e^{\alpha + 1}}{2\lambda}\left(1+\dfrac{1}{\sqrt{12x}}+\dfrac{x+4}{\sqrt{x^2+8x}}\right)\\
    &=\frac{e^{\alpha + 1}}{2\lambda}\left (x_f+\sqrt{\dfrac{x_f}{3}}+\frac{1}{2}\sqrt{x_f^2+8x_f}-x_0 \right.
    \nonumber\\
    &\quad \left.-\sqrt{\dfrac{x_0}{3}}-\frac{1}{2}\sqrt{x_0^2+8x_0}\right).
\end{align}
Similarly, we can use the lower bound to establish an upper bound on the runtime.
Here we introduce our assumption that, in the worst case, the MSE of the phase estimate $\varepsilon^2_{\theta}$ is twice the variance (i.e. the variance equals the bias),
so the variance must reach half the MSE:
$\sigma^2=\varepsilon_{\theta}^2/2=\lambda^2/x$.
In the best case, we assume the bias in the estimate is zero and set $\varepsilon_{\theta}^2=\lambda^2/x$.
We combine these bounds with the upper and lower bounds of Eq.~(\ref{eq:ratebounds}) to arrive at the bounds on the estimation runtime as a function of target MSE,
\begin{widetext}
\begin{align}
\label{eq:runtimebounds}
   (e-1)\frac{e^{-\lambda}}{2\bar{p}^2}\left(\frac{ \lambda }{\varepsilon_{\theta}^2} + \frac{1}{\sqrt{3}\varepsilon_{\theta}} +\sqrt{\left(\frac{\lambda}{\varepsilon_{\theta}^2}\right)^2+\left(\frac{2\sqrt{2}}{\varepsilon_{\theta}}\right)^2}\right) \leq t_{\varepsilon_{\theta}} \leq \frac{e^2}{e-1}\frac{e^{-\lambda}}{\bar{p}^2}\left(\frac{ \lambda }{\varepsilon_{\theta}^2} + \frac{1}{\sqrt{2}\varepsilon_{\theta}} +\sqrt{\left(\frac{\lambda}{\varepsilon_{\theta}^2}\right)^2+\left(\frac{2\sqrt{2}}{\varepsilon_{\theta}}\right)^2}\right),
\end{align}
\end{widetext}
where $\theta\in[0.1\pi, 0.9\pi]$.

At this point, we can convert our phase estimate $\hat{\theta}$ back into the amplitude estimate $\hat{\Pi}$.
The MSE with respect to the amplitude estimate $\varepsilon^2_{\Pi}$ can be approximated in terms of the phase estimate MSE as
\begin{align}
    \varepsilon^2_{\Pi}&=\mathbb{E}(\hat{\Pi}-\Pi)^2\nonumber\\
    &=\mathbb{E}(\cos\hat{\theta}-\cos\theta)^2\nonumber\\
    &\approx\mathbb{E}((\hat{\theta}-\theta)\frac{\d \cos\theta}{\d \theta})^2\nonumber\\
    &=\varepsilon^2_{\theta}\sin^2\theta,
\end{align}
where we have assumed that the distribution of the estimator is sufficiently peaked about $\theta$ to ignore higher-order terms.
This leads to $\varepsilon^2_{\theta}=\varepsilon^2_{\Pi}/(1-\Pi^2)$, which can be substituted into the above expressions for the bounds, which hold for $\Pi\in [\cos0.9\pi, \cos 0.1\pi] \approx [-0.95, 0.95]$.
Dropping the estimator subscripts (as they only contribute constant factors), we can establish the runtime scaling in the low-noise and high-noise limits,
\begin{align}
    t_{\varepsilon}=\begin{cases}
    O(e^{\alpha}/\varepsilon) & \lambda\ll \varepsilon,\\
    O(e^{\alpha}\lambda/\varepsilon^2) & \lambda\gg \varepsilon,
    \end{cases}
\end{align}
observing that the Heisenberg-limit scaling and shot-noise limit scaling are each recovered.

We arrived at these bounds using properties of Chebyshev likelihood functions.
As we have shown in the previous section, by engineering likelihood functions, in many cases we can reduce estimation runtimes.
Motivated by our numerical findings of the variance reduction factors of engineered likelihood functions (see, e.g.
Figure \ref{fig:R0-L_6_p_0.9}), we conjecture that using engineered likelihood functions increases the worst case inverse-variance rate in Eq.~(\ref{eq:ratebounds}) to $\bar{R}(\sigma; \lambda, \alpha)\leq R^*_C(\mu, \sigma; \lambda, \alpha)$.

In order to give more meaning to this model, we will refine it to be in terms of number of qubits $n$ and two-qubit gate fidelities $f_{2Q}$.
We consider the task of estimating the expectation value of a Pauli string $P$ with respect to state $\ket{A}$.
Assume that $\Pi=\bra{A}P\ket{A}$ is very near zero so that $\varepsilon^2=\varepsilon^2_{\Pi}\approx\varepsilon^2_{\theta}$.
Let the two-qubit gate depth of each of the $L$ layers be $D$.
We model the total layer fidelity as $p=f_{2Q}^{nD/2}$, where we have ignored errors due to single-qubit gates.
From this, we have $\lambda = \frac{1}{2}nD \ln (1/f_{2Q})$
and $\alpha = 2\ln (1/\bar{p})-\frac{1}{2}nD\ln(1/f_{2Q})$.
Putting these together and using the lower bound expression in Eq. \eqref{eq:runtimebounds}, we arrive at the runtime expression,
\begin{widetext}
\begin{align}
    t_{\varepsilon} = e\frac{f_{2Q}^{nD/2}}{2\bar{p}^2}\left(\frac{ nD\ln (1/f_{2Q}) }{2\varepsilon^2} + \frac{1}{\sqrt{3}\varepsilon} +\sqrt{\left(\frac{nD\ln (1/f_{2Q})}{2\varepsilon^2}\right)^2+\left(\frac{2\sqrt{2}}{\varepsilon}\right)^2}\right).
\end{align}
\end{widetext}

Finally, we will put some meaningful numbers in this expression and estimate the required runtime in seconds as a function of two-qubit gate fidelities.
To achieve quantum advantage we expect that the problem instance will require on the order of $n=100$ logical qubits and that the two-qubit gate depth is on the order of the number of qubits, $D=200$.
Furthermore, we expect that target accuracies $\varepsilon$ will need to be on the order of $\varepsilon = 10^{-3}$ to $10^{-5}$.
The runtime model measures time in terms of ansatz circuit durations.
To convert this into seconds we assume each layer of two-qubit gates will take time $G=10^{-8}$s, which is an optimistic assumption for today's superconducting qubit hardware.
Figure \ref{fig:runtime_model} shows this estimated runtime as a function of two-qubit gate fidelity.

\begin{figure}[!ht]
\includegraphics[width=0.95\linewidth]{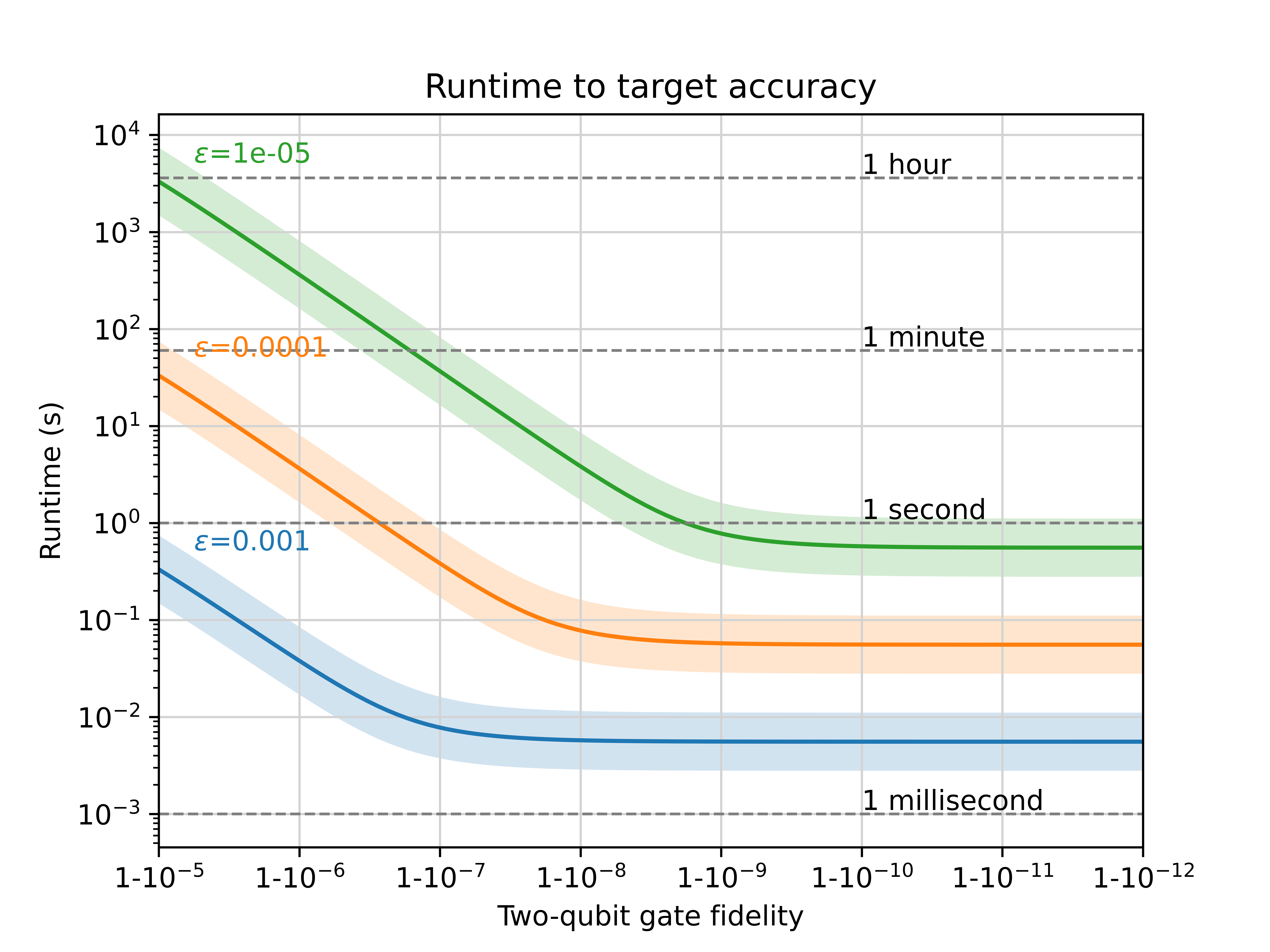}
\caption{As two-qubit gate fidelities are improved, deeper enhanced sampling circuits warrant being implemented, yielding shorter estimation runtimes. Here, we consider the case of $n=100$ qubits, $D=200$ two-qubit gate depth per layer, and target accuracies of $\varepsilon=10^{-3}$, $\varepsilon=10^{-4}$, and $\varepsilon=10^{-5}$. The bands indicate the upper and lower bounds of Eq. \eqref{eq:runtimebounds}.}
\label{fig:runtime_model}
\end{figure}

The two-qubit gate fidelities required to reduce runtimes into a practical region will most likely require error correction.
Performing quantum error correction requires an overhead that increases these runtimes.
In designing quantum error correction protocols, it is essential that the improvement in gate fidelities is not outweighed by the increase in estimation runtime.
The proposed model gives a means of quantifying this trade-off:
the product of gate infidelity and (error-corrected) gate time should decrease as useful error correction is incorporated.
In practice, there are many subtleties that should be accounted for to make a more rigorous statement.
These include considering the variation in gate fidelities among gates in the circuit and the varying time costs of different types of gates.
Nevertheless, the cost analyses afforded by this simple model may a useful tool in the design of quantum gates, quantum chips, error correcting schemes, and noise mitigation schemes.

\section{Outlook}
\label{sec:outlook}
This work was motivated by the impractical runtimes required by many NISQ-amenable quantum algorithms.
We aimed to improve the performance of estimation subroutines that have relied on standard sampling, as used in VQE.
Drawing on the recent alpha-VQE \cite{wang2019accelerated} and quantum metrology \cite{giovannetti2006quantum}, we investigated the technique of enhanced sampling to explore the continuum between standard sampling and quantum amplitude (or phase) estimation.
In this continuum, we can make optimal use of the quantum coherence available on a given device to speed up estimation.
Similar to standard sampling in VQE, enhanced sampling does not require ancilla qubits.
Quantum advantage for tasks relying on estimation will likely occur within this continuum rather than at one of the extremes.

Our central object of study was the quantum generated likelihood function, relating measurement outcome data to a parameter of interest encoded in a quantum circuit.
We explored engineering likelihood functions to optimize their statistical power.
This led to several insights for improving estimation.
First, we should incorporate a well-calibrated noise model directly in the likelihood function to make inference robust to certain error.
Second, we should choose a circuit depth (reflected in the number of enhanced sampling circuit layers) that balances gain in statistical power with accrual of error.
Finally, we should tune generalized reflection angles to mitigate the effect of ``deadspots'' during the inference process.

We used engineered likelihood functions to carry out adaptive approximate Bayesian inference for parameter estimation.
Carrying out this process in simulation required us to build mathematical and algorithmic infrastructure.
We developed mathematical tools for analyzing a class of quantum generated likelihood functions.
From this analysis, we proposed several optimization algorithms for tuning circuit parameters to engineer likelihood functions.
We investigated the performance of estimation using engineered likelihood functions and compared this to estimation using fixed likelihood functions.
Finally, we proposed a model for predicting the performance of enhanced sampling estimation algorithms as the quality of quantum devices is improved.

These simulations and the model led to several insights.
As highlighted in Section \ref{subsec:compare_schemes}, for the degree of device error expected in the near term (two-qubit gate fidelities of $\sim 99.92\%$), we have shown that enhanced sampling and engineered likelihood functions can be used to outperform standard sampling used in VQE.
Furthermore, these simulations suggest a non-standard perspective on the tolerance of error in quantum algorithm implementations.
We found that, for fixed gate fidelities, the best performance is achieved when we push circuit depths to a point where circuit fidelities are around the range of $0.5-0.7$. 
This suggests that, compared to the logical circuit fidelities suggested in other works (e.g. $0.99$ in \cite{Babbush2018}), we can afford a 50-fold increase in circuit depth.
We arrive at this balance between fidelity and statistical power by taking estimation runtime to be the cost to minimize.

The runtime model developed in Section \ref{sec:runtimemodel} sheds light on the trade-off between gate times and gate fidelity for estimation.
For gate times that are one-thousand times slower, the gate fidelities must have three more nines to achieve the same estimation runtimes.
The runtime model gives insight on the role of quantum error correction in estimation algorithms.
Roughly, we find that for quantum error correction to be useful for estimation, the factor of increase in runtime from error correction overhead must be less than the factor of decrease in logical gate error rates.
Additionally, the runtime model predicts that for a given estimation task, there is a level of logical gate fidelity beyond which further improvements effectively do not reduce runtimes (espectially if time overhead is taken into account).
For the 100-qubit example considered, seven nines in two-qubit gate fidelities sufficed.

We leave a number of questions for future investigation.
In VQE a set of techniques referred to as ``grouping'' are used to reduce the measurement count \cite{kandala2017hardware, verteletskyi2020measurement, izmaylov2019unitary, crawford2019efficient, zhao2019measurement}.
These grouping techniques allow sampling of multiple operators at once, providing a type of measurement parallelization.
The grouping method introduced in \cite{izmaylov2019unitary, zhao2019measurement} decomposes a Pauli Hamiltonian into sets of mutually-anticommuting Pauli strings, which ensures that the sum within each set is a Hermitian reflection.
This method of grouping is compatible with enhanced sampling, as the Hermitian reflections can be both measured and implemented as generalized reflections (i.e. an example of operator $P$).
However, it remains to explore if the additional circuit depth incurred by implementing these reflections is worth the variance reduction in the resulting estimators.
Beyond existing grouping techniques, we anticipate opportunities for parallelizing measurements that are dedicated to enhanced sampling.

Our work emphasizes the importance of developing accurate error models at the algorithmic level. 
Efficiently learning the ``nuisance parameters'' of these models, or \emph{likelihood function calibration}, will be an essential ingredient to realizing the performance gain of any enhanced sampling methods in the near term.
The motivation is similar to that of randomized benchmarking \cite{hincks2018bayesian}, where measurement data is fit to models of gate noise.
An important problem that we leave for future work is to improve methods for likelihood function calibration.
Miscalibration can introduce a systematic bias in parameter estimates.
A back-of-the-envelope calculation predicts that the relative error in estimation due to miscalibration bias is inversely proportional to the number of Grover iterates $L$ (which is set proportionally to $\lambda$) and proportional to the absolute error in the likelihood function.
Assuming that this absolute error grows sublinearly in $L$, we would expect a degree of robustness in the estimation procedure as $L$ is increased.
In future work we will explore in more detail to what precision we must calibrate the likelihood function so that the bias introduced by miscalibration (or model error) is negligible.
Finally, we leave investigations into analytical upper and lower bounds on estimation performance to future work.

We have aimed to present a viable solution to the ``measurement problem'' \cite{gonthier2020identifying} that plagues VQE.
It is likely that these methods will be needed to achieve quantum advantage for problems in quantum chemistry and materials.
Furthermore, such amplitude estimation techniques may help to achieve quantum advantage for applications in finance and machine learning tasks as well.
We hope that our model for estimation performance as a function of device metrics is useful in assessing the relative importance of a variety of quantum resources including qubit number, two-qubit gate fidelity, gate times, qubit stability, error correction cycle time, readout error rate, and others.

\section*{Acknowledgments}

We would like to thank Pierre-Luc Dallaire-Demers, Amara Katabarwa, Jhonathan Romero, Max Radin, Peter Love, Yihui Quek, Jonny Olson, Hannah Sim, and Jens Eisert for insightful conversations and valuable feedback, and Christopher Savoie for grammatical tidying.

\appendix

\section{A sufficient model of circuit noise}
\label{sec:app_noise_model}
Here we describe one circuit noise model which yields the likelihood function noise model of Eq.~\eqref{eq:noisy_elf}.
We expect that this circuit noise model is too simplistic to describe the output density matrix of the circuit.
However, we also expect that physically realistic variants of this circuit noise model will also lead to an exponentially decaying likelihood function bias.
In other words, this circuit noise model is sufficient, but not necessary for yielding the likelihood function noise model used in this work.
The circuit noise model assumes that the noisy version of each circuit layer $V(x_{2j})U(\theta; x_{2j-1})$ implements a mixture of the target operation and the completely depolarizing channel acting on the same input state, i.e.
\begin{align}
\mathcal{U}_j(\rho) &= p V(x_{2j})U(\theta; x_{2j-1})\rho U^\dagger(\theta; x_{2j-1}) V^\dagger(x_{2j})\nonumber\\
&\quad +(1-p)\dfrac{I}{2^n},
\end{align}
where $p$ is the \emph{fidelity} of this layer. Under composition of such imperfect operations, the output state of the $L$-layer circuit becomes
\begin{align}
\rho_L = p^L Q(\theta; \vec x)\ketbra{A} Q^\dagger(\theta; \vec x)+(1-p^L)\dfrac{\identity}{2^n}.
\end{align}
This imperfect circuit is preceded by an imperfect preparation of $\ket{A}$ and followed by an imperfect measurement of $P$. In the context of randomized benchmarking, such errors are referred to as \emph{state preparation and measurement} (SPAM) errors \cite{gambetta2012characterization}.
We will also model SPAM error with a depolarizing model, taking the noisy preparation of $\ket{A}$ to be $p_{SP}\ketbra{A}+(1-p_{SP})\frac{I}{2^n}$ and taking the noisy measurement of $P$ to be the POVM $\{p_{M} \frac{I+P}{2} + (1-p_{M}) \frac{I}{2}, p_{M} \frac{I-P}{2} + (1-p_{M}) \frac{I}{2}\}$. Combining the SPAM error parameters into $\bar{p}=p_{SP}p_{M}$, we arrive at a model for the noisy likelihood function
\begin{align}
\mathbb{P}(d|\theta; f, \vec x)=\frac 12 \left[ 
1+(-1)^d \bar{p}p^L\Delta(\theta; \vec x) \right],
\label{eq:app_noisy_elf}
\end{align}
where $f=\bar{p}p^L$ is the \emph{fidelity} of the whole process for generating the ELF, and $\Delta(\theta,\vec{x})$ is the bias of the ideal likelihood function as defined in Eq.~(\ref{eq:delta_af}).

\section{Proof of Lemma \ref{lem:eval_coeff_func_af}} 
\label{sec:eval_coeff_func_af}
In this appendix, we prove Lemma \ref{lem:eval_coeff_func_af}, which states that the CSBD coefficient functions of $\Delta(\theta; \vec x)$ and $\Delta'(\theta; \vec x)$ with respect to $x_j$ can be evaluated in $O(L)$ time for any $j \in \{1,2,\dots, 2L\}$. 

For convenience, we introduce the following notation. Let $W_{2i}=U^\dagger (\theta;x_{2i+1})=U(\theta;-x_{2i+1})$, $W_{2i+1}=V^\dagger(x_{2i+2}) = V(-x_{2i+2})$, $W_{4L-2i}=U(\theta;x_{2i+1})$, and $W_{4L-2i-1}=V(x_{2i+2})$, 
for $i=0,1,\dots, L-1$, and $W_{2L}=P(\theta)$. Furthermore, let $W'_j = \partial_{\theta} W_j$ for $j=0, 1, \dots, 4L$. Note that $W'_j=0$ if $j$ is odd. Then we define $P_{a,b}=W_aW_{a+1}\dots W_b$ if $0 \le a \le  b \le 4L$, and $P_{a,b}=I$ otherwise. 

With this notation, Eq.~(\ref{eq:q_theta_vecx}) implies that
\begin{widetext}
\begin{align}
Q^\dagger(\theta;\vec x) = P_{0, a-1} W_a P_{a+1, 2L-1}, &\quad \forall 0 \le a \le 2L-1, 
\label{eq:qdthetax_af}
\end{align}
\begin{align}
Q(\theta;\vec x) = P_{2L+1, b-1} W_b P_{b+1, 4L}, &\quad \forall 2L+1 \le b \le 4L, 
\label{eq:qthetax_af}
\end{align}
\begin{align}
Q^\dagger(\theta;\vec x)  P(\theta)Q(\theta;\vec x) = P_{0, a-1} W_a P_{a+1, b-1} W_b P_{b+1, 4L}, &\quad \forall 0 \le a < b \le 4L. 
\label{eq:qdpqthetax_af}
\end{align}
Moreover, taking the partial derivative of Eq.~(\ref{eq:q_theta_vecx}) with respect to $\theta$ yields
\begin{align}
Q'(\theta; \vec x)  &= \frac{\partial Q(\theta;\vec x)}{\partial \theta} \\
&= V(x_{2L})U'(\theta; x_{2L-1}) V(x_{2L-2}) U(\theta; x_{2L-3})\dots V(x_4)U(\theta; x_3)V(x_2)U(\theta; x_1) \nonumber \\
&\quad +
V(x_{2L})U(\theta; x_{2L-1}) V(x_{2L-2}) U'(\theta; x_{2L-3})\dots V(x_4)U(\theta; x_3)V(x_2)U(\theta; x_1) \nonumber \\
&\quad + \dots \nonumber \\
&\quad + 
V(x_{2L})U(\theta; x_{2L-1}) V(x_{2L-2}) U(\theta; x_{2L-3})\dots V(x_4)U'(\theta; x_3)V(x_2)U(\theta; x_1) \nonumber \\
&\quad + 
V(x_{2L})U(\theta; x_{2L-1}) V(x_{2L-2}) U(\theta; x_{2L-3})\dots V(x_4)U(\theta; x_3)V(x_2)U'(\theta; x_1),
\label{eq:derivative_qalphabeta}
\end{align}
where 
\begin{align}
U'(\theta; \alpha) = \frac{\partial U(\theta; \alpha)}{\partial \theta}  
=-\i \sinp{\alpha} P'(\theta)  
= \i \sinp{\alpha} (\sinp{\theta} \bar{Z} - \cosp{\theta} \bar{X})    
\label{eq:derivative_ualpha}
\end{align}
is the partial derivative of $U(\theta; \alpha)$ with respect to $\theta$, in which 
\begin{align}
P'(\theta) = -\sinp{\theta} \bar{Z} + \cosp{\theta} \bar{X}    
\end{align}
is the derivative of $P(\theta)$ with respect to $\theta$. It follows that
\begin{align}
Q'(\theta;\vec x) &= P_{2L+1, 2L+1} W'_{2L+2} P_{2L+3, 4L} 
 +P_{2L+1, 2L+3} W'_{2L+4} P_{2L+5, 4L} 
 + \dots \nonumber \\
&\quad +P_{2L+1, 4L-3} W'_{4L-2} P_{4L-1, 4L} + P_{2L+1, 4L-1} W'_{4L}. 
\label{eq:derivative_qthetax_af}
\end{align}

The following facts will be useful. Suppose $A$, $B$ and $C$ are arbitrary linear operators on the Hilbert space $\mathcal{H}=\mathrm{span}\{\ket{\bar{0}}, \ket{\bar{1}}\}$. Then by direct calculation, one can verify that
\begin{align}
\bra{\bar{0}} A V(-x) B V(x) C \ket{\bar{0}}
&= 
\bra{\bar{0}} A \lrbb{\cosp{x}\bar{I} + \i \sinp{x} \bar{Z}} B \lrbb{\cosp{x}\bar{I} - \i \sinp{x} \bar{Z}} C \ket{\bar{0}}
\\
&= \dfrac{1}{2} [\cosp{2x} \bra{\bar{0}} A\lrb{B-\bar{Z}B\bar{Z}}C \ket{\bar{0}} - \i \sinp{2x} \bra{\bar{0}} A\lrb{B\bar{Z}-\bar{Z}B} C  \ket{\bar{0}} \nonumber \\
&\quad + \bra{\bar{0}} A\lrb{B+\bar{Z}B\bar{Z}}C  \ket{\bar{0}} ],
\label{eq:avbvc}
\end{align}
\begin{align}
\bra{\bar{0}} A U(\theta; -x) B U(\theta; x) C \ket{\bar{0}}
&=
\bra{\bar{0}} A \lrbb{\cosp{x} \bar{I} + \i \sinp{x} P(\theta)} B \lrbb{\cosp{x} \bar{I} - \i \sinp{x} P(\theta)} C \ket{\bar{0}} \\
&= \dfrac{1}{2} [\cosp{2x} \bra{\bar{0}} A\lrb{B  -  P(\theta)BP(\theta)} C \ket{\bar{0}} - \i \sinp{2x} \bra{\bar{0}} A\lrb{BP(\theta) - P(\theta)B} C  \ket{\bar{0}} \nonumber \\
&\quad + \bra{\bar{0}} A\lrb{B  +  P(\theta)BP(\theta)} C \ket{\bar{0}} ],
\label{eq:aubuc}
\end{align}
and
\begin{align}
\bra{\bar{0}} A U(\theta; -x) B U'(\theta; x) C \ket{\bar{0}}
&=\bra{\bar{0}} A \lrbb{\cosp{x}\bar{I} + \i \sinp{x} P(\theta)} B \lrbb{- \i \sinp{x} P'(\theta)} C \ket{\bar{0}} \\
&= \dfrac{1}{2}[-\cosp{2x} \bra{\bar{0}} AP(\theta)BP'(\theta)C\ket{\bar{0}} - \i \sinp{2x} \bra{\bar{0}} A B P'(\theta)C \ket{\bar{0}} \nonumber \\
&\quad + \bra{\bar{0}} A P(\theta) B P'(\theta) C\ket{\bar{0}}].
\label{eq:aubupc}
\end{align}

The following fact will be also useful. Taking the partial derivative of Eq.~(\ref{eq:delta_af}) with respect to $\theta$ yields
\begin{align}
\Delta'(\theta; \vec x) &= \bra{\bar{0}} Q^\dagger(\theta; \vec x) P(\theta) Q'(\theta; \vec x) \ket{\bar{0}}  
+
\bra{\bar{0}} Q^\dagger(\theta; \vec x) P'(\theta) Q(\theta; \vec x) \ket{\bar{0}} +
\bra{\bar{0}} (Q'(\theta; \vec x))^\dagger P(\theta) Q(\theta; \vec x) \ket{\bar{0}} \\
&= 2 ~\mathrm{Re}(\bra{\bar{0}} Q^\dagger(\theta; \vec x) P(\theta) Q'(\theta; \vec x) \ket{\bar{0}}) 
+
\bra{\bar{0}} Q^\dagger(\theta; \vec x) P'(\theta) Q(\theta; \vec x) \ket{\bar{0}}.
\label{eq:derivative_delta_af}
\end{align}

In order to evaluate $C_j(\theta; \vec x_{\neg j})$, $S_j(\theta; \vec x_{\neg j})$, 
$B_j(\theta; \vec x_{\neg j})$,
$C'_j(\theta; \vec x_{\neg j})$, $S'_j(\theta; \vec x_{\neg j})$ 
and $B'_j(\theta; \vec x_{\neg j})$ for given $\theta$ and $\vec x_{\neg j}$, we consider the case $j$ is even and the case $j$ is odd separately.

\begin{itemize}
\item Case 1: $j=2(t+1)$ is even, where $0 \le t \le L-1$. In this case, $W_{2t+1}=V(-x_j)$, and $W_{4L-2t-1}=V(x_j)$. Then by Eqs.~(\ref{eq:delta_af}), (\ref{eq:qdpqthetax_af}) and (\ref{eq:avbvc}), we obtain 
\begin{align}
\Delta(\theta; \vec x) &= \bra{\bar{0}} P_{0, 2t} V(-x_j) P_{2t+2, 4L-2t-2} V(x_j) P_{4L-2t, 4L}  \ket{\bar{0}} \\
&= 
C_j(\theta;\vec x_{\neg j}) \cosp{2x_j} 
+
S_j(\theta;\vec x_{\neg j}) \sinp{2x_j} 
+
B_j(\theta;\vec x_{\neg j}), 
\end{align}
where
\begin{align}
C_j(\theta;\vec x_{\neg j}) &=
\dfrac{1}{2}  \bra{\bar{0}} P_{0, 2t}\lrb{P_{2t+2, 4L-2t-2} - \bar{Z}P_{2t+2, 4L-2t-2}\bar{Z}}P_{4L-2t, 4L} \ket{\bar{0}}, 
\label{eq:cj_af1}
\\
S_j(\theta;\vec x_{\neg j}) &=
- \dfrac{i}{2} 
\bra{\bar{0}} P_{0, 2t} \lrb{P_{2t+2, 4L-2t-2}\bar{Z} - \bar{Z}P_{2t+2, 4L-2t-2}} P_{4L-2t, 4L} \ket{\bar{0}}, 
\label{eq:sj_af1}
\\
B_j(\theta;\vec x_{\neg j}) &= \dfrac{1}{2} \bra{\bar{0}} P_{0, 2t} \lrb{P_{2t+2, 4L-2t-2} + \bar{Z}P_{2t+2, 4L-2t-2}\bar{Z}} P_{4L-2t, 4L} \ket{\bar{0}}.
\label{eq:bj_af1}
\end{align}
Given $\theta$ and $\vec x_{\neg j}$, we first compute $P_{0, 2t}$, $P_{2t+2, 4L-2t-2}$ and $P_{4L-2t, 4L}$ in $O(L)$ time. Then we calculate $C_j(\theta;\vec x_{\neg j})$, $S_j(\theta;\vec x_{\neg j})$ and $B_j(\theta;\vec x_{\neg j})$ by Eqs.~(\ref{eq:cj_af1}-\ref{eq:bj_af1}). This procedure takes only $O(L)$ time. 

Next, we show how to compute $C'_j(\theta;\vec x_{\neg j})$, $S'_j(\theta;\vec x_{\neg j})$ and $B'_j(\theta;\vec x_{\neg j})$. Using Eq.~(\ref{eq:derivative_qthetax_af}) and the fact $P_{a,b}=P_{a, 4L-2t-2} W_{4L-2t-1} P_{4L-2t, b}$ for any $a \le 4L-2t-1 \le b$, we obtain
\begin{align}
Q'(\theta;\vec x) &=
P_{2L+1, 2L+1} W'_{2L+2} P_{2L+3, 4L-2t-2} W_{4L-2t-1} P_{4L-2t, 4L} \nonumber \\
&\quad +P_{2L+1, 2L+3} W'_{2L+4} P_{2L+5, 4L-2t-2} W_{4L-2t-1}
P_{4L-2t, 4L} \nonumber\\
&\quad + \dots \nonumber \\
&\quad +P_{2L+1, 4L-2t-3} W'_{4L-2t-2} W_{4L-2t-1} P_{4L-2t, 4L} \nonumber\\
&\quad +P_{2L+1, 4L-2t-2} W_{4L-2t-1} W'_{4L-2t} P_{4L-2t+1, 4L} \nonumber \\
&\quad + \dots \nonumber \\
&\quad +P_{2L+1, 4L-2t-2} W_{4L-2t-1}P_{4L-2t, 4L-3} W'_{4L-2} P_{4L-1, 4L} \nonumber\\
&\quad +P_{2L+1, 4L-2t-2} W_{4L-2t-1}P_{4L-2t, 4L-1} W'_{4L}.
\label{eq:derivative_qthetax_af2}
\end{align}
Then it follows from Eqs.~(\ref{eq:qdthetax_af}) and (\ref{eq:derivative_qthetax_af2}) that
\begin{align}
Q^\dagger(\theta;\vec x) P(\theta) Q'(\theta;\vec x) &=
A_t^{(1)} W_{2t+1} B_t^{(1)} W_{4L-2t-1} C_t^{(1)}
+
A_t^{(2)} W_{2t+1} B_t^{(2)} W_{4L-2t-1} C_t^{(2)}, \\
&=
A_t^{(1)} V(-x_j) B_t^{(1)} V(x_j) C_t^{(1)}
+
A_t^{(2)} V(-x_j) B_t^{(2)} V(x_j) C_t^{(2)},
\end{align}
where
\begin{align}
A_t^{(1)} &= P_{0, 2t}, \label{eq:at11}\\
B_t^{(1)} &= P_{2t+2, 4L-2t-2}, \label{eq:bt11}\\
C_t^{(1)} &= \sum_{k=0}^t P_{4L-2t, 4L-2k-1} W'_{4L-2k} P_{4L-2k+1, 4L} \\
&= \sum_{k=0}^t P_{4L-2t, 4L-2k-1} U'(\theta; x_{2k+1}) P_{4L-2k+1, 4L}, \label{eq:ct11} \end{align}
\begin{align}
A_t^{(2)} &= P_{0, 2t}, \label{eq:at21} \\
B_t^{(2)} &= \sum_{k=t+1}^{L-1} P_{2t+2, 4L-2k-1} W'_{4L-2k} P_{4L-2k+1, 4L-2t-2} \\
&= \sum_{k=t+1}^{L-1} P_{2t+2, 4L-2k-1} U'(\theta; x_{2k+1}) P_{4L-2k+1, 4L-2t-2}, \label{eq:bt21} \\
C_t^{(2)} &= P_{4L-2t, 4L}.\label{eq:ct21}
\end{align}
Meanwhile, we have
\begin{align}
Q^\dagger(\theta;\vec x) P'(\theta) Q(\theta;\vec x)
= A_t^{(3)} W_{2t+1} B_t^{(3)} W_{4L-2t-1} C_t^{(3)} 
= A_t^{(3)} V(-x_j) B_t^{(3)} V(x_j) C_t^{(3)},
\end{align}
where
\begin{align}
A_t^{(3)} &= P_{0, 2t}, \label{eq:at31}\\
B_t^{(3)} &= P_{2t+2, 2L-1} P'(\theta) P_{2L+1, 4L-2t-2}, \label{eq:bt31}\\
C_t^{(3)} &= P_{4L-2t, 4L}. \label{eq:ct31}
\end{align}
Combining the above facts with Eqs.~(\ref{eq:avbvc}) and (\ref{eq:derivative_delta_af}) yields
\begin{align}
\Delta'(\theta;\vec x)
= C'_j(\theta;\vec x_{\neg j}) \cosp{2x_j} + S'_j(\theta;\vec x_{\neg j}) \sinp{2x_j} + B'_j(\theta;\vec x_{\neg j}),
\end{align}
where 
\begin{align}
C'_j(\theta;\vec x_{\neg j}) &= \rep{\bra{\bar{0}} A_t^{(1)}\lrb{B_t^{(1)}-\bar{Z}B_t^{(1)}\bar{Z}}C_t^{(1)}
\ket{\bar{0}}}  +\rep{\bra{\bar{0}} A_t^{(2)}\lrb{B_t^{(2)}-\bar{Z} B_t^{(2)} \bar{Z}}C_t^{(2)}\ket{\bar{0}}} \nonumber  \\
&\quad +\dfrac{1}{2}\bra{\bar{0}} A_t^{(3)}\lrb{B_t^{(3)}-\bar{Z} B_t^{(3)} \bar{Z}}C_t^{(3)} \ket{\bar{0}}, 
\label{eq:cjp_af1}\\
S'_j(\theta;\vec x_{\neg j}) &=
\imp{ \bra{\bar{0}} A_t^{(1)}\lrb{B_t^{(1)} \bar{Z}-\bar{Z} B_t^{(1)}} C_t^{(1)}
\ket{\bar{0}}}  +\imp{ \bra{\bar{0}} A_t^{(2)}\lrb{B_t^{(2)} \bar{Z} - \bar{Z} B_t^{(2)}} C_t^{(2)}\ket{\bar{0}}} \nonumber \\
&\quad -\dfrac{i}{2} \bra{\bar{0}}\lrbb{A_t^{(3)} \lrb{B_t^{(3)} \bar{Z} - \bar{Z} B_t^{(3)}} C_t^{(3)}} \ket{\bar{0}}
\label{eq:sjp_af1}\\
B'_j(\theta;\vec x_{\neg j}) &=
\rep{\bra{\bar{0}} A_t^{(1)}\lrb{B_t^{(1)}+\bar{Z}B_t^{(1)}\bar{Z}}C_t^{(1)}
\ket{\bar{0}}}  +\rep{\bra{\bar{0}} A_t^{(2)}\lrb{B_t^{(2)}+\bar{Z} B_t^{(2)} \bar{Z}}C_t^{(2)}\ket{\bar{0}}} \nonumber  \\
&\quad +\dfrac{1}{2}\bra{\bar{0}} A_t^{(3)}\lrb{B_t^{(3)}+\bar{Z}B_t^{(3)}\bar{Z}}C_t^{(3)} \ket{\bar{0}}.
\label{eq:bjp_af1}
\end{align}
Given $\theta$ and $\vec x_{\neg j}$, we first compute the following matrices in a total of $O(L)$ time by standard dynamic programming technique:
\begin{itemize}
    \item $P_{0, 2t}$, $P_{2t+2, 4L-2t-2}$, $P_{4L-2t, 4L}$,    $P_{2t+2, 2L-1}$, $P_{2L+1, 4L-2t-2}$; 
    \item $P_{4L-2t, 4L-2k-1}$ and $P_{4L-2k+1, 4L}$ for $k=0, 1, \dots, t$;
    \item $P_{2t+2, 4L-2k-1}$ and $P_{4L-2k+1, 4L-2t-2}$ for $k=t+1, t+2, \dots, L-1$. 
\end{itemize}
Then we compute $A_t^{(i)}$, $B_t^{(i)}$ and $C_t^{(i)}$ for $i=1,2,3$ by Eqs.~(\ref{eq:at11}-\ref{eq:ct11}), (\ref{eq:at21}-\ref{eq:ct21}) and (\ref{eq:at31}-\ref{eq:ct31}). After that, we calculate $C'_j(\theta;\vec x_{\neg j})$, $S'_j(\theta;\vec x_{\neg j})$ and $B'_j(\theta;\vec x_{\neg j})$ by Eqs.~(\ref{eq:cjp_af1}-\ref{eq:bjp_af1}). Overall, this procedure takes $O(L)$ time.

\item Case 2: $j=2t+1$ is odd, where $0 \le t \le L-1$. In this case,  $W_{2t}=U(\theta;-x_j)$, and $W_{4L-2t}=U(\theta; x_j)$. They by Eqs.~(\ref{eq:delta_af}), (\ref{eq:qdpqthetax_af}) and (\ref{eq:aubuc}), we get
\begin{align}
\Delta(\theta; \vec x) &= \bra{\bar{0}} P_{0, 2t-1} U(\theta; -x_j) P_{2t+1, 4L-2t-1} U(\theta; x_j) P_{4L-2t+1, 4L}  \ket{\bar{0}} \\
&=  
C_j(\theta;\vec x_{\neg j}) \cosp{2x_j}
+
S_j(\theta;\vec x_{\neg j}) \sinp{2x_j} 
+
B_j(\theta;\vec x_{\neg j}), 
\end{align}
where
\begin{align}
C_j(\theta;\vec x_{\neg j}) &=
\dfrac{1}{2}  \bra{\bar{0}} P_{0, 2t-1}\lrb{ P_{2t+1, 4L-2t-2} - P(\theta)P_{2t+1, 4L-2t-1}P(\theta)}P_{4L-2t+1, 4L} \ket{\bar{0}},  \label{eq:cj_af2}\\
S_j(\theta;\vec x_{\neg j}) &=
- \dfrac{i}{2} 
\bra{\bar{0}} P_{0, 2t-1} \lrb{P_{2t+1, 4L-2t-1}P(\theta) - P(\theta)P_{2t+1, 4L-2t-1}}P_{4L-2t+1, 4L}  \ket{\bar{0}}, \label{eq:sj_af2} \\
B_j(\theta;\vec x_{\neg j}) &= \dfrac{1}{2} \bra{\bar{0}} P_{0, 2t-1}
\lrb{P_{2t+1, 4L-2t-1} + P(\theta)P_{2t+1, 4L-2t-1}P(\theta)}P_{4L-2t+1, 4L}  \ket{\bar{0}}.
\label{eq:bj_af2}
\end{align}
Given $\theta$ and $\vec x_{\neg j}$, we first compute $P_{0, 2t-1}$, $P_{2t+1, 4L-2t-1}$ and $P_{4L-2t+1, 4L}$ in $O(L)$ time. Then we calculate $C_j(\theta;\vec x_{\neg j})$, $S_j(\theta;\vec x_{\neg j})$ and $B_j(\theta;\vec x_{\neg j})$ by Eqs.~(\ref{eq:cj_af2}-\ref{eq:bj_af2}). This procedure takes only $O(L)$ time.

Next, we describe how to compute $C'_j(\theta;\vec x_{\neg j})$, $S'_j(\theta;\vec x_{\neg j})$ and $B'_j(\theta;\vec x_{\neg j})$. Using Eq.~(\ref{eq:derivative_qthetax_af}) and the fact $P_{a,b}=P_{a, 4L-2t-1} W_{4L-2t} P_{4L-2t+1, b}$ for any $a \le 4L-2t \le b$, we obtain
\begin{align}
Q'(\theta;\vec x) &=
P_{2L+1, 2L+1} W'_{2L+2} P_{2L+3, 4L-2t-1} W_{4L-2t} P_{4L-2t+1, 4L} \nonumber \\
&\quad +P_{2L+1, 2L+3} W'_{2L+4} P_{2L+5, 4L-2t-1} W_{4L-2t}
P_{4L-2t+1, 4L} \nonumber\\
&\quad + \dots \nonumber \\
&\quad +P_{2L+1, 4L-2t-3} W'_{4L-2t-2} P_{4L-2t-1, 4L-2t-1} W_{4L-2t} P_{4L-2t+1, 4L} \nonumber\\
&\quad +P_{2L+1, 4L-2t-1} W'_{4L-2t} P_{4L-2t+1, 4L} \nonumber \\
&\quad +P_{2L+1, 4L-2t-1} W_{4L-2t} P_{4L-2t+1, 4L-2t+1} W'_{4L-2t+2} P_{4L-2t+3, 4L} \nonumber \\
&\quad + \dots \nonumber \\
&\quad +P_{2L+1, 4L-2t-1} W_{4L-2t}P_{4L-2t+1, 4L-3} W'_{4L-2} P_{4L-1, 4L} \nonumber\\
&\quad +P_{2L+1, 4L-2t-1} W_{4L-2t}P_{4L-2t+1, 4L-1} W'_{4L}.
\label{eq:derivative_qthetax_af3}
\end{align}
Then it follows from Eqs.~(\ref{eq:qdthetax_af}) and (\ref{eq:derivative_qthetax_af2}) that
\begin{align}
Q^\dagger(\theta;\vec x) P(\theta) Q'(\theta;\vec x) &=
A_t^{(1)} W_{2t} B_t^{(1)} W_{4L-2t} C_t^{(1)} +
A_t^{(2)} W_{2t} B_t^{(2)} W'_{4L-2t} C_t^{(2)} +
A_t^{(3)} W_{2t} B_t^{(3)} W_{4L-2t} C_t^{(3)} \\
&=
A_t^{(1)} U(\theta;-x_j) B_t^{(1)} U(\theta;x_j) C_t^{(1)} +
A_t^{(2)} U(\theta;-x_j) B_t^{(2)} U'(\theta;x_j) C_t^{(2)} \nonumber \\
&\quad+
A_t^{(3)} U(\theta;-x_j) B_t^{(3)} U(\theta;x_j) C_t^{(3)},
\end{align}
where
\begin{align}
A_t^{(1)} &= P_{0, 2t-1}, \label{eq:at12}\\
B_t^{(1)} &= P_{2t+1, 4L-2t-1}, \label{eq:bt12}\\
C_t^{(1)} &= \sum_{k=0}^{t-1} P_{4L-2t+1, 4L-2k-1} W'_{4L-2k} P_{4L-2k+1, 4L} \\
&= \sum_{k=0}^{t-1} P_{4L-2t+1, 4L-2k-1} U'(\theta; x_{2k+1}) P_{4L-2k+1, 4L}, \label{eq:ct12} \\
A_t^{(2)} &= P_{0, 2t-1}, \label{eq:at22} \\
B_t^{(2)} &= P_{2t+1, 4L-2t-1}, \label{eq:bt22} \\
C_t^{(2)} &= P_{4L-2t+1, 4L},\label{eq:ct22} \\
A_t^{(3)} &= P_{0, 2t-1}, \label{eq:at32} \\
B_t^{(3)} &= \sum_{k=t+1}^{L-1} P_{2t+1, 4L-2k-1} W'_{4L-2k} P_{4L-2k+1, 4L-2t-1} \\
&= \sum_{k=t+1}^{L-1} P_{2t+1, 4L-2k-1} U'(\theta; x_{2k+1}) P_{4L-2k+1, 4L-2t-1}, \label{eq:bt32} \\
C_t^{(3)} &= P_{4L-2t+1, 4L}.\label{eq:ct32}
\end{align}
Meanwhile, we have
\begin{align}
Q^\dagger(\theta;\vec x) P'(\theta) Q(\theta;\vec x)
= A_t^{(4)} W_{2t} B_t^{(4)} W_{4L-2t} C_t^{(4)} = A_t^{(4)} U(\theta;-x_j) B_t^{(4)} U(\theta;x_j) C_t^{(4)},
\end{align}
where
\begin{align}
A_t^{(4)} &= P_{0, 2t-1}, \label{eq:at4}\\
B_t^{(4)} &= P_{2t+1, 2L-1} P'(\theta) P_{2L+1, 4L-2t-1}, \label{eq:bt4}\\
C_t^{(4)} &= P_{4L-2t+1, 4L}. \label{eq:ct4}
\end{align}
Combining the above facts with Eqs.~(\ref{eq:aubuc}), (\ref{eq:aubupc}) and (\ref{eq:derivative_delta_af}) yields
\begin{align}
\Delta'(\theta;\vec x)
= C'_j(\theta;\vec x_{\neg j}) \cosp{2x_j} + S'_j(\theta;\vec x_{\neg j}) \sinp{2x_j} + B'_j(\theta;\vec x_{\neg j}),
\end{align}
where 
\begin{align}
C'_j(\theta;\vec x_{\neg j}) &= \rep{\bra{\bar{0}} A_t^{(1)}\lrb{B_t^{(1)}-P(\theta)B_t^{(1)}P(\theta)}C_t^{(1)}
\ket{\bar{0}}} \nonumber\\
&\quad -\rep{\bra{\bar{0}} A_t^{(2)} P(\theta) B_t^{(2)} P'(\theta) C_t^{(2)}\ket{\bar{0}}} \nonumber \\
&\quad + \rep{\bra{\bar{0}} A_t^{(3)}\lrb{B_t^{(3)}-P(\theta) B_t^{(3)} P(\theta)}C_t^{(3)}\ket{\bar{0}}} \nonumber \\
&\quad +\dfrac{1}{2}\bra{\bar{0}} A_t^{(4)}\lrb{B_t^{(4)}-P(\theta)B_t^{(4)}P(\theta)}C_t^{(4)} \ket{\bar{0}}, 
\label{eq:cjp_af2}\\
S'_j(\theta;\vec x_{\neg j}) &=
\imp{ \bra{\bar{0}} A_t^{(1)}\lrb{B_t^{(1)} P(\theta)-P(\theta) B_t^{(1)}} C_t^{(1)}
\ket{\bar{0}}} \nonumber \\
&\quad +\imp{\bra{\bar{0}} A_t^{(2)}B_t^{(2)}P'(\theta) C_t^{(2)}\ket{\bar{0}}} \nonumber \\
&\quad +\imp{ \bra{\bar{0}} A_t^{(3)}\lrb{B_t^{(3)} P(\theta) - P(\theta) B_t^{(3)}} C_t^{(3)}\ket{\bar{0}}} \nonumber \\
&\quad -\dfrac{i}{2} \bra{\bar{0}}\lrbb{A_t^{(4)} \lrb{B_t^{(4)} P(\theta) - P(\theta) B_t^{(4)}} C_t^{(4)}} \ket{\bar{0}},
\label{eq:sjp_af2}\\
B'_j(\theta;\vec x_{\neg j}) &=
\rep{\bra{\bar{0}} A_t^{(1)}\lrb{B_t^{(1)}+P(\theta)B_t^{(1)}P(\theta)}C_t^{(1)}
\ket{\bar{0}}} \nonumber\\
&\quad +\rep{\bra{\bar{0}} A_t^{(2)}P(\theta)B_t^{(2)}P'(\theta)C_t^{(2)}\ket{\bar{0}}} \nonumber \\
&\quad +\rep{\bra{\bar{0}} A_t^{(3)}\lrb{B_t^{(3)}+P(\theta) B_t^{(3)} P(\theta)}C_t^{(3)}\ket{\bar{0}}} \nonumber \\
&\quad +\dfrac{1}{2}\bra{\bar{0}} A_t^{(4)}\lrb{B_t^{(4)}+P(\theta)B_t^{(4)}P(\theta)}C_t^{(4)} \ket{\bar{0}}.
\label{eq:bjp_af2}
\end{align}
Given $\theta$ and $\vec x_{\neg j}$, we first compute the following matrices in a total of $O(L)$ time by standard dynamic programming technique:
\begin{itemize}
    \item $P_{0, 2t-1}$, $P_{2t+1, 4L-2t-1}$, $P_{4L-2t+1, 4L}$,    $P_{2t+1, 2L-1}$, $P_{2L+1, 4L-2t-1}$; 
    \item $P_{4L-2t+1, 4L-2k-1}$ and $P_{4L-2k+1, 4L}$ for $k=0, 1, \dots, t-1$;
    \item $P_{2t+1, 4L-2k-1}$ and $P_{4L-2k+1, 4L-2t-1}$ for $k=t+1, t+2, \dots, L-1$. 
\end{itemize}
Then we compute $A_t^{(i)}$, $B_t^{(i)}$ and $C_t^{(i)}$ for $i=1,2,3,4$ by Eqs.~(\ref{eq:at12}-\ref{eq:ct12}), (\ref{eq:at22}-\ref{eq:ct22}),  (\ref{eq:at32}-\ref{eq:ct32}) and (\ref{eq:at4}-\ref{eq:ct4}). After that, we calculate $C'_j(\theta;\vec x_{\neg j})$, $S'_j(\theta;\vec x_{\neg j})$ and $B'_j(\theta;\vec x_{\neg j})$ by Eqs.~(\ref{eq:cjp_af2}-\ref{eq:bjp_af2}). Overall, this procedure takes $O(L)$ time.
\end{itemize}
\end{widetext}

\section{Algorithms for maximizing the slope of the ancilla-free likelihood function}
\label{sec:alg_opt_slope_af}
In this appendix, we present two algorithms for maximizing the slope of the ancilla-free  likelihood function $\mathbb{P}(d|\theta; f, \vec x)$ at a given point $\theta=\mu$ (i.e. the prior mean of $\theta$). Namely, our goal is to find $\vec x \in \mathbb{R}^{2L}$ that maximizes $|\mathbb{P}'(\mu; f, \vec x)| = f |\Delta'(\mu; \vec x)|/2$. 

Similar to Algorithms \ref{alg:ga_opt_v0_af} and \ref{alg:ca_opt_v0_af} for Fisher information maximization, our algorithms for slope maximization are also based on gradient ascent and coordinate ascent, respectively. They both need to call the procedures in Lemma \ref{lem:eval_coeff_func_af} to evaluate $C'(\mu; \vec x_{\neg j})$, $S'(\mu; \vec x_{\neg j})$ and $B'(\mu; \vec x_{\neg j})$ for given $\mu$ and $\vec x_{\neg j}$. However, the gradient-ascent-based algorithm uses the above quantities to compute the partial derivative of $(\Delta'(\mu; \vec x))^2$ with respect to $x_j$, while the coordinate-ascent-based algorithm uses them to directly update the value of $x_j$. These algorithms are formally described in Algorithms \ref{alg:ga_opt_slope_af} and \ref{alg:ca_opt_slope_af}, respectively. 

\begin{algorithm*}[ht]
 \KwIn{The prior mean $\mu$ of $\theta$, the number $L$ of circuit layers, the step size schedule $\delta: \mathbb{Z}^{\ge 0} \to \mathbb{R}^+$, the error tolerance $\epsilon$ for termination.}
 \KwOut{A set of parameters $\vec x = (x_1, x_2, \dots, x_{2L}) \in \mathbb{R}^{2L}$ that are a local maximum point of the function $|\Delta'(\mu; \vec x)|$.}

Choose random initial point $\vec x^{(0)}=(x^{(0)}_1, x^{(0)}_2, \dots, x^{(0)}_{2L}) \in (-\pi, \pi]^{2L}$; \\
$t \leftarrow 0$; \\
\While{True}{

\For{$j \leftarrow 1$ \KwTo $2L$}{
    Let $\vec x^{(t)}_{\neg j} = (x^{(t)}_1, \dots, x^{(t)}_{j-1}, x^{(t)}_{j+1}, \dots, x^{(t)}_{2L})$; \\
    Compute $C'^{(t)}_j:= C'_j(\mu; \vec x^{(t)}_{\neg j})$, $S'^{(t)}_j:= S'_j(\mu; \vec x^{(t)}_{\neg j})$ and $B'^{(t)}_j:= B'_j(\mu; \vec x^{(t)}_{\neg j})$
    by using the procedures in Lemma \ref{lem:eval_coeff_func_af}; \\
    Compute $\Delta'(\mu; \vec x)$ at $\vec x=\vec x^{(t)}$ as follows:
    \begin{align}
    \Delta'^{(t)} := \Delta'(\mu; \vec x^{(t)}) = C'^{(t)}_j \cosp{2 x^{(t)}_j} + S'^{(t)}_j \sinp{2 x^{(t)}_j} + B'^{(t)}_j;
    \end{align}
    Compute the partial derivative of $\Delta'(\mu; \vec x)$ with respect to $x_j$ as follows: 
    \begin{align}
    \gamma^{(t)}_j := \dfrac{\partial \Delta'(\mu; \vec x)}{\partial x_j} |_{\vec x = \vec x^{(t)}} = 2 \lrb{-C'^{(t)}_j \sinp{2x^{(t)}_j}+S'^{(t)}_j \cosp{2 x^{(t)}_j}};    
    \end{align}
} 

Set $\vec x^{(t+1)}=\vec x^{(t)} + \delta(t) \nabla^{(t)}$, where
$\nabla^{(t)}:=  (2\Delta'^{(t)}\gamma^{(t)}_1, 2\Delta'^{(t)}\gamma^{(t)}_2, \dots, 2\Delta'^{(t)}\gamma^{(t)}_{2L})$ is the gradient of $(\Delta'(\mu; \vec x))^2$ at $\vec x=\vec x^{(t)}$;\\
\If{$|\Delta'(\mu; \vec x^{(t+1)}) - \Delta'(\mu; \vec x^{(t)})| < \epsilon$}{break;}
$t \leftarrow t+1$;
}
Return $\vec x^{(t+1)}=(x_1^{(t+1)}, x_2^{(t+1)}, \dots, x_{2L}^{(t+1)})$ as the optimal parameters.
\caption{Gradient ascent for slope maximization in the ancilla-free case}
\label{alg:ga_opt_slope_af}
\end{algorithm*}

\begin{algorithm*}[ht]
 \KwIn{The prior mean $\mu$ of $\theta$, the number $L$ of circuit layers, the error tolerance $\epsilon$ for termination.}
 \KwOut{A set of parameters $\vec x = (x_1, x_2, \dots, x_{2L}) \in (-\pi, \pi]^{2L}$ that are a local maximum point of the function $|\Delta'(\mu; \vec x)|$.}

Choose random initial point $\vec x^{(0)}=(x^{(0)}_1, x^{(0)}_2, \dots, x^{(0)}_{2L}) \in (-\pi, \pi]^{2L}$; \\
$t \leftarrow 1$; \\
\While{True}{

\For{$j \leftarrow 1$ \KwTo $2L$}{
    Let $\vec x^{(t)}_{\neg j} = (x^{(t)}_1, \dots, x^{(t)}_{j-1}, x^{(t-1)}_{j+1}, \dots, x^{(t-1)}_{2L})$; \\
    Compute $C'^{(t)}_j:= C'_j(\mu; \vec x^{(t)}_{\neg j})$, $S'^{(t)}_j:= S'_j(\mu; \vec x^{(t)}_{\neg j})$, $B'^{(t)}_j:= B'_j(\mu; \vec x^{(t)}_{\neg j})$
    by using the procedures in Lemma \ref{lem:eval_coeff_func_af}; \\
    Set $x^{(t)}_j = \arg{\sgn{B'^{(t)}_j} (C'^{(t)}_j + i S'^{(t)}_j)}/2$, where $\sgn{z}=1$ if $z\ge 0$ and $-1$ otherwise;
} 
\If{$|\Delta'(\mu; \vec x^{(t)}) - \Delta'(\mu; \vec x^{(t-1)})| < \epsilon$}{break;}
$t \leftarrow t+1$;
}
Return $\vec x^{(t)}=(x_1^{(t)}, x_2^{(t)}, \dots, x_{2L}^{(t)})$ as the optimal parameters.
\caption{Coordinate ascent for slope maximization in the ancilla-free case}
\label{alg:ca_opt_slope_af}
\end{algorithm*}

\section{Ancilla-based scheme}
\label{sec:ancilla_based}

In this appendix, we present an alternative scheme, called the \textit{ancilla-based} scheme. In this scheme, the engineered likelihood function (ELF) is generated by the quantum circuit in Figure~\ref{fig:elf_circuit_ab}, where $U(\theta; x_{2j-1})$, $V(x_{2j})$ and $Q(\theta; \vec x)$ are the same as in the ancilla-free scheme (i.e. they satisfy Eqs.~\eqref{eq:u_theta_x},~\eqref{eq:v_x} and ~\eqref{eq:q_theta_vecx}), in which $\vec x=(x_1, x_2, \dots, x_{2L-1}, x_{2L}) \in\mathbb R^{2L}$ are tunable parameters.

\begin{figure*}[!ht]
\begin{align*}
\LARGE
\Qcircuit @C=1.2em @R=0.36em @!R {
\lstick{\ket {0}} & \gate H & \ctrl 1 & \gate H & \meter & \rstick{d\in\{0,1\}} \cw \\
\lstick{\ket A} & \qw & \gate{Q(\theta;\vec x)} & \qw & \qw & 
}
\end{align*}
\begin{align*}
\Qcircuit @C=1.2em @R=0.36em @!R { 
& \gate{Q(\theta;\vec x)} & \qw 
& = & & \gate{U(\theta;x_1)} & \gate{V(x_2)} & \gate{U(\theta;x_3)} & \gate{V(x_4)} &\qw & \cdots & & \gate{U(\theta;x_{2L-1})} & \gate{V(x_{2L})} & \qw 
}
\end{align*}
\caption{Quantum circuit for the ancilla-based engineered likelihood functions}
\label{fig:elf_circuit_ab}
\end{figure*}
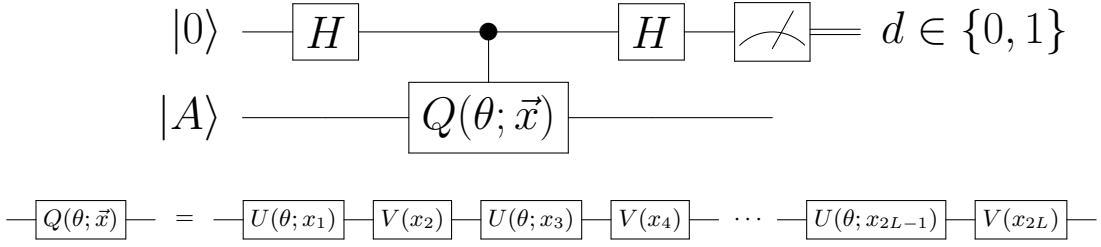

Assuming the circuit in Figure \ref{fig:elf_circuit_ab} is noiseless, the ancilla-based engineered likelihood function (AB ELF) is given by
\begin{align}
\mathbb{P}(d|\theta; \vec x)=
\frac 12 \left[ 
1+(-1)^d \Lambda(\theta; \vec x)\right],
~~\forall d \in \{0,1\}, 
\label{eq:noiseless_elf_ab}
\end{align}
where
\begin{align}
\Lambda(\theta; \vec x) = \mathrm{Re}(\bra{A} Q(\theta; \vec x) \ket{A})    
\label{eq:delta_ab}
\end{align}
is the bias of the likelihood function. In particular, if $\vec x=(\pi/2, \pi/2, \dots, \pi/2)$, we get $\Lambda(\theta; \vec x) = (-1)^L \cosp{L \theta}$ and 
call the corresponding likelihood function the ancilla-based Chebyshev likelihood function (AB CLF).

It turns out that most of the argument in Section \ref{subsec:qc_elf} still holds in the ancilla-based case, except that we need to replace  $\Delta(\theta; \vec x)$ with $\Lambda(\theta; \vec x)$. So we will use the same notation (e.g. $\ket{\bar{0}}$, $\ket{\bar{1}}$, $\bar{X}$, $\bar{Y}$, $\bar{Z}$, $\bar{I}$) as before, unless otherwise stated. In particular, when we take the errors in the circuit in Figure \ref{fig:elf_circuit_ab} into account, the noisy likelihood function is given by
\begin{align}
\mathbb{P}(d|\theta; f, \vec x)=
\frac 12 \left[ 
1+(-1)^d f \Lambda(\theta; \vec x)\right],
~~\forall d \in \{0,1\},
\label{eq:noisy_elf_ab}
\end{align}
where $f$ is the fidelity of the process for generating the ELF. Note that, however, there does exist a difference between $\Delta(\theta; \vec x)$ and $\Lambda(\theta; \vec x)$, as the former is trigono-multiqudaratic in $\vec x$, while the latter is trigono-multilinear in $\vec x$.

We will tune the circuit parameters $\vec x$ and perform Bayesian inference with the resultant ELFs in the same way as in Section \ref{subsec:bi_elf}. In fact, the argument in Section \ref{subsec:bi_elf} still holds in the ancilla-based case, except that
we need to replace $\Delta(\theta; \vec x)$ with $\Lambda(\theta; \vec x)$. So we will use the same notation as before, unless otherwise stated. In particular, we also define the variance reduction factor $\mathcal{V}(\mu, \sigma; f, \vec x)$ as in Eqs.~(\ref{eq:b}) and (\ref{eq:varReductionFactor}), replacing $\Delta(\theta; \vec x)$ with $\Lambda(\theta; \vec x)$. As shown in Appendix \ref{sec:proxy_comparison}, 
\begin{align}
\mathcal{V}(\mu, \sigma; f, \vec x) &\approx \mathcal{I}(\mu; f, \vec x) = \dfrac{f^2\lrb{\Lambda'(\theta; \vec x)}^2}{ 1 - f^2 \lrb{\Lambda(\theta; \vec x)}^2}, & \nonumber\\ & \qquad \mbox{when $\sigma$ is small}, 
\label{eq:v_fi_ab}
\end{align}
and
\begin{align}
\mathcal{V}(\mu, \sigma; f, \vec x) &\approx f^2 (\Lambda'(\mu; \vec x))^2,
\nonumber\\
& \qquad \mbox{when both $\sigma$ and $f$ are small}.  
\label{eq:v_slope_ab}
\end{align}
Namely, the Fisher information and slope of the likelihood function $\mathbb{P}(d|\theta; f, \vec x)$ at $\theta=\mu$ are two proxies of the variance reduction factor $\mathcal{V}(\mu, \sigma; f, \vec x)$ under reasonable assumptions. Since the direct optimization of $\mathcal{V}$ is hard in general, we will tune the parameters $\vec x$ by optimizing these proxies instead.

\subsection{Efficient maximization of proxies of the variance reduction factor}
Now we present efficient heuristic algorithms for maximizing two proxies of the variance reduction factor $\mathcal{V}$ -- the Fisher information and slope of the likelihood function $\mathbb{P}(d|\theta; f, \vec x)$. All of these algorithms make use of the following procedures for evaluating the CSD coefficient functions of the bias $\Lambda(\theta; \vec x)$ and its derivative $\Lambda'(\theta; \vec x)$ with respect to $x_j$ for $j=1,2,\dots, 2L$. 

\subsubsection{Evaluating the CSD coefficient functions of the bias and its derivative}
\label{sec:ecfbd_ab}

Since $\Lambda(\theta; \vec x)$ is trigono-multilinear in $\vec x$, for any $j\in \{1,2,...,2L\}$, there exist functions $C_j(\theta; \vec x_{\neg j})$ and $S_j(\theta; \vec x_{\neg j})$, that are trigono-multilinear in $\vec x_{\neg j}=(x_1, \dots, x_{j-1}, x_{j+1}, \dots, x_{2L})$, such that
\begin{align}
\Lambda(\theta; \vec x) = C_j(\theta; \vec x_{\neg j}) \cosp{x_j} + S_j(\theta; \vec x_{\neg j}) \sinp{x_j}.
\end{align}
It follows that
\begin{align}
\Lambda'(\theta; \vec x) = C_j'(\theta; \vec x_{\neg j}) \cosp{x_j} + S_j'(\theta; \vec x_{\neg j}) \sinp{x_j}
\end{align}
is also trigono-multilinear in $\vec x$, where
$C'_j(\theta;\vec x_{\neg j}) = \partial_{\theta} C_j(\theta; \vec x_{\neg j})$ and
$S'_j(\theta;\vec x_{\neg j}) = \partial_{\theta} S_j(\theta; \vec x_{\neg j})$
are the partial derivatives of $C_j(\theta; \vec x_{\neg j})$ and $S_j(\theta; \vec x_{\neg j})$ with respect to $\theta$, respectively. 

Our optimization algorithms require efficient procedures for evaluating 
$C_j(\theta; \vec x_{\neg j})$, $S_j(\theta; \vec x_{\neg j})$, $C_j'(\theta; \vec x_{\neg j})$ and $S_j'(\theta; \vec x_{\neg j})$ for given $\theta$ and $\vec x_{\neg j}$. It turns out that these tasks can be accomplished in $O(L)$ time.

\begin{lemma}
Given $\theta$ and $\vec x_{\neg j}$, each of $C_j(\theta; \vec x_{\neg j})$, $S_j(\theta; \vec x_{\neg j})$, $C_j'(\theta; \vec x_{\neg j})$ and $S_j'(\theta; \vec x_{\neg j})$ can be computed in $O(L)$ time.
\label{lem:eval_coeff_func_ab}
\end{lemma}
\begin{proof}
For convenience, we introduce the following notation. Let $W_{2i}=V(x_{2L-2i})$, $W_{2i+1}=U(\theta;x_{2L-2i-1})$, for $i=0, 1, \dots, L-1$. Furthermore, let $W'_j=\partial_{\theta} W_j$ for $j=0, 1, \dots, 2L-1$. Note that $W'_j=0$ if $j$ is even.
Then we define $P_{a,b}=W_{a}W_{a+1}\dots W_b$ if $0 \le a \le  b \le 2L-1$, and $P_{a,b}=I$ otherwise. 

With this notation, Eqs.~(\ref{eq:q_theta_vecx}) and (\ref{eq:derivative_qalphabeta}) imply that
\begin{align}
Q(\theta;\vec x) = P_{0, a-1} W_a P_{a+1, 2L-1}, & \forall 0 \le a \le 2L-1, 
\end{align}
and
\begin{align}
Q'(\theta; \vec x) &= P_{0, 0} W'_1 P_{2, 2L-1}+
P_{0, 2} W'_3 P_{4, 2L-1} + \dots 
\nonumber\\
&\quad +
P_{0, 2L-4} W'_{2L-3} P_{2L-2, 2L-1}  
+ P_{0, 2L-2} W'_{2L-1}.
\label{eq:derivative_qalphabeta2}
\end{align}

In order to evaluate $C_j(\theta; \vec x_{\neg j})$, $S_j(\theta; \vec x_{\neg j})$, $C_j'(\theta; \vec x_{\neg j})$ and $S_j'(\theta; \vec x_{\neg j})$ for given $\theta$ and $\vec x_{\neg j}$, we consider the case $j$ is even and the case $j$ is odd separately.

\begin{itemize}

\item Case 1: $j=2(L-t)$ is even, where $0 \le t \le L-1$. In this case,  $W_{2t}=V(x_j)$. Using the fact 
\begin{align}
Q(\theta;\vec x) &= P_{0, 2t-1} W_{2t} P_{2t+1, 2L-1} \\
&= P_{0, 2t-1} (\cosp{x_j} \bar{I} - \i \sinp{x_j} \bar{Z}) P_{2t+1, 2L-1} \\
&= \cosp{x_j} P_{0, 2t-1} P_{2t+1, 2L-1}  \nonumber\\
&\quad -\i \sinp{x_j} P_{0, 2t-1}  
\bar{Z} P_{2t+1, 2L-1},
\end{align}
we obtain 
\begin{align}
\Lambda(\theta; \vec x) &=
C_j(\theta;\vec x_{\neg j}) \cosp{x_j} + 
S_j(\theta;\vec x_{\neg j})\sinp{x_j},
\end{align}
where
\begin{align}
C_j(\theta;\vec x_{\neg j}) &= \rep{\bra{\bar{0}}P_{0, 2t-1} P_{2t+1, 2L-1} \ket{\bar{0}}}, 
\label{eq:cj_ab1}\\
S_j(\theta;\vec x_{\neg j}) &= \imp{\bra{\bar{0}} P_{0, 2t-1}  \bar{Z} P_{2t+1, 2L-1} \ket{\bar{0}}}.
\label{eq:sj_ab1}
\end{align}
Given $\theta$ and $\vec x_{\neg j}$, we first compute $P_{0, 2t-1}$ and $P_{2t+1, 2L-1}$ in $O(L)$ time. Then we calculate $C_j(\theta;\vec x_{\neg j})$ and $S_j(\theta;\vec x_{\neg j})$ by Eqs.~(\ref{eq:cj_ab1}) and (\ref{eq:sj_ab1}). This procedure takes only $O(L)$ time. 

Next, we describe how to compute $C'_j(\theta;\vec x_{\neg j})$ and $S'_j(\theta;\vec x_{\neg j})$. Using Eq.~(\ref{eq:derivative_qalphabeta2}) and the fact  $P_{a,b}=P_{a,2t-1}W_{2t}P_{2t+1,b}$, for any $a \le 2t \le b$,
we obtain
\begin{align}
Q'(\theta; \vec x) &= P_{0, 0} W'_1 P_{2, 2t-1} W_{2t} P_{2t+1, 2L-1} \nonumber\\
&\quad + P_{0, 2} W'_3 P_{4, 2t-1} W_{2t} P_{2t+1, 2L-1} \nonumber\\
&\quad + \dots \nonumber \\
&\quad + P_{0, 2t-2} W'_{2t-1} W_{2t} P_{2t+1, 2L-1} \nonumber \\
&\quad + P_{0, 2t-1} W_{2t} W'_{2t+1} P_{2t+2, 2L-1} \nonumber \\
&\quad + \dots \nonumber\\
&\quad + P_{0, 2t-1} W_{2t} P_{2t+1, 2L-4} W'_{2L-3} P_{2L-2, 2L-1}\nonumber \\
&\quad + P_{0, 2t-1} W_{2t} P_{2t+1, 2L-2} W'_{2L-1}.
\label{eq:derivative_qalphabeta3}
\end{align}

Let 
\begin{align}
A_{t} &= \sum_{s=1}^{t} P_{0, 2s-2} W'_{2s-1} P_{2s, 2t-1} \\
&= \sum_{s=1}^{t} P_{0, 2s-2} U'(\theta; x_{2L-2s+1}) P_{2s, 2t-1},
\label{eq:at_ab} \\
B_{t} &= \sum_{s=t+1}^{L} P_{2t+1, 2s-2} W'_{2s-1} P_{2s, 2L-1}\\
&= \sum_{s=t+1}^{L} P_{2t+1, 2s-2} U'(\theta; x_{2L-2s+1}) P_{2s, 2L-1}.
\label{eq:bt_ab}
\end{align}
Then Eq.~(\ref{eq:derivative_qalphabeta3}) yields
\begin{align}
Q'(\theta; \vec x) &=
A_{t} W_{2t} P_{2t+1, 2L-1} + P_{0, 2t-1} W_{2t} B_{t} \\
&=
A_{t} \lrb{\cosp{x_j} \bar{I} - \i \sinp{x_j} \bar{Z}} P_{2t+1, 2L-1}  \nonumber \\ 
&\quad + P_{0, 2t-1}\lrb{\cosp{x_j}) \bar{I} - \i \sinp{x_j} \bar{Z}} B_{t} \\
&=
\cosp{x_j} \lrb{A_{t} P_{2t+1, 2L-1} + P_{0, 2t-1} B_{t}} \nonumber \\
&\quad - \i \sinp{x_j} \lrb{A_{t} \bar{Z} P_{2t+1, 2L-1} + P_{0, 2t-1} \bar{Z} B_{t}},
\end{align}
which leads to
\begin{align}
\Lambda'(\theta;\vec x) &=
C'_j(\theta;\vec x_{\neg j})\cosp{x_j} 
+ S'_j(\theta;\vec x_{\neg j})\sinp{x_j},
\end{align}
where
\begin{align}
C'_j(\theta;\vec x_{\neg j}) &=  \rep{\bra{\bar{0}} (A_{t} P_{2t+1, 2L-1} + P_{0, 2t-1} B_{t}) \ket{\bar{0}}}, 
\label{eq:dcj_ab} \\
S'_j(\theta;\vec x_{\neg j}) &= \imp{\bra{\bar{0}}  (A_{t} \bar{Z} P_{2t+1, 2L-1} + P_{0, 2t-1} \bar{Z} B_{t}) \ket{\bar{0}}}.
\label{eq:dsj_ab}
\end{align}
Given $\theta$ and $\vec x_{\neg j}$, we first compute the following matrices in a total of $O(L)$ time by standard dynamic programming techniques: 
\begin{itemize}
\item $P_{0, 2s-2}$ and $P_{2s, 2t-1}$ for $s=1, 2, \dots, t$;
\item $P_{2t+1, 2s-2}$ and $P_{2s, 2L-1}$ for $s=t+1, t+2, \dots, L$;
\item $P_{0, 2t-1}$ and $P_{2t+1, 2L-1}$.
\end{itemize} 
Then we compute $A_t$ and $B_t$ by Eqs.~(\ref{eq:at_ab}) and (\ref{eq:bt_ab}). After that, we calculate $C'_j(\theta;\vec x_{\neg j})$ and $S'_j(\theta;\vec x_{\neg j})$ by Eqs.~(\ref{eq:dcj_ab}) and (\ref{eq:dsj_ab}). Overall, this procedure takes $O(L)$ time.

\item Case 2: $j=2(L-t)-1$ is odd, where $0 \le t \le L-1$. In this case,  $W_{2t+1}=U(\theta; x_j)$. Using the fact 
\begin{align}
Q(\theta;\vec x) 
&= P_{0, 2t} W_{2t+1} P_{2t+2, 2L-1} \\
&= P_{0, 2t} \lrb{\cosp{x_j} \bar{I} - \i \sinp{x_j} P(\theta)} P_{2t+2, 2L-1} \\
&= \cosp{x_j} P_{0, 2t} P_{2t+2, 2L-1} \nonumber\\
&\quad - \i \sinp{x_j} P_{0, 2t}  P(\theta) P_{2t+2, 2L-1},
\end{align}
we obtain 
\begin{align}
\Lambda(\theta; \vec x) = 
C_j(\theta;\vec x_{\neg j})\cosp{x_j}+ S_j(\theta;\vec x_{\neg j})\sinp{x_j},
\end{align}
where
\begin{align}
C_j(\theta;\vec x_{\neg j}) &= \rep{\bra{\bar{0}}P_{0, 2t} P_{2t+2, 2L-1} \ket{\bar{0}}}, \label{eq:cj_ab2}\\
S_j(\theta;\vec x_{\neg j}) &= \imp{\bra{\bar{0}} P_{0, 2t}  P(\theta) P_{2t+2, 2L-1} \ket{\bar{0}}}.\label{eq:sj_ab2}
\end{align}
Given $\theta$ and $\vec x_{\neg j}$, we first compute $P_{0, 2t}$ and $P_{2t+2, 2L-1}$ in $O(L)$ time. Then we calculate $C_j(\theta;\vec x_{\neg j})$ and $S_j(\theta;\vec x_{\neg j})$ by Eqs.~(\ref{eq:cj_ab2}) and (\ref{eq:sj_ab2}). This procedure takes only
$O(L)$ time. 

Next, we describe how to compute $C'_j(\theta;\vec x_{\neg j})$ and $S'_j(\theta;\vec x_{\neg j})$. Using Eq.~(\ref{eq:derivative_qalphabeta2}) and the fact $P_{a,b}=P_{a, 2t} W_{2t+1} P_{2t+2, b}$ for any $a \le 2t+1 \le b$, we get
\begin{align}
Q'(\theta; \vec x) &= P_{0, 0} W'_1 P_{2, 2t} W_{2t+1} P_{2t+2, 2L-1} \nonumber\\
&\quad + P_{0, 2} W'_3 P_{4, 2t} W_{2t+1} P_{2t+2, 2L-1} \nonumber\\
&\quad + \dots \nonumber \\
&\quad + P_{0, 2t-2} W'_{2t-1} P_{2t, 2t} W_{2t+1} P_{2t+2, 2L-1} \nonumber \\
&\quad + P_{0, 2t} W'_{2t+1} P_{2t+2, 2L-1} \nonumber \\
&\quad + P_{0, 2t} W_{2t+1} P_{2t+2, 2t+2} W'_{2t+3} P_{2t+4, 2L-1} \nonumber \\
&\quad + \dots \nonumber\\
&\quad + P_{0, 2t} W_{2t+1} P_{2t+2, 2L-4} W'_{2L-3} P_{2L-2, 2L-1}\nonumber \\
&\quad + P_{0, 2t} W_{2t+1} P_{2t+2, 2L-2} W'_{2L-1}.
\label{eq:derivative_qalphabeta4}
\end{align}

Let 
\begin{align}
A_{t} &= \sum_{s=1}^t P_{0, 2s-2} W'_{2s-1} P_{2s, 2t} \\
&= \sum_{s=1}^t P_{0, 2s-2} U'(\theta; x_{2L-2s+1}) P_{2s, 2t}, \label{eq:at_ab2}\\
B_{t}&= \sum_{s=t+2}^L P_{2t+2, 2s-2} W'_{2s-1} P_{2s, 2L-1} \\
&=\sum_{s=t+2}^L P_{2t+2, 2s-2} U'(\theta; x_{2L-2s+1}) P_{2s, 2L-1}. \label{eq:bt_ab2}
\end{align} 
Then Eq.~(\ref{eq:derivative_qalphabeta4}) yields
\begin{align}
Q'(\theta; \vec x) &=
A_{t} W_{2t+1} P_{2t+2, 2L-1}
+
P_{0, 2t}W'_{2t+1} P_{2t+2, 2L-1}
\nonumber\\
&\quad +
P_{0, 2t} W_{2t+1} B_{ t} \\
&=
A_{t} (\cosp{x_j} \bar{I} - \i \sinp{x_j} P(\theta)) P_{2t+2, 2L-1} \nonumber \\
&\quad- \i \sinp{x_j} P_{0, 2t} P'(\theta) P_{2t+2, 2L-1} \nonumber \\
&\quad+
P_{0, 2t} (\cosp{x_j} \bar{I} - \i \sinp{x_j} P(\theta)) B_{ t} \\
&=
\cosp{x_j} \lrb{A_t P_{2t+2, 2L-1}+P_{0, 2t}B_t} \nonumber\\
&\quad -\i \sinp{x_j} (A_t P(\theta) P_{2t+2, 2L-1} 
\nonumber\\
&\quad + 
P_{0, 2t} P'(\theta)P_{2t+2, 2L-1}+P_{0, 2t}P(\theta) B_t),
\end{align}
which leads to
\begin{align}
\Lambda'(\theta;\vec x) = 
C'_j(\theta;\vec x_{\neg j}) \cosp{x_j} + S'_j(\theta;\vec x_{\neg j})\sinp{x_j},
\end{align}
where
\begin{align}
C'_j(\theta;\vec x_{\neg j}) &=  \rep{\bra{\bar{0}} \lrb{A_t P_{2t+2, 2L-1}+P_{0, 2t}B_t} \ket{\bar{0}}}, \label{eq:dcj_ab2}\\
S'_j(\theta;\vec x_{\neg j}) &=  \mathrm{Im}(\bra{\bar{0}} (A_t P(\theta) P_{2t+2, 2L-1} \nonumber\\
&\quad + P_{0, 2t} P'(\theta)P_{2t+2, 2L-1}+P_{0, 2t}P(\theta) B_t) \ket{\bar{0}}). 
\label{eq:dsj_ab2}
\end{align}
Given $\theta$ and $\vec x_{\neg j}$, we first compute the following matrices in a total of $O(L)$ time by standard dynamic programming techniques: 
\begin{itemize}
\item $P_{0, 2s-2}$ and $P_{2s, 2t}$ for $s=1, 2, \dots, t$;
\item $P_{2t+2, 2s-2}$ and $P_{2s, 2L-1}$ for $s=t+2, t+3, \dots, L$;
\item $P_{0, 2t}$ and $P_{2t+2, 2L-1}$.
\end{itemize} 
Then we compute $A_t$ and $B_t$ by Eqs.~(\ref{eq:at_ab2}) and (\ref{eq:bt_ab2}). After that, we calculate $C'_j(\theta;\vec x_{\neg j})$ and $S'_j(\theta;\vec x_{\neg j})$ by Eqs.~(\ref{eq:dcj_ab2}) and (\ref{eq:dsj_ab2}). Overall, this procedure takes $O(L)$ time.
\end{itemize}
\end{proof}

\subsubsection{Maximizing the Fisher information of the likelihood function}
\label{sec:opt_v0_ab}
We propose two algorithms for maximizing the Fisher information of the likelihood function $\mathbb{P}(d|\theta; f, \vec x)$ at a given point $\theta=\mu$ (i.e. the prior mean of $\theta$). Namely, our goal is to find $\vec x \in \mathbb{R}^{2L}$ that maximizes
\begin{align}
\mathcal{I}(\mu; f, \vec x) = \dfrac{f^2\lrb{\Lambda'(\mu; \vec x)}^2}{ 1 - f^2 \lrb{\Lambda(\mu; \vec x)}^2}.
\end{align}

These algorithms are similar to Algorithms \ref{alg:ga_opt_v0_af} and \ref{alg:ca_opt_v0_af} for Fisher information maximization in the ancilla-free case, in the sense that they are also based on gradient ascent and coordinate ascent, respectively. The main difference is that now we invoke the procedures in Lemma \ref{lem:eval_coeff_func_ab} to evaluate $C(\mu; \vec x_{\neg j})$, $S(\mu; \vec x_{\neg j})$, $C'(\mu; \vec x_{\neg j})$ and $S'(\mu; \vec x_{\neg j})$ for given $\mu$ and $\vec x_{\neg j}$, and then use them to either compute the partial derivative of $\mathcal{I}(\mu; f, \vec x)$ with respect to $x_j$ (in gradient ascent) or define a single-variable optimization problem for $x_j$ (in coordinate ascent). These algorithms are formally described in Algorithms \ref{alg:ga_opt_v0_ab} and \ref{alg:ca_opt_v0_ab}.

\begin{algorithm*}[ht]
 \KwIn{The prior mean $\mu$ of $\theta$, the number $L$ of circuit layers, the fidelity $f$ of the process for generating the ELF, the step size schedule $\delta: \mathbb{Z}^{\ge 0} \to \mathbb{R}^+$, the error tolerance $\epsilon$ for termination.}
 \KwOut{A set of parameters $\vec x = (x_1, x_2, \dots, x_{2L}) \in \mathbb{R}^{2L}$ that are a local maximum point of the function $\mathcal{I}(\mu; f, \vec x)$.}

Choose random initial point $\vec x^{(0)}=(x^{(0)}_1, x^{(0)}_2, \dots, x^{(0)}_{2L}) \in (-\pi, \pi]^{2L}$; \\
$t \leftarrow 0$; \\
\While{True}{

\For{$j \leftarrow 1$ \KwTo $2L$}{
    Let $\vec x^{(t)}_{\neg j} = (x^{(t)}_1, \dots, x^{(t)}_{j-1}, x^{(t)}_{j+1}, \dots, x^{(t)}_{2L})$; \\
    Compute $C^{(t)}_j:= C_j(\mu; \vec x^{(t)}_{\neg j})$, $S^{(t)}_j:= S_j(\mu; \vec x^{(t)}_{\neg j})$, $C'^{(t)}_j:= C'_j(\mu; \vec x^{(t)}_{\neg j})$ and $S'^{(t)}_j:= S'_j(\mu; \vec x^{(t)}_{\neg j})$ by using the procedures in Lemma \ref{lem:eval_coeff_func_ab}; \\
    Compute $\Lambda(\mu; \vec x)$, $\Lambda'(\mu; \vec x)$ and their partial derivatives with respect to $x_j$ at $\vec x=\vec x^{(t)}$ as follows:
    \begin{align}
    &\Lambda^{(t)} := \Lambda(\mu; \vec x^{(t)}) = C^{(t)}_j \cosp{x_j} + S^{(t)}_j \sinp{x_j}, \\
    &\Lambda'^{(t)} := \Lambda'(\mu; \vec x^{(t)}) = C'^{(t)}_j \cosp{x_j} + S'^{(t)}_j \sinp{x_j}, \\
    &\chi^{(t)}_j := \dfrac{\partial \Lambda(\mu; \vec x)}{\partial x_j}|_{\vec x = \vec x^{(t)}}
    = -C^{(t)}_j \sinp{x_j} + S^{(t)}_j \cosp{x_j},  \\
    &\chi'^{(t)}_j := \dfrac{\partial \Lambda'(\mu; \vec x)}{\partial x_j}|_{\vec x = \vec x^{(t)}}
    = -C'^{(t)}_j \sinp{x_j} + S'^{(t)}_j \cosp{x_j}; 
    \end{align}
    Compute the partial derivative of $\mathcal{I}(\mu; f, \vec x)$ with respect to $x_j$ at $\vec x=\vec x^{(t)}$ as follows:
    \begin{align}
        \gamma^{(t)}_j := {\dfrac{\partial \mathcal{I}(\mu; f, \vec x)}{\partial x_j}}|_{\vec x=\vec x^{(t)}} = \dfrac{2f^2 \lrbb{(1-f^2 (\Lambda^{(t)})^2) \Lambda'^{(t)} \chi'^{(t)}_j
        +f^2 \Lambda^{(t)} \chi^{(t)}_j (\Lambda'^{(t)})^2}
        }{\lrbb{1-f^2 (\Lambda^{(t)})^2}^2}
    \end{align}
} 
Set $\vec x^{(t+1)}=\vec x^{(t)} + \delta(t) \nabla{\mathcal{I}(\mu; f, \vec x)} |_{\vec x=\vec x^{(t)}}$, where
$\nabla{\mathcal{I}(\mu; f, \vec x)}|_{\vec x=\vec x^{(t)}}=(\gamma^{(t)}_1, \gamma^{(t)}_2, \dots, \gamma^{(t)}_{2L})$;\\
\If{$|\mathcal{I}(\mu; f, \vec x^{(t+1)}) - \mathcal{I}(\mu; f, \vec x^{(t)})| < \epsilon$}{break;}
$t \leftarrow t+1$;
}
Return $\vec x^{(t+1)}=(x_1^{(t+1)}, x_2^{(t+1)}, \dots, x_{2L}^{(t+1)})$ as the optimal parameters.
\caption{Gradient ascent for Fisher information maximization in the ancilla-based case}
\label{alg:ga_opt_v0_ab}
\end{algorithm*}

\begin{algorithm*}[ht]
 \KwIn{The prior mean $\mu$ of $\theta$, the number $L$ of circuit layers, the fidelity $f$ of the process for generating the ELF, the error tolerance $\epsilon$ for termination.}
 \KwOut{A set of parameters $\vec x = (x_1, x_2, \dots, x_{2L}) \in (-\pi, \pi]^{2L}$ that are a local maximum point of the function $\mathcal{I}(\mu; f, \vec x)$.}

Choose random initial point $\vec x^{(0)}=(x^{(0)}_1, x^{(0)}_2, \dots, x^{(0)}_{2L}) \in (-\pi, \pi]^{2L}$; \\
$t \leftarrow 1$; \\
\While{True}{

\For{$j \leftarrow 1$ \KwTo $2L$}{
    Let $\vec x^{(t)}_{\neg j} = (x^{(t)}_1, \dots, x^{(t)}_{j-1}, x^{(t-1)}_{j+1}, \dots, x^{(t-1)}_{2L})$; \\
    Compute $C^{(t)}_j:= C_j(\mu; \vec x^{(t)}_{\neg j})$, $S^{(t)}_j := S_j(\mu; \vec x^{(t)}_{\neg j})$, $C'^{(t)}_j:= C'_j(\mu; \vec x^{(t)}_{\neg j})$ and $S'^{(t)}_j:= S'_j(\mu; \vec x^{(t)}_{\neg j})$ by using the procedures in Lemma \ref{lem:eval_coeff_func_ab}; \\
    Solve the single-variable optimization problem 
    $$\argmax_{z} \dfrac{f^2\lrb{C'^{(t)}_j \cosp{z} + S'^{(t)}_j \sinp{z}}^2}{1-f^2\lrb{C^{(t)}_j \cosp{z}+ S^{(t)}_j \sinp{z} }^2}$$
    by standard gradient-based methods and set $x^{(t)}_j$ to be its solution;
} 
\If{$|\mathcal{I}(\mu; f, \vec x^{(t)}) - \mathcal{I}(\mu; f, \vec x^{(t-1)})| < \epsilon$}{break;}
$t \leftarrow t+1$;
}
Return $\vec x^{(t)}=(x_1^{(t)}, x_2^{(t)}, \dots, x_{2L}^{(t)})$ as the optimal parameters.

\caption{Coordinate ascent for Fisher information maximization in the ancilla-based case}
\label{alg:ca_opt_v0_ab}
\end{algorithm*}

We have used Algorithms \ref{alg:ga_opt_v0_ab} and \ref{alg:ca_opt_v0_ab} to find the parameters $\vec x_{\theta} \in \mathbb{R}^{2L}$ that maximize $\mathcal{I}(\theta; f, \vec x)$ for various $\theta \in (0, \pi)$ (fixing $f$) and obtained Figure \ref{fig:fisher_information_ab}. This figure indicates that the Fisher information of AB ELF is larger than that of AB CLF for the majority of $\theta \in (0, \pi)$. Consequently, the estimation algorithm based on AB ELF is more efficient than the one based on AB CLF, as demonstrated in Section \ref{sec:simulatin_results}.

\begin{figure}[!ht]
\center
\includegraphics[width=0.9\linewidth]{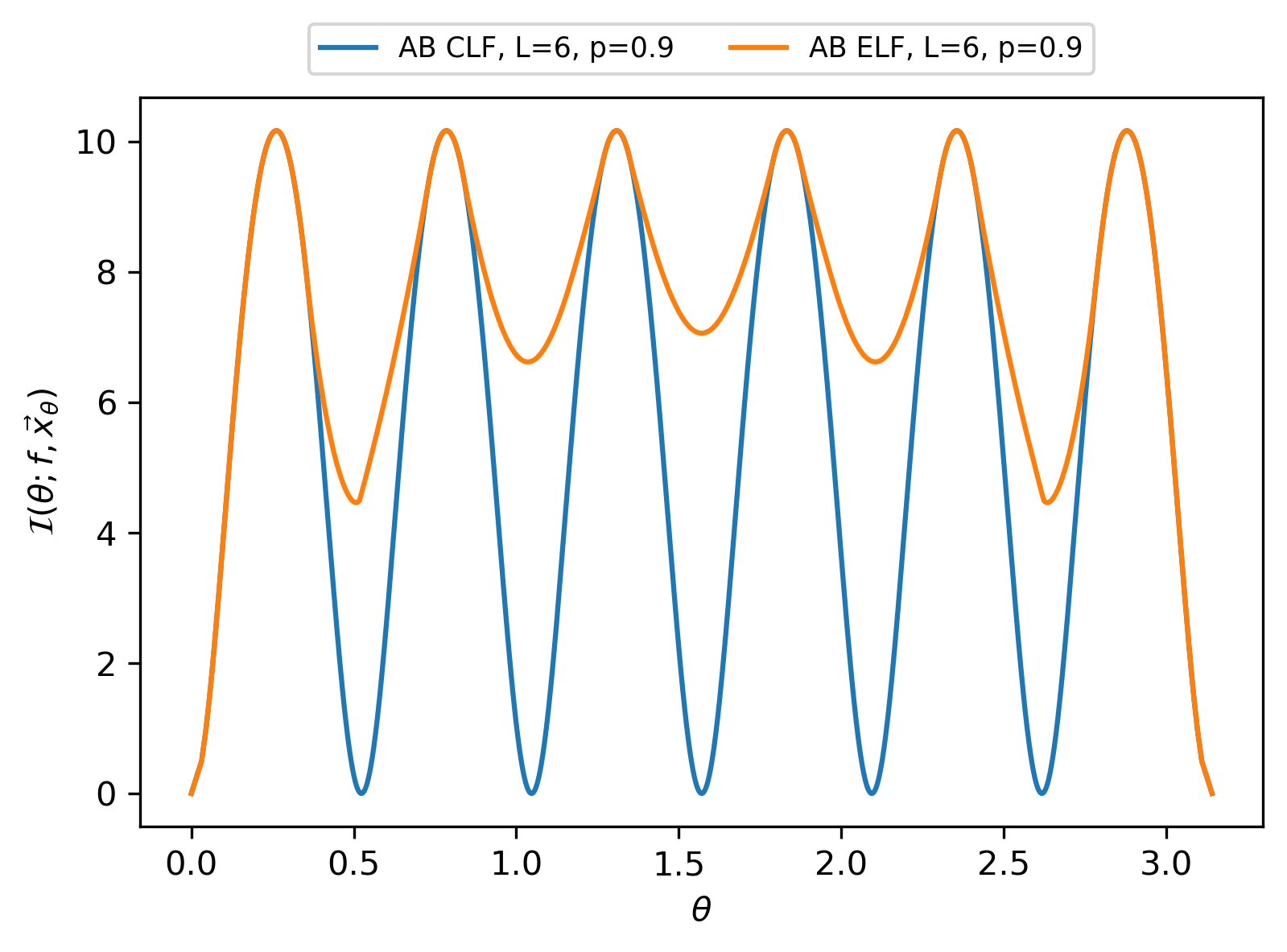}
\caption{This figure compares the Fisher information of AB ELF and AB CLF for various $\theta \in (0, \pi)$, when the number of circuit layers is $L=6$, the fidelity of each layer is $p=0.9$, and there is no SPAM error (i.e. $\bar{p}=1)$. For AB ELF, $\vec x_{\theta}$ is a global maximum point of $\mathcal{I}(\theta; f, \vec x)$ for given $\theta$ and $f=\bar{p}p^L=0.531441$. For AB CLF, $\vec x_{\theta}=(\pi/2, \pi/2, \dots, \pi/2)$ is fixed. One can see that the Fisher information of AB ELF is larger than that of AB CLF for the majority of $\theta \in (0, \pi)$. Furthermore, the Fisher information of AB CLF changes dramatically for different $\theta$'s (in fact, it is exactly $0$ when $\theta = j \pi/L$ for $j=0, 1, \dots, L$), whereas the Fisher information of AB ELF is less sensitive to the value of $\theta$.}
\label{fig:fisher_information_ab}
\end{figure}

\subsubsection{Maximizing the slope of the likelihood function}
\label{sec:slope_opt_ab}

We also propose two algorithms for maximizing the slope of the likelihood function $\mathbb{P}(d|\theta; f, \vec x)$ at a given point $\theta=\mu$ (i.e. the prior mean of $\theta$). Namely, our goal is to find $\vec x \in \mathbb{R}^{2L}$ that maximizes
$|\mathbb{P}'(d|\mu; f, \vec x)|=f|\Lambda'(\mu; \vec x)|/2$.

These algorithms are similar to Algorithms \ref{alg:ga_opt_slope_af} and \ref{alg:ca_opt_slope_af} for slope maximization in the ancilla-free case, in the sense that they are also based on gradient ascent and coordinate ascent, respectively. The main difference is that now we invoke the procedures in Lemma \ref{lem:eval_coeff_func_ab} to evaluate $C'(\mu; \vec x_{\neg j})$ and $S'(\mu; \vec x_{\neg j})$ for given $\mu$ and $\vec x_{\neg j}$. Then we use these quantities to either compute the partial derivative of $(\Lambda(\mu; \vec x))^2$ with respect to $x_j$ (in gradient ascent) or directly update the value of $x_j$ (in coordinate ascent). These algorithms are formally described in Algorithms \ref{alg:ga_opt_slope_ab} and \ref{alg:ca_opt_slope_ab}.

\begin{algorithm*}[ht]
 \KwIn{The prior mean $\mu$ of $\theta$, the number $L$ of circuit layers, the step size schedule $\delta: \mathbb{Z}^{\ge 0} \to \mathbb{R}^+$, the error tolerance $\epsilon$ for termination.}
 \KwOut{A set of parameters $\vec x = (x_1, x_2, \dots, x_{2L}) \in \mathbb{R}^{2L}$ that are a local maximum point of the function $|\Lambda'(\mu; \vec x)|$.}

Choose random initial point $\vec x^{(0)}=(x^{(0)}_1, x^{(0)}_2, \dots, x^{(0)}_{2L}) \in (-\pi, \pi]^{2L}$; \\
$t \leftarrow 0$; \\
\While{True}{

\For{$j \leftarrow 1$ \KwTo $2L$}{
    Let $\vec x^{(t)}_{\neg j} = (x^{(t)}_1, \dots, x^{(t)}_{j-1}, x^{(t)}_{j+1}, \dots, x^{(t)}_{2L})$; \\
    Compute $C'^{(t)}_j:= C'_j(\mu; \vec x^{(t)}_{\neg j})$ and $S'^{(t)}_j:= S'_j(\mu; \vec x^{(t)}_{\neg j})$ 
    by using the procedures in Lemma \ref{lem:eval_coeff_func_ab}; \\
    Compute $\Lambda'(\mu; \vec x)$ at $\vec x=\vec x^{(t)}$ as follows:
    \begin{align}
    \Lambda'^{(t)} := \Lambda'(\mu; \vec x^{(t)}) = C'^{(t)}_j \cosp{x^{(t)}_j} + S'^{(t)}_j \sinp{x^{(t)}_j};
    \end{align}
    Compute the partial derivative of $\Lambda'(\mu; \vec x)$ with respect to $x_j$ as follows: 
    \begin{align}
    \gamma^{(t)}_j := \dfrac{\partial \Lambda'(\mu; \vec x)}{\partial x_j} |_{\vec x = \vec x^{(t)}} = -C'^{(t)}_j \sinp{x^{(t)}_j}+S'^{(t)}_j \cosp{ x^{(t)}_j};    
    \end{align}
} 

Set $\vec x^{(t+1)}=\vec x^{(t)} + \delta(t) \nabla^{(t)}$, where
$\nabla^{(t)}:=  (2\Lambda'^{(t)}\gamma^{(t)}_1, 2\Lambda'^{(t)}\gamma^{(t)}_2, \dots, 2\Lambda'^{(t)}\gamma^{(t)}_{2L})$ is the gradient of $(\Lambda'(\mu; \vec x))^2$ at $\vec x=\vec x^{(t)}$;\\
\If{$|\Lambda'(\mu; \vec x^{(t+1)}) - \Lambda'(\mu; \vec x^{(t)})| < \epsilon$}{break;}
$t \leftarrow t+1$;
}
Return $\vec x^{(t+1)}=(x_1^{(t+1)}, x_2^{(t+1)}, \dots, x_{2L}^{(t+1)})$ as the optimal parameters.
\caption{Gradient ascent for slope maximization in the ancilla-based case}
\label{alg:ga_opt_slope_ab}
\end{algorithm*}

\begin{algorithm*}[ht]
 \KwIn{The prior mean $\mu$ of $\theta$, the number $L$ of circuit layers, the error tolerance $\epsilon$ for termination.}
 \KwOut{A set of parameters $\vec x = (x_1, x_2, \dots, x_{2L}) \in (-\pi, \pi]^{2L}$ that are a local maximum point of the function $|\Lambda'(\mu; \vec x)|$.}

Choose random initial point $\vec x^{(0)}=(x^{(0)}_1, x^{(0)}_2, \dots, x^{(0)}_{2L}) \in (-\pi, \pi]^{2L}$; \\
$t \leftarrow 1$; \\
\While{True}{

\For{$j \leftarrow 1$ \KwTo $2L$}{
    Let $\vec x^{(t)}_{\neg j} = (x^{(t)}_1, \dots, x^{(t)}_{j-1}, x^{(t-1)}_{j+1}, \dots, x^{(t-1)}_{2L})$; \\
    Compute $C'^{(t)}_j:= C'_j(\mu; \vec x^{(t)}_{\neg j})$ and $S'^{(t)}_j:= S'_j(\mu; \vec x^{(t)}_{\neg j})$ by using the procedures in Lemma \ref{lem:eval_coeff_func_ab}; \\
    Set $x^{(t)}_j = \arg{C'^{(t)}_j+i S'^{(t)}_j}$;
} 
\If{$|\Lambda'(\mu; \vec x^{(t)}) - \Lambda'(\mu; \vec x^{(t-1)})| < \epsilon$}{break;}
$t \leftarrow t+1$;
}
Return $\vec x^{(t)}=(x_1^{(t)}, x_2^{(t)}, \dots, x_{2L}^{(t)})$ as the optimal parameters.

\caption{Coordinate ascent for slope optimization in the ancilla-based case}
\label{alg:ca_opt_slope_ab}
\end{algorithm*}

We have also used Algorithms \ref{alg:ga_opt_slope_ab} and \ref{alg:ca_opt_slope_ab} to find the parameters $\vec x_{\theta}$ that maximize $|\Lambda'(\theta; \vec x)|$ for various $\theta \in (0, \pi)$ and obtained Figure \ref{fig:slope_af}. This figure implies that the slope-based AB ELF is steeper than AB CLF for the majority of $\theta \in (0, \pi)$ and hence has more statistical power than AB CLF (at least) in the low-fidelity setting by Eq.~\eqref{eq:v_slope_ab}.

\begin{figure}[!ht]
\center
\includegraphics[width=0.9\linewidth]{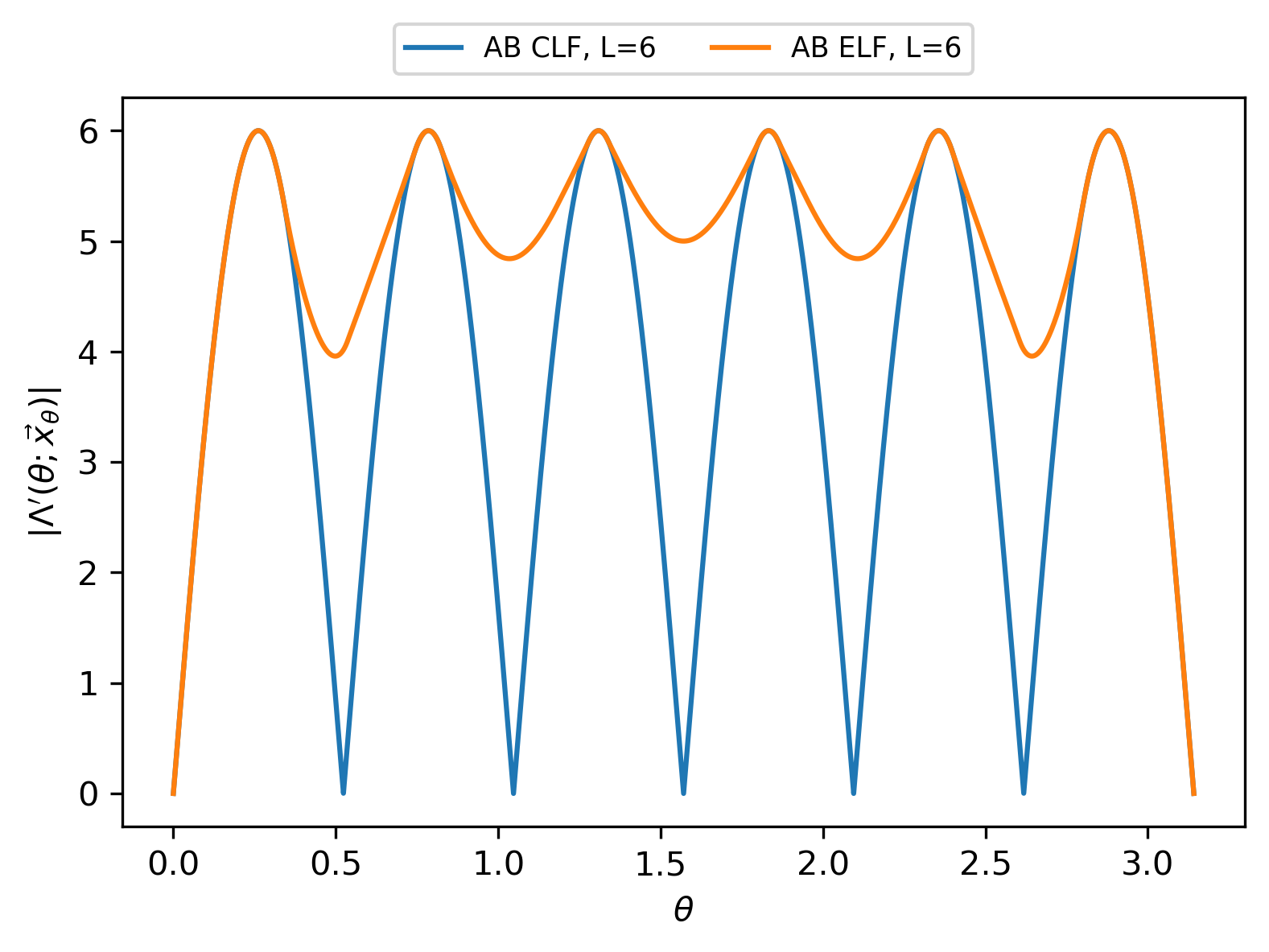}
\caption{This figure compares the values of $|\Lambda'(\theta; \vec x_{\theta})|$ for slope-based AB ELF and AB CLF for various $\theta \in (0, \pi)$, when the number of circuit layers is $L=6$. For AB ELF, $\vec x_{\theta}$ is a global maximum point of $|\Lambda'(\theta; \vec x)|$
for given $\theta$. For AB CLF, $\vec x_{\theta}=(\pi/2, \pi/2, \dots, \pi/2)$ is fixed. This figure implies that the slope-based AB ELF is steeper than AB CLF for the majority of $\theta \in (0, \pi)$. Furthermore, the slope of AB CLF changes dramatically for different $\theta$'s (in fact, it is exactly $0$ when $\theta = j \pi/L $ for $j=0, 1, \dots, L$), whereas the slope of AB ELF is less sensitive to the value of $\theta$.}
\label{fig:slope_ab}
\end{figure}

\subsection{Approximate Bayesian inference with engineered likelihood functions}
\label{sec:bi_elf_ab}
With the algorithms for tuning the circuit parameters $\vec x$ in place, we now briefly describe how to perform Bayesian inference efficiently with the resultant likelihood functions. The idea is similar to the one in Section \ref{sec:bi_elf_af} for the ancilla-free scheme.

Suppose $\theta$ has prior distribution $\mathcal{N}(\mu, \sigma^2)$, where $\sigma \ll 1/L$, and the fidelity of the process for generating the ELF is $f$. We find that the parameters $\vec x=(x_1, x_2, \dots, x_{2L})$ that maximize $\mathcal{I}(\mu; f, \vec x)$ (or $|\Lambda'(\mu; \vec x)|$) satisfy the following property: When $\theta$ is close to $\mu$, i.e. $\theta \in [\mu-O(\sigma), \mu+O(\sigma)]$, we have
\begin{align}
\mathbb{P}(d|\theta; f, \vec x) \approx \dfrac{1+(-1)^d f \sinp{r\theta+b}}{2}    
\end{align}
for some $r, b \in \mathbb{R}$. We find the best-fitting $r$ and $b$ by solving the following least squares problem:
\begin{align}
(r^*, b^*) = \argmin_{r, b} \sum_{\theta \in \Theta}  \left|\arcsinp{\Lambda(\theta;\vec x)}-r\theta-b \right|^2,
\label{eq:lstsq_ab}
\end{align}
where $\Theta=\{\theta_1, \theta_2, \dots, \theta_k\} \subseteq [\mu-O(\sigma), \mu+O(\sigma)]$. This least-squares problem has the following analytical solution:
\begin{align}
\begin{pmatrix}
r^*\\
b^*
\end{pmatrix}
= A^+ z
= (A^TA)^{-1} A^T z,
\end{align}
where 
\begin{align}
A =
\begin{pmatrix}
\theta_1 &  1 \\
\theta_2 &  1 \\
\vdots & \vdots\\
\theta_k & 1
\end{pmatrix}, 
\quad z = \begin{pmatrix}
\arcsinp{\Lambda(\theta_1;\vec x)} \\
\arcsinp{\Lambda(\theta_2;\vec x)} \\
\vdots \\
\arcsinp{\Lambda(\theta_k;\vec x)}
\end{pmatrix}.
\end{align}
Figure~\ref{fig:fit_lf_ab} illustrates an example of the true and fitted likelihood functions.
\begin{figure}[!ht]
\center
\includegraphics[width=0.9\linewidth]{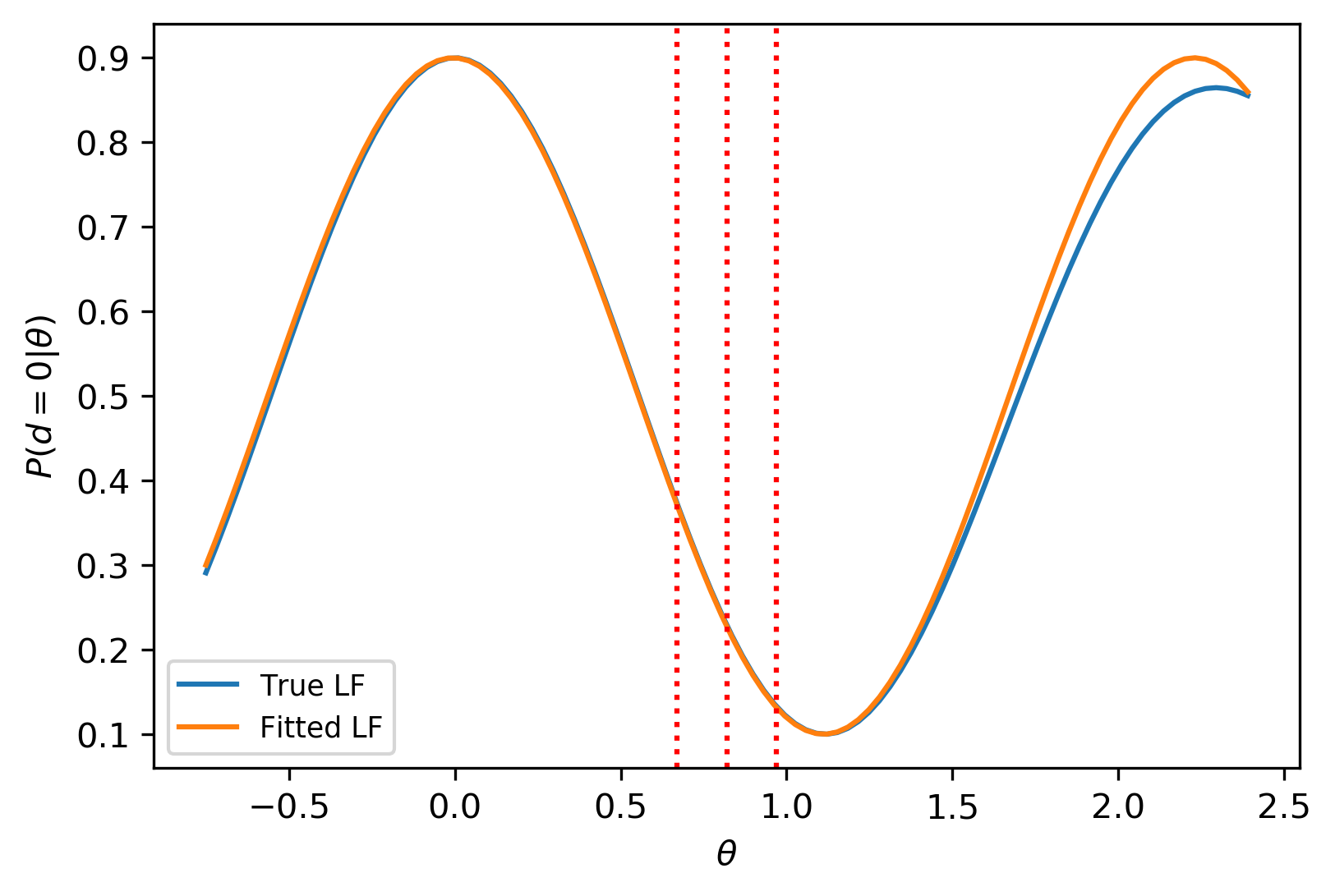}
\caption{The true and fitted likelihood functions when $L=3$, $f=0.8$, and $\theta$ has prior distribution $\mathcal{N}(0.82, 0.0009)$. The true likelihood function is generated by Algorithm \ref{alg:ca_opt_v0_ab}. During the sinusoidal fitting of this function, we
set $\Theta=\{\mu-\sigma, \mu-0.8\sigma, \dots, \mu+0.8\sigma, \mu+\sigma\}$ (i.e. $\Theta$ contains $11$ uniformly distributed points in $[\mu-\sigma, \mu+\sigma]$) in Eq.~(\ref{eq:lstsq_ab}). The fitted likelihood function is $\mathbb{P}(d|\theta)=(1+(-1)^d f\sinp{r \theta + b })/2$, where $r=-2.81081$ and $b=1.55477$. Note that the true and fitted likelihood functions are close for $\theta \in [0.67, 0.97]$. }
\label{fig:fit_lf_ab}
\end{figure}

Once we obtain the optimal $r$ and $b$, we approximate the posterior mean and variance of $\theta$ with the ones for 
\begin{equation}
\mathbb{P}(d|\theta; f) = \dfrac{1+(-1)^d f \sinp{r\theta+b}}{2},
\end{equation}
which have analytical formulas. Specifically, suppose $\theta$ has prior distribution $\mathcal{N}(\mu_k, \sigma_k^2)$ at round $k$. Let $d_k$ be the measurement outcome and $(r_k, b_k)$ be the best-fitting parameters at this round. Then we approximate the posterior mean and variance of $\theta$ by
\begin{widetext}
\begin{eqnarray}
\mu_{k+1} & =  &
\mu_k+\dfrac{(-1)^{d_k} f e^{-r_k^2\sigma_k^2/2}r_k \sigma_k^2\cosp{r_k\mu_k+b_k}}{1+(-1)^{d_k} f  e^{-r_k^2\sigma_k^2/2}\sinp{r_k\mu_k+b_k}}, 
\label{eq:bayes_update_mu_ab} \\
\sigma_{k+1}^2 & = & 
\sigma_k^2 \left (1-\dfrac{f r_k^2 \sigma_k^2 e^{-r_k^2\sigma_k^2/2} [f e^{-r_k^2\sigma_k^2/2}+(-1)^{d_k}\sinp{r_k\mu_k+b_k}]}{[1+(-1)^{d_k} f e^{-r_k^2\sigma_k^2/2}\sinp{r_k\mu_k+b_k}]^2}\right ).
\label{eq:bayes_update_sigma_ab}
\end{eqnarray}
\end{widetext}
After that, we proceed to the next round, setting $\mathcal{N}(\mu_{k+1}, \sigma_{k+1}^2)$ as the prior distribution of $\theta$ for that round. The approximation errors incurred by Eqs.~(\ref{eq:bayes_update_mu_ab}) and (\ref{eq:bayes_update_sigma_ab}) are small and have negligible impact on the performance of the whole algorithm for the same reason as in the ancilla-free case.

\section{Bias and variance of the ELF-based estimator of \texorpdfstring{$\Pi$}{Pi}}
\label{sec:evidences}
In this appendix, we state two facts about the ELF-based estimator $\hat{\mu}_t$ of $\Pi$ during the inference process (recall that we use the Gaussian distribution $\mathcal{N}(\hat{\mu}_t, \hat{\sigma}_t^2)$ to represent our knowledge of $\Pi$ at time $t$, where $\hat{\mu}_t$ and $\hat{\sigma}_t$ are random variables depending on the history of random measurement outcomes up to time $t$):
\begin{itemize}
    \item In both the ancilla-based and ancilla-free cases, the squared bias of $\hat{\mu}_t$ is smaller than its variance, i.e.  
    $\operatorname{Bias}(\hat{\mu}_t)^2=(\mathbb{E}[\hat{\mu}_t] - \Pi)^2 < \operatorname{Var}(\hat{\mu}_t)$, for large enough $t$. See Figure \ref{fig:bias_vs_std} for examples.
    
    \item In both the ancilla-based and ancilla-free cases, the perceived variance $\hat{\sigma}_t^2$ of $\Pi$ is often close to the variance of $\hat{\mu}_t$, i.e. $\hat{\sigma}_t^2 \approx \operatorname{Var}(\hat{\mu}_t)$ with high probability, for large enough $t$. See Figure \ref{fig:two_variances} for examples.
\end{itemize}
Combining these facts, we know that for large enough $t$, $\operatorname{MSE}(\hat{\mu}_t)=\mathbb{E}[(\hat{\mu}_t-\Pi)^2] \le 2 \hat{\sigma}_t^2$ with high probability. Similar statements hold for the ELF-based estimator $\mu_t$ of $\theta$ during the inference process.

\begin{figure*}[!ht]
\begin{minipage}{.48\textwidth}
\center
\includegraphics[width=0.95\linewidth]{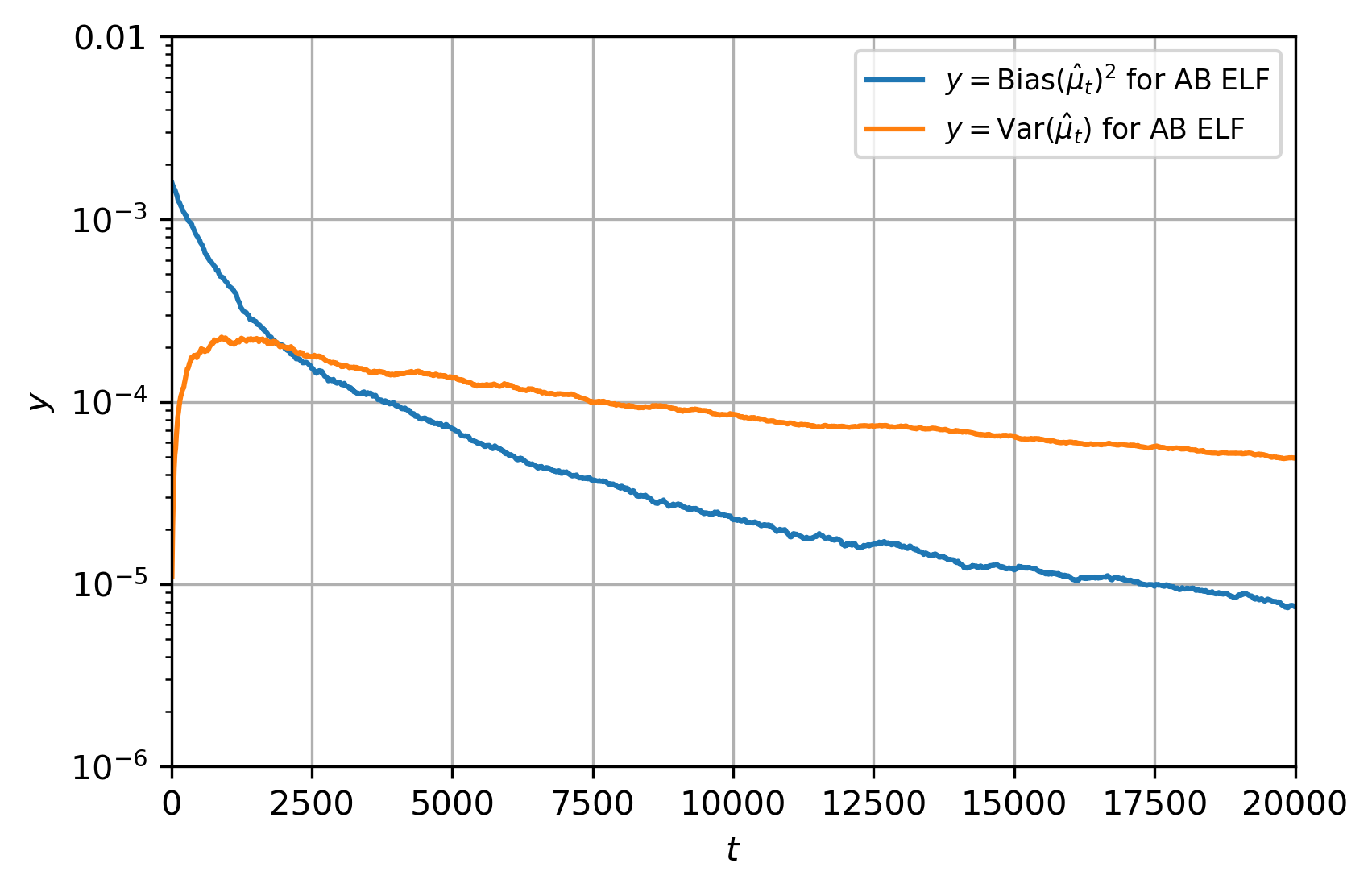}
\end{minipage}
\begin{minipage}{.48\textwidth}
\center
\includegraphics[width=0.95\linewidth]{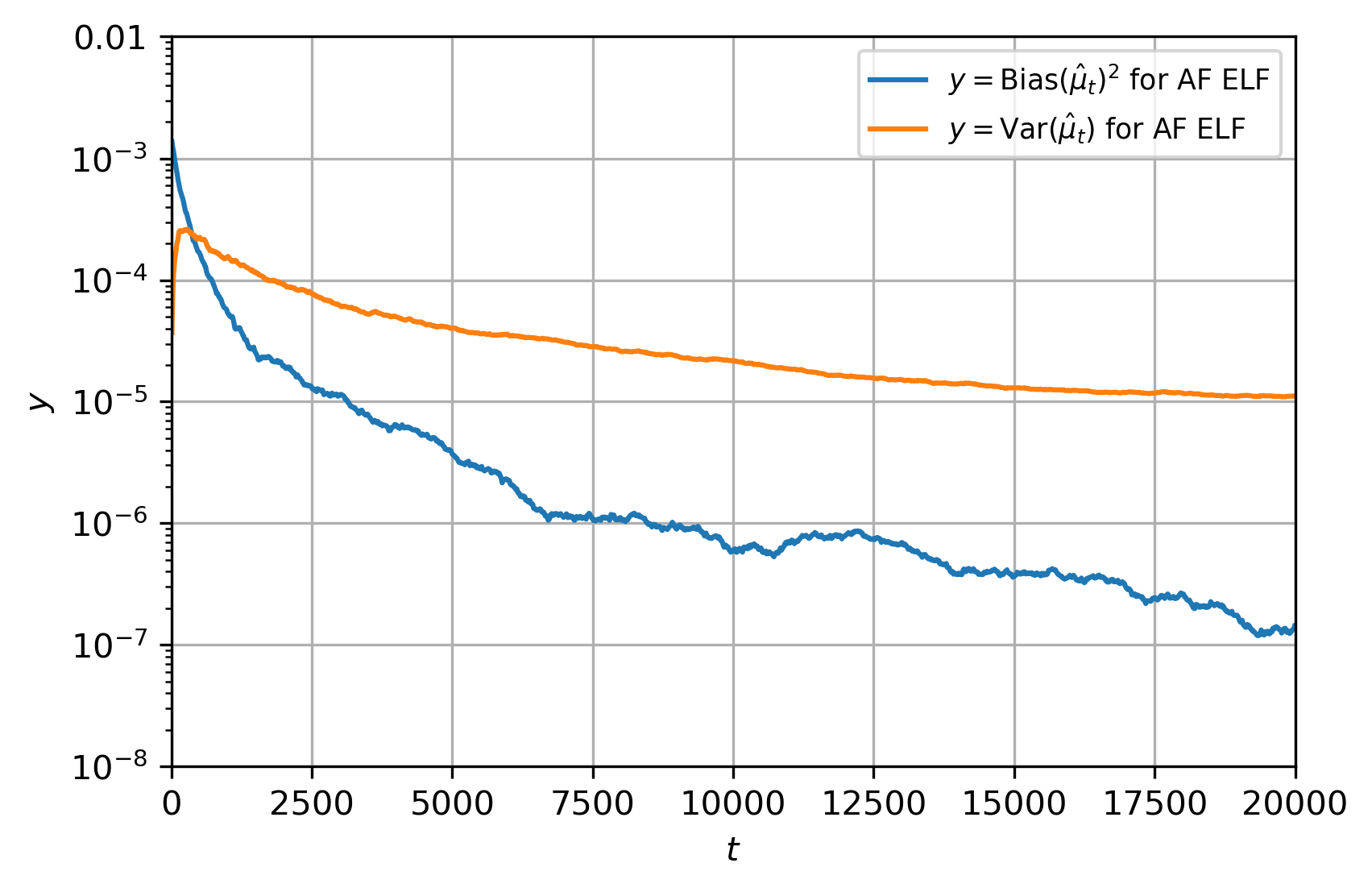}
\end{minipage}
\caption{This figure compares $\operatorname{Bias}(\hat{\mu}_t)^2$ and $\operatorname{Var}(\hat{\mu}_t)$ during the inference process for AB ELF and AF ELF, respectively. Here $\Pi$ has true value $0.6$ and prior distribution $\mathcal{N}(0.64, 0.0009)$, the number of circuit layers is $L=6$, each layer has fidelity $p=0.9$ and there is no SPAM error (i.e. $\bar{p}=1$). Each plot is generated by simulating the inference process $300$ times.
Note that $\operatorname{Bias}(\hat{\mu}_t)^2<\operatorname{Var}(\hat{\mu}_t)$ as soon as $t$ becomes sufficiently large. The variance of $\hat{\mu}_t$ increases at the early stage of the algorithm, because we always start with the same $\hat{\mu}_0$ and it takes certain number of Bayesian updates for $\hat{\mu}_t$ to become sufficiently dispersed.}
\label{fig:bias_vs_std}
\end{figure*}

\begin{figure*}[!ht]
\begin{minipage}{.48\textwidth}
\center
\includegraphics[width=0.95\linewidth]{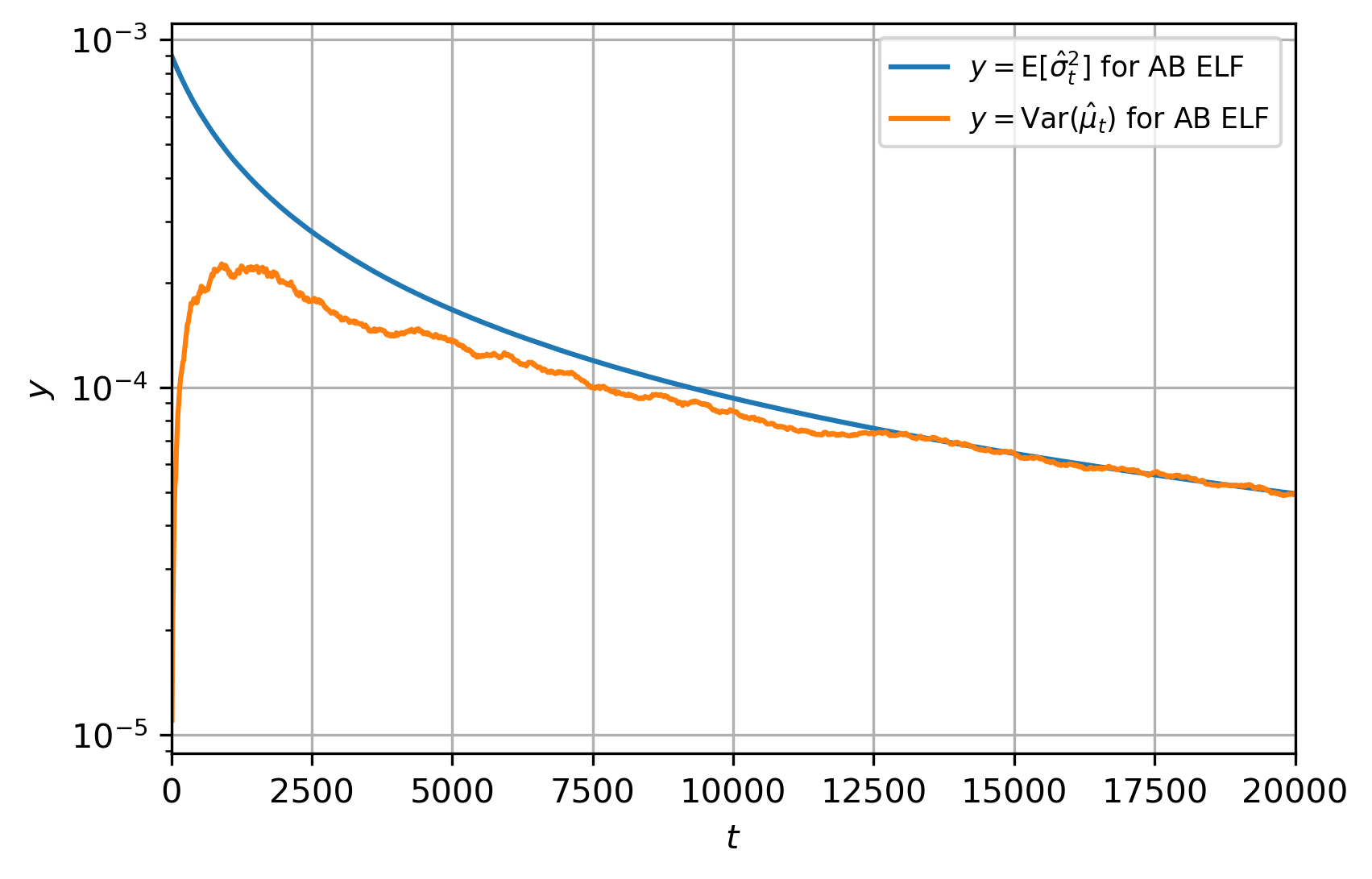}
\end{minipage}
\begin{minipage}{.48\textwidth}
\center
\includegraphics[width=0.95\linewidth]{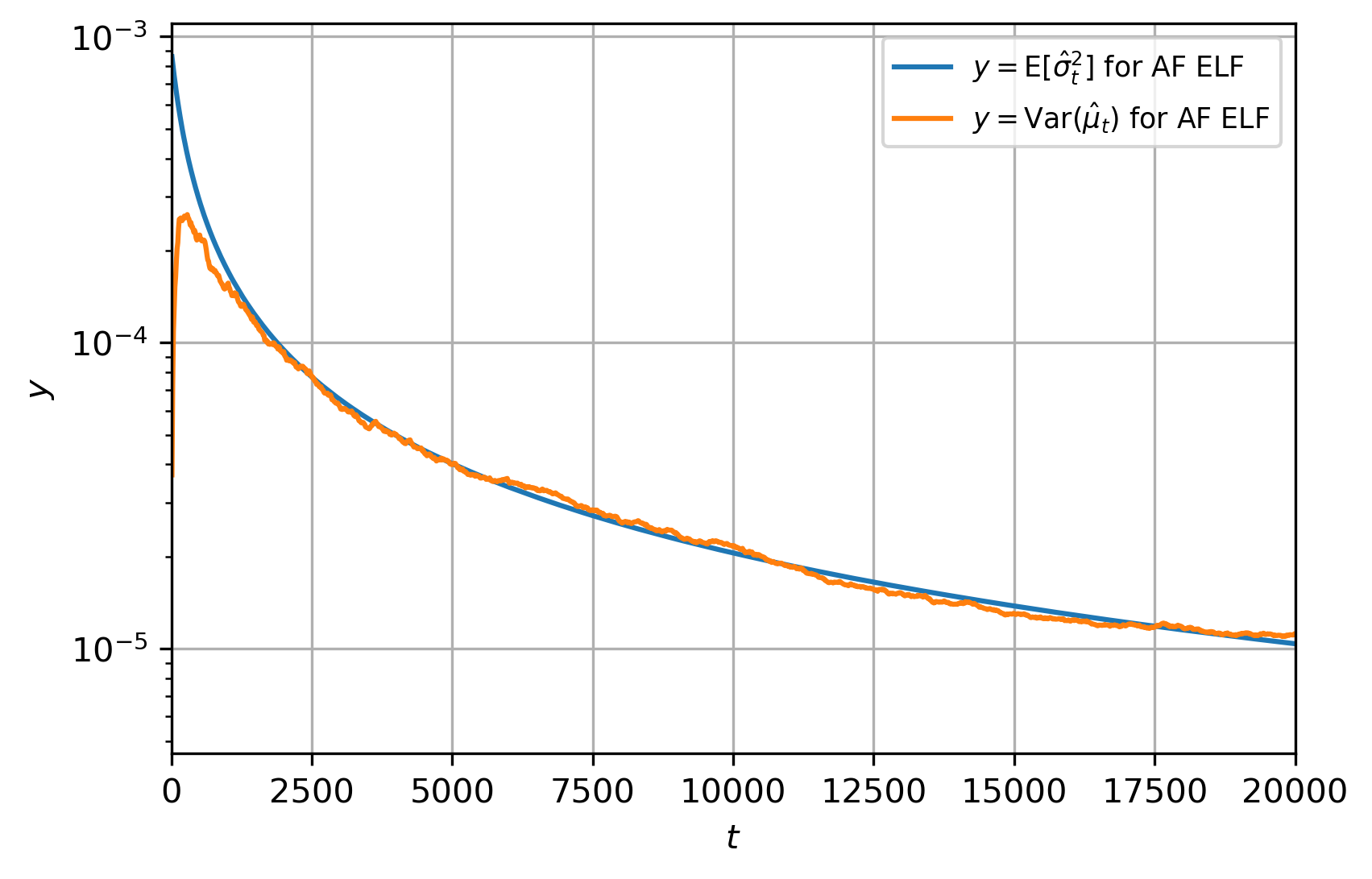}
\end{minipage}
\caption{This figure compares $\mathbb{E}[\hat{\sigma}_t^2]$ and 
$\operatorname{Var}(\hat{\mu}_t)$ during the inference process for AB ELF and AF ELF, respectively. Here $\Pi$ has true value $0.6$ and prior distribution $\mathcal{N}(0.64, 0.0009)$, the number of circuit layers is $L=6$, each layer has fidelity $p=0.9$ and there is no SPAM error (i.e. $\bar{p}=1$). Each plot is generated by simulating the inference process $300$ times.Note that $\mathbb{E}[\hat{\sigma}_t^2] \approx \operatorname{Var}(\hat{\mu}_t)$ once $t$ becomes sufficiently large. Furthermore, we discover that $\hat{\sigma}_t^2$ does not vary much among different runs of the algorithm. Namely, $\hat{\sigma}_t^2 \approx \mathbb{E}[\hat{\sigma}_t^2]$ with high probability. So $\hat{\sigma}_t^2 \approx \operatorname{Var}(\hat{\mu}_t)$ with high probability for large enough $t$. Note that the variance of $\hat{\mu}_t$ increases at the early stage of the algorithm, because we always start with the same $\hat{\mu}_0$ and it takes certain number of Bayesian updates for $\hat{\mu}_t$ to become sufficiently dispersed.}
\label{fig:two_variances}
\end{figure*}

\section{Comparison of exact optimization with optimization of proxies}
\label{sec:proxy_comparison}

In this appendix, we compare the maximization of the variance reduction factor with the maximization of the proxies used in Section \ref{sec:optalg}. We start with motivating the use of these proxies by studying the limiting behavior of $\mathcal{V}(\mu, \sigma; f, \vec x)$ as $\sigma\rightarrow 0$ or as $f\rightarrow 0$ or $1$.

\subsection{Limiting behavior of the variance reduction factor}
\label{sec:limitingBehaviorV}
Consider the following three limiting situations:
(i) when the variance vanishes, i.e. $\sigma^2 \rightarrow 0$, (ii) when the fidelity is close to zero, i.e. $f\approx 0$, and (iii) when the fidelity is equal to one, i.e. $f=1$. We will derive expressions for the variance reduction factor in these conditions (see Figure \ref{fig:flowchartVfactorLimits} for a flowchart summarizing the results below). For simplicity, we will assume in this section that either (i) $0 \leq f < 1$ or (ii) $f=1$ and $|\Delta(\mu;\vec x)|\neq 1$. While the results here are stated in terms of the ancilla-free bias \eqref{eq:delta_af}, they also hold for the ancilla-based bias \eqref{eq:delta_ab} (just replace $\Delta$ with $\Lambda$).

\tikzstyle{blankBlock} = [rectangle]
\tikzstyle{redBlock} = [rectangle, rounded corners, minimum width=3.5cm, minimum height=1.5cm,text centered, draw=orange,fill=orange!15,thick]
\tikzstyle{block} = [rectangle, rounded corners, minimum width=3.5cm, minimum height=1.5cm,text centered, draw=black,fill=yellow!10]
\tikzstyle{arrow} = [->,>=stealth]

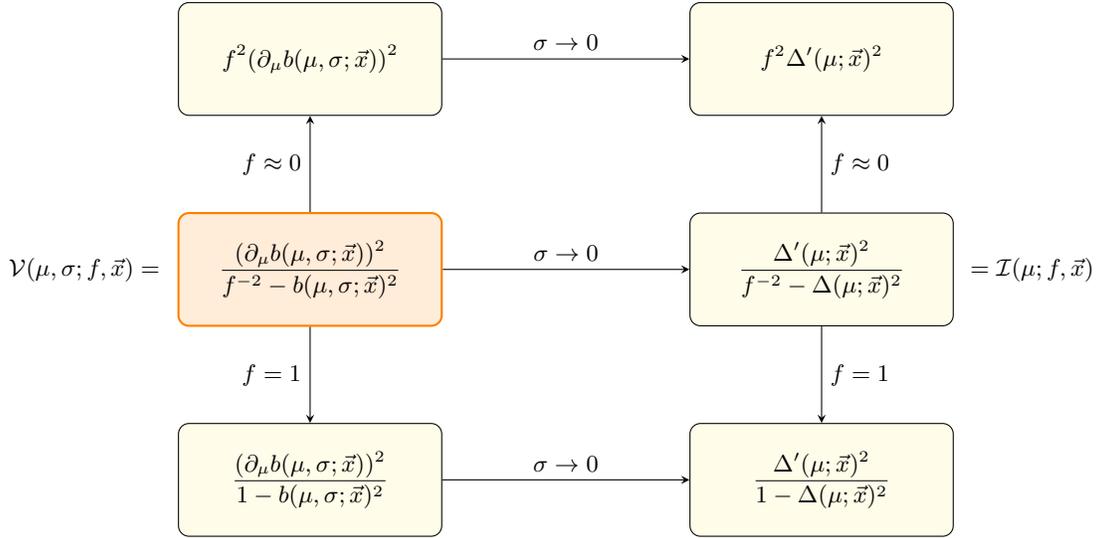
\begin{figure*}[!ht]
    \begin{tikzpicture}[node distance = 2.8cm, auto]
  \node [block] (00) {$\displaystyle{\frac{(\partial_{\mu} b(\mu,\sigma;\vec x))^2}{1-b(\mu,\sigma;\vec x)^2}}$};
\node [redBlock, above of=00] (01) {$\displaystyle{ \frac{(\partial_{\mu} b(\mu,\sigma;\vec x))^2}{f^{-2}-b(\mu,\sigma;\vec x)^2}}$};
  \node[blankBlock, left of=01, xshift=-0.2cm, yshift=0cm] (out) {$\mathcal V(\mu,\sigma;f, \vec x) =$};
  \node [block, above of=01] (02) {$\displaystyle{f^2(\partial_{\mu} b(\mu,\sigma;\vec x))^2}$};   
  \node [block, right of=00, xshift=4cm] (10) {$ \displaystyle{\frac{\Delta'(\mu; \vec x)^2}{1-\Delta(\mu; \vec x)^2}}$};
  \node [block, above of=10] (11) {$\displaystyle{\frac{\Delta'(\mu; \vec x)^2}{f^{-2}-\Delta(\mu; \vec x)^2}}$};
  \node [block, above of=11] (12) {$f^2\Delta'(\mu; \vec x)^2$};
    \node[blankBlock, right of=11] (out3) {$=\mathcal I(\mu;f, \vec x)
  $};

  \draw [arrow] (01) -- node[anchor=east] {$f=1$} (00);
  \draw [arrow] (01) -- 
  node[anchor=east] {$f\approx 0$} (02);
  \draw [arrow] (11) -- node[anchor=west] {$f=1$}(10);
  \draw [arrow] (11) -- node[anchor=west] {$f\approx 0$}(12);
  \draw [arrow] (02) -- node[anchor=south] {$\sigma\rightarrow 0$}(12);
  \draw [arrow] (01) -- node[anchor=south] {$\sigma\rightarrow 0$}(11);
  \draw [arrow] (00) -- node[anchor=south] {$\sigma\rightarrow 0$}(10);
   
\end{tikzpicture}
\caption{Flowchart showing the behavior of the expected posterior variance $\displaystyle{\mathcal V(\mu,\sigma;f, \vec x)}$ (in orange box) as (i) $\sigma \rightarrow 0$, (ii) $f\approx 0$, (iii) $f=1$. For ease of notation we use $f=\bar{p}p^{L}$, where $p$ is the fidelity of each circuit layer and $\bar{p}$ captures the SPAM error.}
\label{fig:flowchartVfactorLimits}
\end{figure*}

\renewcommand\labelitemi{\tiny$\bullet$}
\begin{itemize}
    \item In Case (i), as $\sigma \rightarrow 0$, the expected bias \eqref{eq:b} behaves as 
\begin{align}
    b(\mu,\sigma;\vec x)    &\rightarrow \Delta(\mu;\vec x),
    \nonumber\\
    \partial_{\mu} b(\mu,\sigma;\vec x) &\rightarrow \Delta'(\mu;\vec x),
\end{align}
which implies that the limit of the variance reduction factor as $\sigma \rightarrow 0$ is
\begin{align}
    \mathcal V_0(\mu;f, \vec x) := \lim_{\sigma \rightarrow 0} \mathcal V(\mu,\sigma;f, \vec x) =\mathcal{I}(\mu; f, \vec x),
    \label{eq:V0Limit}
\end{align}
where
\begin{align}
    \mathcal{I}(\mu; f, \vec x)  := \frac{f^2 \Delta'(\mu; \vec x)^2}{1-f^2\Delta(\mu; \vec x)^2}
\end{align}
is the Fisher information corresponding to the noisy version of the engineered likelihood function given by Eq.~\eqref{eq:lf_af}.

\item
In Case (ii), we get
\begin{align}
    \mathcal V(\mu,\sigma;f, \vec x) \approx f^2 (\partial_{\mu} b(\mu,\sigma;\vec x))^2, \quad\mbox{for $f\approx 0$}.
\end{align}
\item In Case (iii), we get
\begin{align}
    \mathcal V(\mu,\sigma;f, \vec x) 
    = 
    \frac{(\partial_{\mu} b(\mu,\sigma;\vec x))^2}{1-b(\mu,\sigma;\vec x)^2}, \quad\mbox{for $f= 1$}.
\end{align}

\item Combining Cases (i) and (ii), we get
\begin{align}
    \mathcal V(\mu,\sigma;f, \vec x) \approx f^2 \Delta'(\mu;\vec x)^2, \quad\mbox{for $f\approx 0$, $\sigma \approx 0$.}
    \label{eq:V0limitf0sigma0}
\end{align}

\item Combining Cases (i) and (iii), we get
\begin{align}
    \lim_{\sigma \rightarrow 0}
    \mathcal V(\mu,\sigma;f, \vec x) =
    \frac{\Delta'(\mu; \vec x)^2}{1-\Delta(\mu; \vec x)^2} = \mathcal{I}(\mu; 1, \vec x),\quad\mbox{for $f=1$.}
    \label{eq:Vlimit}
\end{align}
\end{itemize}

In the next part, we will show how these approximations give us good proxies for maximizing the variance reduction factor.

\subsection{Implementing the exact variance reduction factor optimization and comparison with proxies}
\label{subsec:proxy-v-factor}

Consider the following optimization problems: 
\begin{align}
\begin{tabular}{cp{8cm}}
$\mathtt{Input}$: 
&
$(\mu, f)$, where $\mu\in \mathbb R, f\in [0,1]$
\\[0.2cm]
$\mathtt{Output}$:
&
$\displaystyle{ \argmax_{\vec x \in (-\pi,\pi]^{2L} }\mathcal I(\mu;f, \vec x)
= \argmax_{\vec x \in (-\pi,\pi]^{2L} } \frac{\Delta'(\mu; \vec x)^2}{f^{-2}-\Delta(\mu; \vec x)^2}
}
$.
\end{tabular}
\label{eq:optProblemV0}
\end{align}
and 
\begin{align}
\begin{tabular}{cp{8cm}}
$\mathtt{Input}$: 
&
$\mu\in \mathbb R$
\\[0.2cm]
$\mathtt{Output}$:
&
$\displaystyle \argmax_{\vec x \in (-\pi,\pi]^{2L} }|\Delta'(\mu; \vec x)|$.
\end{tabular}
\label{eq:optProblemSlope}
\end{align}

By Eq.~\eqref{eq:V0Limit}, we expect that a solution to \eqref{eq:optProblemV0} would be a good proxy for maximizing the expected posterior variance when $\sigma$ is small, i.e. if $\hat x_{\mu,f}$ is a solution to \eqref{eq:optProblemV0} on input $(\mu,f)$, then we expect that
\begin{align}
    \max_{\vec x \in (-\pi,\pi]^{2L} }\mathcal V(\mu,\sigma;f, \vec x) & \approx \frac{\Delta'(\mu; \hat x_{\mu,f})^2}{f^{-2}-\Delta(\mu; \hat x_{\mu,f})^2},
    \nonumber\\
    &\qquad\mbox{when $\sigma$ is small.}
    \label{eq:proxy_sigma_small}
\end{align}

Similarly, by Eq.~\eqref{eq:V0limitf0sigma0}, we expect that a solution to \eqref{eq:optProblemSlope} would be a good proxy for maximizing the expected posterior variance when both $f$ and $\sigma$ is small, i.e. if $\hat x_\mu$ is a solution to \eqref{eq:optProblemSlope} on input $\mu$, then we expect that
\begin{align}
    \max_{\vec x \in (-\pi,\pi]^{2L} }\mathcal V(\mu,\sigma;f, \vec x) &\approx f^2 \Delta'(\mu;\hat x_\mu)^2,
    \nonumber\\
    &\qquad\mbox{when $\sigma$ and $f$ are small.}
    \label{eq:proxy_sigma_f_small}
\end{align}

To investigate the performance of the proxies \eqref{eq:proxy_sigma_small} and \eqref{eq:proxy_sigma_f_small}, we numerically maximize the expected posterior variance and the proxies \eqref{eq:optProblemV0} and \eqref{eq:optProblemSlope} for small problem sizes $L$. The results of this optimization are presented in Figures \ref{fig:grid_L1}--\ref{fig:grid_L2_AB}. For $L=1$, it turns out that the optimization problem \eqref{eq:optProblemSlope} for the ancilla-free bias can be solved analytically. We present this analytical solution in Appendix \ref{sec:L1example}.

\begin{figure*}[!ht]
    \includegraphics[width=0.8\linewidth]{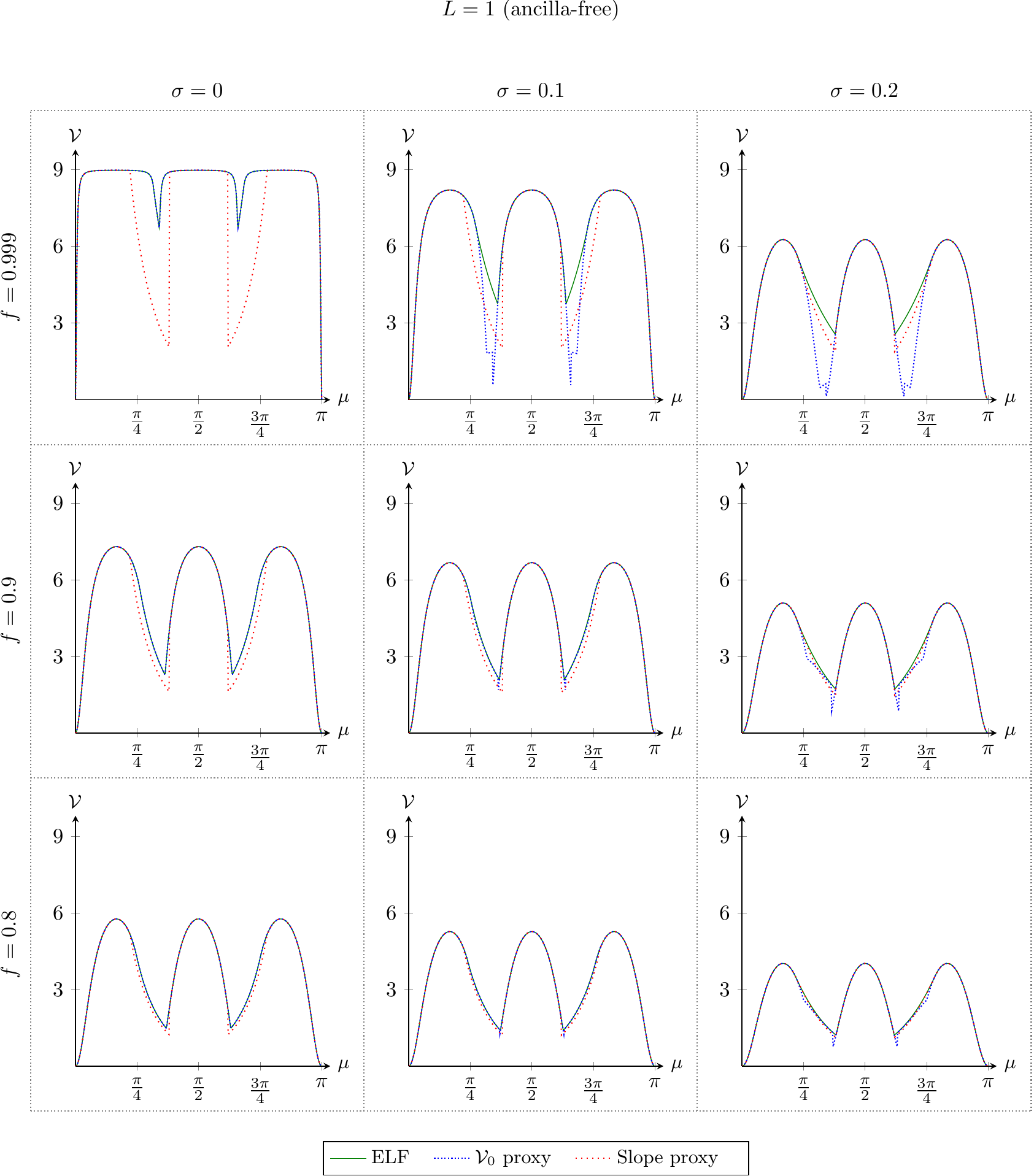}
    \caption{Plots of the variance reduction factor $\mathcal V$ vs the prior mean $\mu$ for $L=1$ for the ancilla-free scheme. The proxies work well when $f$ and $\sigma$ are not too large (say $f\leq 0.9$ and $\sigma \leq 0.1$). When $f$ and $\sigma$ are both large (for example, $f=0.999$ and $\sigma=0.2$), the proxies fail to be good approximations). These figures were obtained with Wolfram Mathematica's \texttt{NMaximize}
RandomSearch method with 120 search points.}
    \label{fig:grid_L1}
\end{figure*}

\begin{figure*}[!ht]
    \includegraphics[width=0.8\linewidth]{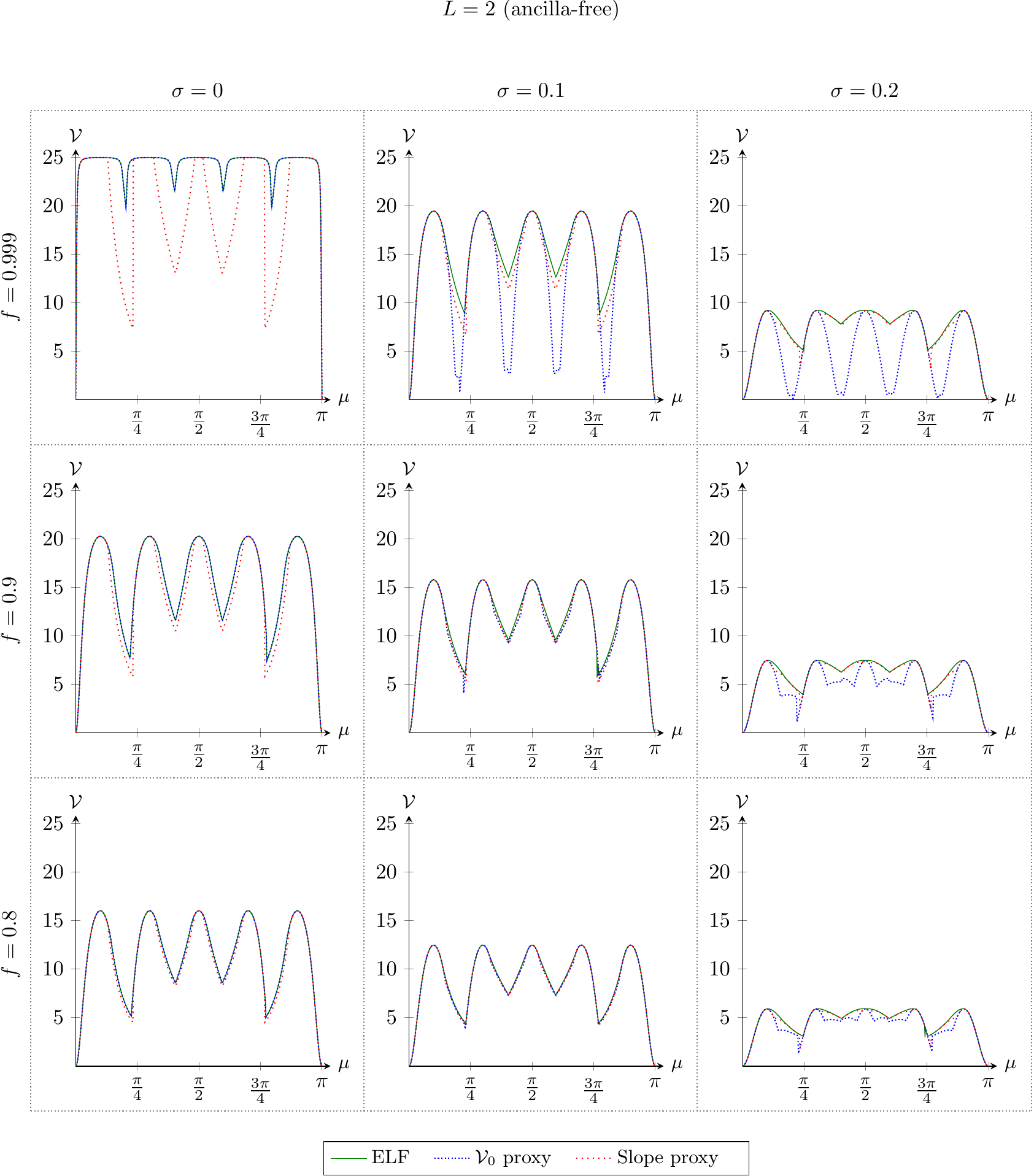}
    \caption{Plots of the variance reduction factor $\mathcal V$ vs the prior mean $\mu$ for $L=2$ for the ancilla-free scheme. The proxies work well when $f$ and $\sigma$ are not too large (say $f\leq 0.9$ and $\sigma \leq 0.1$). When $f$ and $\sigma$ are both large (for example, $f=0.999$ and $\sigma=0.2$), the proxies fail to be good approximations). These figures were obtained with Wolfram Mathematica's \texttt{NMaximize}
RandomSearch method with 220 search points.}
    \label{fig:grid_L2}
\end{figure*}

\begin{figure*}[!ht]
    \includegraphics[width=0.8\linewidth]{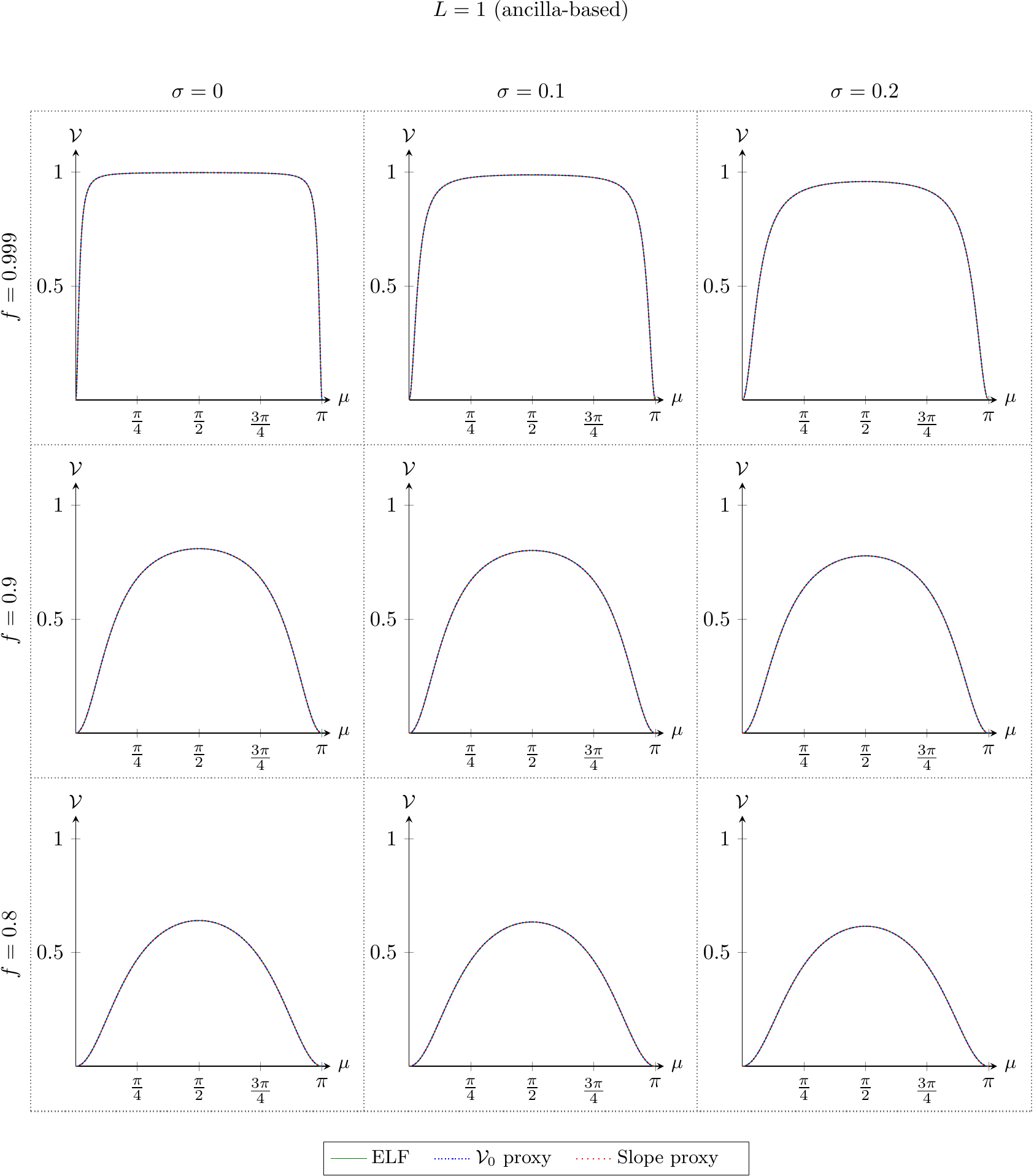}
    \caption{Plots of the variance reduction factor $\mathcal V$ vs the prior mean $\mu$ for $L=1$ for the ancilla-based scheme. The proxies give identical results to the exact optimization of the variance reduction factor. These figures were obtained with Wolfram Mathematica's \texttt{NMaximize}
RandomSearch method with 120 search points.}
    \label{fig:grid_L1_AB}
\end{figure*}

\begin{figure*}[!ht]
    \includegraphics[width=0.8\linewidth]{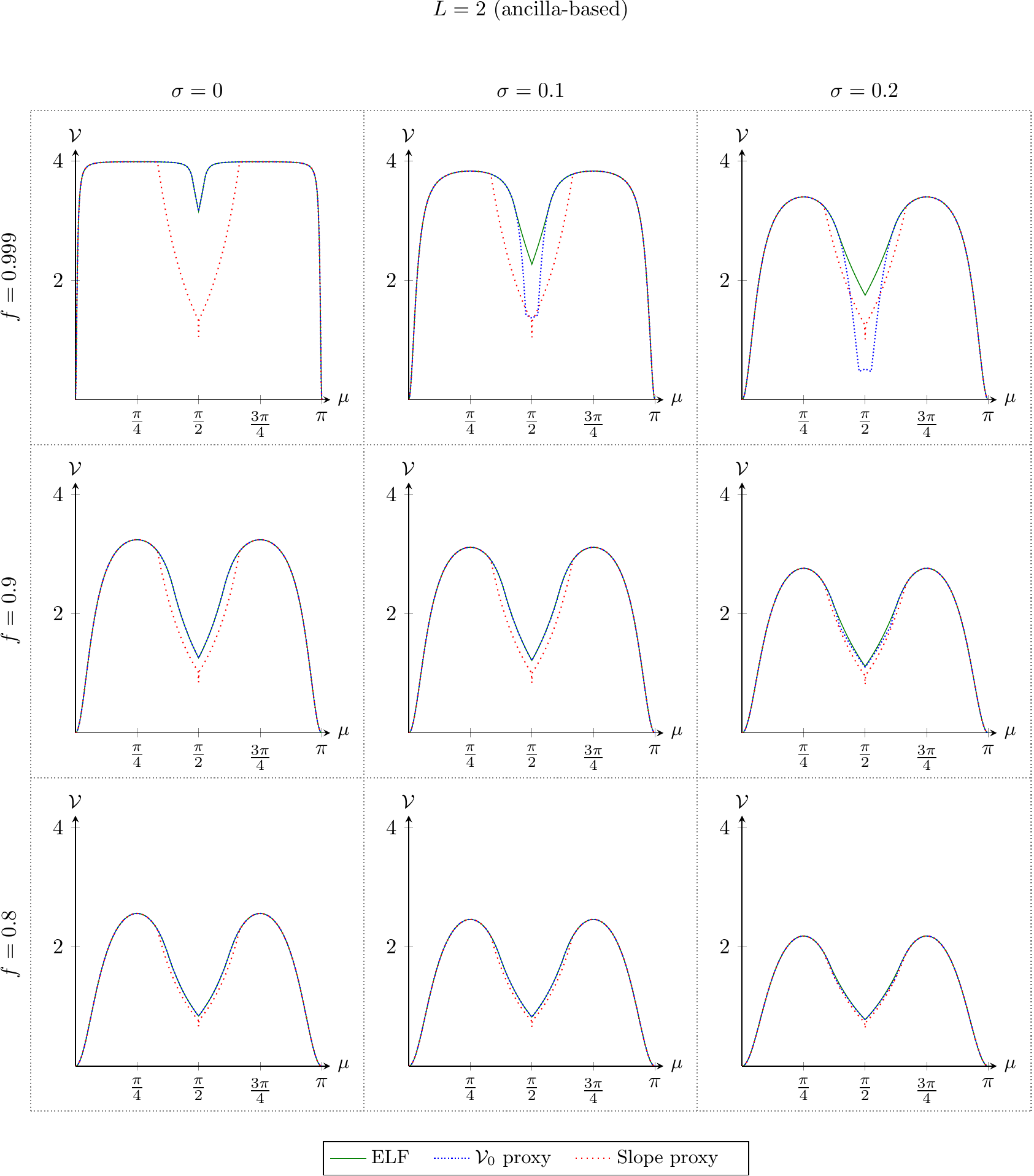}
    \caption{Plots of the variance reduction factor $\mathcal V$ vs the prior mean $\mu$ for $L=2$ for the ancilla-based scheme. The proxies work well when $f$ and $\sigma$ are not too large (say $f\leq 0.9$ and $\sigma \leq 0.1$). When $f$ and $\sigma$ are both large (for example, $f=0.999$ and $\sigma=0.2$), the proxies fail to be good approximations). These figures were obtained with Wolfram Mathematica's \texttt{NMaximize}
RandomSearch method with 120 search points.}
    \label{fig:grid_L2_AB}
\end{figure*}

\subsection{Analytical expressions for \texorpdfstring{$L=1$}{L1} slope proxy}
\label{sec:L1example}

In this appendix, we present an analytical solution to the optimization problem \eqref{eq:optProblemSlope} for the ancilla-free bias when $L=1$. In this case, the bias \eqref{eq:delta_af} can be written as the Fourier series
\begin{align}
    \Delta(\theta;x_1,x_2) = \sum_{l=0}^3 \mu_l(x_1,x_2) \cos(l\theta),
    \label{eq:slopePrime}
\end{align}
where
\begin{align}
    \mu_0(x_1,x_2) &=
    2 \cos(x_1) \sin(x_2) \sin(x_1) \cos(x_2),
    \nonumber\\
    \mu_1(x_1,x_2) &=
    \cos^2(x_1) \cos^2(x_2)
    +
    \cos^2(x_1)\sin^2(x_2)
    \nonumber\\
    &\quad+
    \cos^2(x_1) \sin^2(x_2),
    \nonumber\\
    \mu_2(x_1,x_2) &= -2 \cos(x_1) \cos(x_2) \sin(x_1) \sin(x_2),
    \nonumber\\
    \mu_3(x_1,x_2) &=
    \sin^2(x_1) \sin^2(x_2).
\end{align}

The optimization problem \eqref{eq:optProblemSlope} may be stated as
\begin{align}
&\mbox{Maximize} &|\Delta'(\mu; x_1,x_2)|
\nonumber\\
&\mbox{subject to} & x_1, x_2 \in (-\pi,\pi].
\label{eq:optProblemSlopeL1}
\end{align}
where
\begin{align}
    \Delta'(\mu; x_1,x_2) &= -[\cos^2(x_1) \cos^2(x_2) +
    \cos^2(x_1)\sin^2(x_2) \nonumber \\
& \quad +
    \cos^2(x_1) \sin^2(x_2)]  \sin(\mu)
    \nonumber\\
& \quad + 4 \cos(x_1) \cos(x_2) \sin(x_1) \sin(x_2)
    \sin(2\mu) \nonumber \\
& \quad -3
    \sin^2(x_1) \sin^2(x_2) 
    \sin(3\mu). 
\end{align} 

Solving Eq.~\eqref{eq:optProblemSlopeL1} gives the following solution.

\begin{figure*}[!ht]
\includegraphics[width=0.85\textwidth]{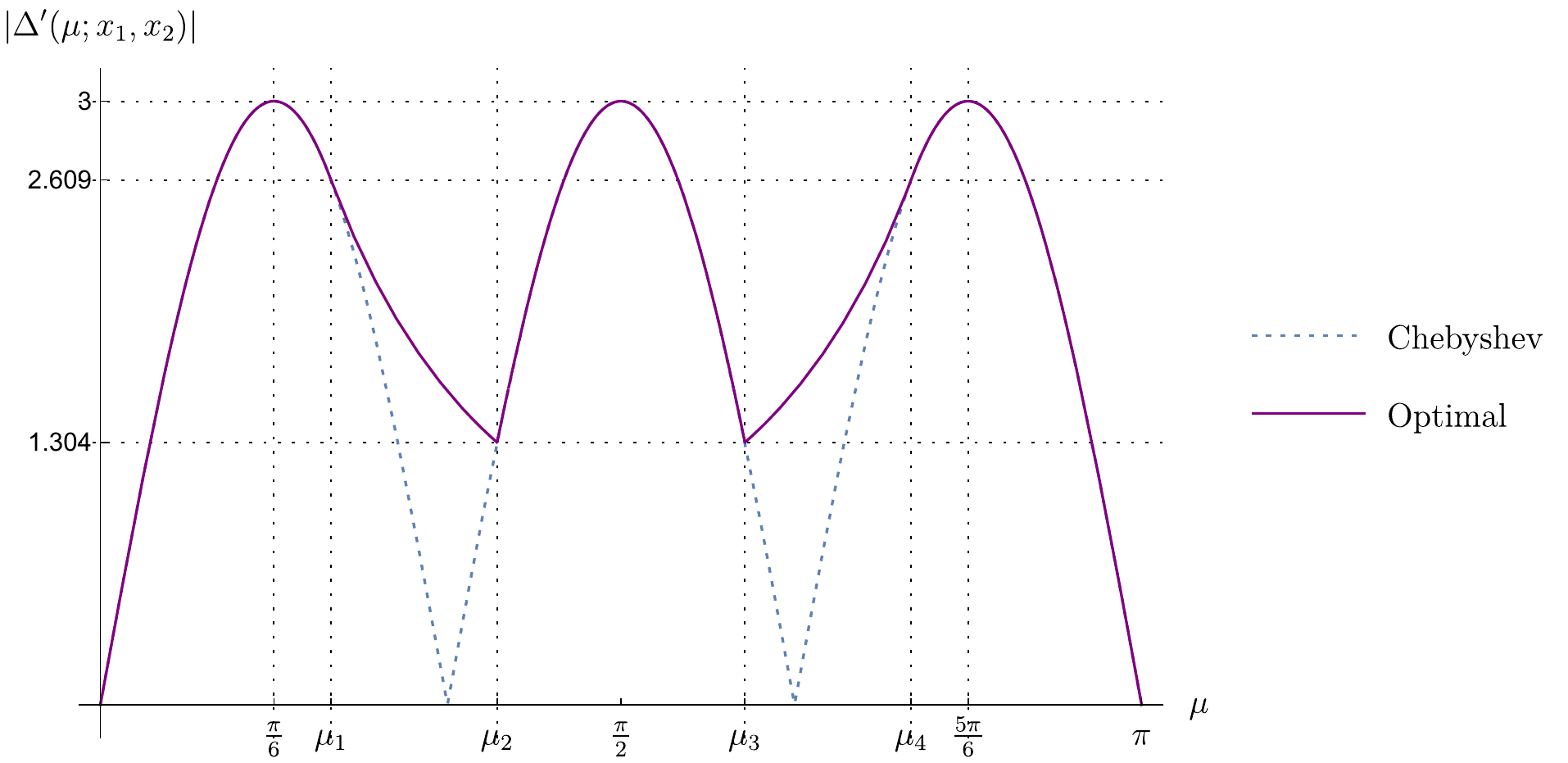}
  \caption{Graph of the optimal slope, with angles $\mu_1$, $\mu_2$, $\mu_3$, and $\mu_4$ given by 
  Eqs.~\eqref{eq:theta1forL1}, \eqref{eq:theta2forL1}, \eqref{eq:theta3forL1}, and \eqref{eq:theta4forL1}, respectively.
  }
  \label{fig:optSlope}
\end{figure*}

\begin{widetext}
\begin{proposition}
The maximum magnitude of the slope of the $L=1$ likelihood function is
\begin{align} 
    & 
    \max_{(x_1,x_2) \in [0,\pi]^2}
    |\Delta'(\mu;x_1,x_2)| = \begin{cases}
        3 \sin(3\mu) , & \mu \in [0,\mu_1]\cup[\mu_4,\pi]
        \\[1.2ex]
        \frac{4 \cos^4 (\mu/2) \cot(\mu/2)}{1+3 \cos(\mu)} , & \mu \in (\mu_1,\mu_2) \\[1.2ex]
        -3 \sin(3\mu) , & \mu \in [\mu_2,\mu_3]
        \\[1.2ex]
        \frac{4 \sin^4 (\mu/2) \tan(\mu/2)}{1-3 \cos(\mu)},  & \mu \in (\mu_3,\mu_4) 
    \end{cases}
    \label{eq:optSlope}
\end{align}
and an example of a pair of angles that achieves this maximum is
\begin{align} 
    (\gamma_1,\gamma_2)\in \argmax_{(x_1,x_2) \in [0,\pi]^2}
    |\Delta'(\mu;x_1,x_2)|,
\end{align}
where for $i=1,2$,
\begin{align}
\gamma_i =
    \begin{cases}
        \frac\pi 2 , & \mu \in [0,\mu_1]\cup[\mu_2,\mu_3]\cup[\mu_4,\pi]
        \\[1.2ex]
        (-1)^i \mathrm{arccot}(\sqrt{1-3\cos(\mu)+\sec(\mu)})
        , & \mu \in (\mu_1,\mu_2) \\[1.2ex] \mathrm{arccot}(\sqrt{1+3\cos(\mu)-\sec(\mu)}), & \mu \in (\mu_3,\mu_4)
    \end{cases}
    \label{eq:optangles}
\end{align}
where $\mu_1,\mu_2,\mu_3,\mu_4$ are given by
\begin{align}
    \mu_1 &= 2 \arctan \left[\sqrt{\tfrac 13 \left(4-\sqrt{13}\right)}\right] \approx 0.6957 
    \label{eq:theta1forL1}
    \\
    \mu_2 &= 4 \arctan\left[\sqrt{\root{3}{p_2}}\right]
    \approx 1.1971 
    \label{eq:theta2forL1}\\
    \mu_3 &= 4 \arctan\left[\sqrt{\root{3}{p_3}}\right] \approx 1.9445 
    \label{eq:theta3forL1}
    \\
    \mu_4 &= 2\arctan\left[ \sqrt{4+\sqrt{13}} \ \right]
    \approx 2.4459
    \label{eq:theta4forL1}
\end{align}
where $p_2$ and $p_3$ are octic polynomials given by
\begin{align}
    p_2(x) &= 1+72x-1540x^2+8568 x^3 - 16506 x^4+8568 x^5 - 1540 x^6 + 72 x^7 + x^8, \\
    p_3(x) &= 9-264x+2492x^2-9016 x^3 +13302 x^4 -9016 x^5 + 2492 x^6 -264 x^7 + 9x^8.
\end{align}
The notation $\root{a}{p}$ refers to the $a$th smallest (with starting index 1) real root of the polynomial $p$.
\end{proposition}

A plot of Eq.~\eqref{eq:optSlope} is shown in Figure \ref{fig:optSlope}.

\end{widetext}

\end{document}